\documentclass[12pt]{article}
\usepackage{amsmath}
\usepackage{graphicx}
\usepackage{enumerate}
\usepackage{natbib}
\usepackage{url} % not crucial - just used below for the URL 

\usepackage{amsfonts}

\usepackage{mathtools}

\usepackage{multirow}
\usepackage{lscape}

\usepackage{soul}

\usepackage{float}
\usepackage{subcaption}  % an alternative package for sub figures
\newcommand{\subfloat}[2][need a sub-caption]{\subcaptionbox{#1}{#2}}

\usepackage{times}

\usepackage{eurosym}

\usepackage{algorithm}
\usepackage[noend]{algpseudocode}

\usepackage{tikz}
\usetikzlibrary{
  decorations.markings, arrows.meta, shapes.geometric
}

\usepackage{caption}

\usepackage{amsthm}
\newtheorem{proposition}{Proposition}

\usepackage[hypertexnames=false]{hyperref}

\newcommand{\blind}{1}

\begin{document}

\def\spacingset#1{\renewcommand{\baselinestretch}%
  {#1}\small\normalsize} \spacingset{1}

%%%%%%%%%%%%%%%%%%%%%%%%%%%%%%%%%%%%%%%%%%%%%%%%%%%%%%%%%%%%%%%%%%%%%%%%%%%%%%

\if1\blind
  {
    \title{\bf A Comprehensive Bayesian Approach to Entity Resolution for Data with Multiple Truths}
    \author{
      Hyungjoon Kim\thanks{Hyungjoon Kim is a Ph.D. Candidate, Department of Statistics, Colorado State University, Fort Collins, CO 80523-1877 (e-mail: hyungjoon.kim@colostate.edu); Andee Kaplan is an Associate Professor, Department of Statistics, Colorado State University, Fort Collins, CO 80523-1877 (e-mail: andee.kaplan@colostate.edu); and Matthew D. Koslovsky is an Assistant Professor, Department of Statistics, Colorado State University, Fort Collins, CO 80523-1877 (e-mail: matt.koslovsky@colostate.edu). This work utilized the Alpine high performance computing resource at the University of Colorado Boulder. Alpine is jointly funded by the University of Colorado Boulder, the University of Colorado Anschutz, and Colorado State University and with support from NSF grants OAC-2201538 and OAC-2322260. Andee Kaplan was partially supported by NSF-2330089 and NSF-2338428, Hyungjoon Kim was partially supported by NSF-2330089, and Matthew D. Koslovsky gratefully acknowledges the support of NSF grant DMS-2245492. The opinions, findings, and conclusions expressed are those of the authors and do not necessarily reflect the views of the NSF.},
      Andee Kaplan,
      and Matthew D. Koslovsky\hspace{.2cm} \\
      Department of Statistics, Colorado State University
    }
    \maketitle
  } \fi

\if0\blind
  {
    \bigskip
    \bigskip
    \bigskip
    \begin{center}
      {\LARGE\bf A Comprehensive Bayesian Approach to Entity Resolution for Data with Multiple Truths}
    \end{center}
    \medskip
  } \fi

\bigskip
\begin{abstract}
  In many applications, from government to ecology, integrating data from diverse and noisy sources is critical for downstream inference.
  However, a unique identifier to link records cleanly from the same entity may not exist. Entity resolution (also referred to as de-duplication or record linkage) merges such databases to identify duplicates, allowing for more complete data for inference.
  A multitude of methods from statistics and computer science in recent years assume a single, immutable true value for each variable used in linkage.
  We argue this restrictive assumption is often violated in practice, leading researchers to discard useful data and biasing downstream inference.
  For example, a respondent's education status may truly change between two different surveys, yet existing methods would treat at least one observation as a distorted version of the truth, even though both are correct.
  In this paper, we propose a novel entity resolution model that comprehensively accommodates ``multiple truths'' by introducing a mixture of exponential family distributions that further handles multiple data types.
  We provide options to fit the model with Markov chain Monte Carlo routines and variational inference for massive datasets.
  We demonstrate the method's value via simulation and linking a longitudinal survey of Italian household wealth. 
\end{abstract}

\noindent%
\vfill

\newpage
\spacingset{1.9} % DON'T change the spacing!
\section{Introduction}
\label{sec:intro}

Entity resolution, or the process of merging potentially noisy databases, is a fundamental step in any data analysis pipeline.
While integrating data is straightforward when each entity has a unique identifier, in practice, multiple datasets often lack a common identifier.
Probabilistic entity resolution, encompassing record linkage and deduplication, enables researchers to merge multiple datasets without a common identifier in the presence of noise.
As the models not only facilitate data integration but also enable downstream data analysis, researchers in various fields, from political science \citep{enamorado2019,goel2020} to biomedical science \citep{hejblum2019,beesley2020,evans2022} to even environmental science \citep{lu2022,drew2025} are employing entity resolution models.

Probabilistic entity resolution models were first formalized by \cite{fellegi1969} and are classified into two categories according to the framework they adopt.
First, the comparison-based framework introduced by \cite{fellegi1969} calculates a matching probability for each record pair based on pairwise comparisons and has subsequently been extended in various directions, including Bayesian implementations, handling multiple files, and ensuring its effective application in downstream analysis \citep{lahiri2005,sadinle2013,murray2015,shan2022,aleshin-guendel2023,slawski2025}.
A competing approach uses a hit-or-miss framework that introduces latent variables and estimates the linkage structure through the relation between the observed records and the latent variables \citep{copas1990,liseo2011,tancredi2011,steorts2015,steorts2016,marchant2021}.
The hit-or-miss framework directly models the true latent records that generate the observed records and determines links via this latent structure that allows records to either equal (hit) or not equal (miss) the true latent structure.

Regardless of the approach taken, probabilistic entity resolution models share the fundamental assumption that each entity has a \textit{single}, \textit{immutable}, true value.
Thus, any observation that does not match this value is considered a distortion from the truth and treated strictly as random noise, meaning the observed erroneous field does not carry any information about the truth.
However, this naive assumption often leads to discrepancies with real-world data.
In many instances, it is inappropriate to treat cases where an observation does not match the latent true value as a random distortion.
For instance, observations that vary locally around the latent truth due to measurement error may inform the truth; recorded observations from the same entity may contain differing true values (e.g., multiple stores selling the same item for different prices); or the observed records' true values may change over time (e.g., price of residence reported in two different years).
In these cases, if an entity resolution model does not account for the fact that the observed values can have multiple truths, it will treat at least one of the observations as a distortion.
This will lead to an overestimation of the distortion rate for that variable and consequently result in failing to link records from the same entity.

In this work, we propose a Comprehensive Hit-or-Miss Probabilistic Entity Resolution (CHOMPER) model to accommodate data with multiple truths.
Our model is based on the hit-or-miss framework.
However, we develop a new hitting mechanism by allowing observations to be considered a ``hit'' when their values are within a local neighborhood of the truth.
This novel \textit{locally-varying hit} mechanism enables the model to consider cases in which true observations differ without considering them as random distortions, preventing overestimation of the distortion rate and allowing better linkage structure estimation.
Additionally, CHOMPER has two supplementary contributions.
First, it allows researchers to model variables as long as they follow an exponential family distribution, whereas existing methods typically only support one type of variable (e.g., categorical or continuous).
Further, most existing Bayesian entity resolution models draw inference from the posterior distribution via Markov chain Monte Carlo (MCMC).
However, the CHOMPER model can adopt other inferential methods, such as variational inference (VI; \cite{blei2017}), which increases scalability to large-scale datasets.

The paper is organized as follows.
First, we conclude this section by introducing the traditional hit-or-miss framework for entity resolution.
Section 2 presents a novel hit-or-miss modeling framework that accounts for locally-varying truths.
In Section 3, we describe estimation and inference for our proposed method, including an evolutionary VI algorithm for entity resolution.
Section 4 provides the performance of the proposed model when applied to real-world data using the Italian Survey on Household Income and Wealth (ISHIW) dataset, and we compare the advantages and disadvantages of MCMC and VI.
In Section 5, we demonstrate the validity of our model through simulation studies.
Section 6 concludes the paper, discussing limitations and future research directions.

\subsection{Traditional Hit-or-Miss Framework}
\label{sec:background}

The traditional hit-or-miss framework models observed data as either an exact match of the true latent value or a distorted version of that latent truth \citep{copas1990}.
It deviates from the comparison-based framework of \cite{fellegi1969}, which models the comparisons between all pairs of observations and uses a likelihood ratio test to estimate each link individually.
In contrast, the hit-or-miss framework estimates the linkage structure by comparing each observation with the corresponding true latent record and linking two records if they have the same latent truth.
This difference gives the hit-or-miss framework the advantage of accommodating potential distortion in observations directly, allowing analysts to incorporate their knowledge of the distortion process through choices in the likelihood and priors in a Bayesian implementation.

The majority of hit-or-miss entity resolution models have been Bayesian extensions of the model presented in \cite{copas1990}, including the work of \cite{liseo2011}, \cite{tancredi2011}, \cite{steorts2015}, \cite{steorts2016}, and \cite{marchant2021}.
These models are defined so that if an observation exactly matches its unique true latent record (hit), the observation has a point mass at the value of the true latent record in the likelihood.
If not (miss), the observation follows the same distribution as the true latent variable.
That is, let $x_{i}$ be an observation, $y_{i}$ a latent variable that generates $x_{i}$, $f_{y_{i}}$ the distribution of $y_{i}$, and $\delta_{y_{i}}$ a Dirac delta function at $y_{i}$.
Then, Bayesian entity resolution models using a hit-or-miss framework generally assume
\begin{equation*}
  x_{i}\mid y_{i}\stackrel{\text{ind}}{\sim}
  \begin{cases}
    \delta_{y_{i}}, & \text{if } x_{i} \text{ is a hit}   \\
    f_{y_{i}},      & \text{if } x_{i} \text{ is a miss.}
  \end{cases}
\end{equation*}
This formulation provides a flexible structure for the model to handle various types of variables, depending on the formulation of $f_{y_{i}}$, including categorical \citep{tancredi2011,steorts2016}, continuous \citep{liseo2011}, and string-valued variables, such as names or addresses \citep{steorts2015}.

\section{Comprehensive Hit-or-Miss Probabilistic Entity Resolution (CHOMPER) Model}
\label{sec:model}

In this section, we introduce the Comprehensive Hit-or-Miss Probabilistic Entity Resolution (CHOMPER) model.
The CHOMPER model generalizes probabilistic entity resolution models that use the hit-or-miss framework from two perspectives.
First, it provides more accurate linkage estimation in many real-world scenarios by accounting for potential local variability of an entity's observed value.
Second, it can handle any type of variable (fields; attributes) that belongs to the exponential family simultaneously.

Before defining the CHOMPER model, we first introduce pertinent notation.
Suppose there are $k$ files, with the $i$th file containing $n_{i}$ records.
If $k=1$, the model is applicable for deduplication, and if $k>1$, it is suitable for record linkage and deduplication simultaneously.
We assume there are $p$ common fields across the $k$ files that are used for the linkage.
We will denote $i=1,\cdots,k$, $j=1,\cdots,n_{i}$, and $\ell=1,\cdots,p$ as the indices of file, record, and variable, respectively, and $N_{\max}=\sum_{i=1}^{k}n_{i}$ as the total number of records across all files.

Let $\boldsymbol{x}_{ij}=\left(x_{ij1},\cdots,x_{ijp}\right)$ represent the $j$th record in the $i$th file, and $\boldsymbol{y}_{\lambda_{ij}}=\left(y_{\lambda_{ij}1},\cdots,y_{\lambda_{ij}p}\right)$ the true latent record, where $\boldsymbol{x}_{ij}$ is linked to $\boldsymbol{y}_{\lambda_{ij}}$ through the latent index $\lambda_{ij}$.
Without loss of generality, suppose that we have two different records, $\boldsymbol{x}_{11}$ and $\boldsymbol{x}_{22}$.
Then, our model considers both records to be linked (or generated from the same entity) if and only if they are linked to the same latent truth, $\boldsymbol{y}_{j'}$, i.e., $\lambda_{11}=\lambda_{22}=j'$.
Note that the traditional hit-or-miss framework introduced in Section~\ref{sec:background} considers an observation to \textit{hit} if $x_{ij\ell}=y_{\lambda_{ij}\ell}$ or \textit{miss} if $x_{ij\ell}$ is a distorted realization of $y_{\lambda_{ij}\ell}$.
Below, we relax this assumption.
Regardless, to distinguish the hit-or-miss status of an observation, we introduce a binary indicator variable $z_{ij\ell}$, where $z_{ij\ell}=0$ indicates a hit and $z_{ij\ell}=1$ a miss.

To allow researchers to flexibly model various types of variables and account for locally-varying truths, we introduce the following likelihood function:
\begin{equation}
  \label{eqn:likelihood-chomper}
  x_{ij\ell}\mid\lambda_{ij},y_{\lambda_{ij}\ell},z_{ij\ell},\boldsymbol{\eta}_{\ell},\boldsymbol{\zeta}_{\ell}\stackrel{\text{ind}}{\sim}
  f_{\ell}\left(x_{ij\ell}\mid\boldsymbol{\eta}_{\ell}\right)^{z_{ij\ell}}f_{\ell}\left(x_{ij\ell}\mid\boldsymbol{\zeta}_{\ell},y_{\lambda_{ij}\ell}\right)^{1-z_{ij\ell}},
\end{equation}
where $f_{\ell}$ is a distribution from the exponential family with natural parameters $\boldsymbol{\eta}_{\ell}$ or $\boldsymbol{\zeta}_{\ell}$.
When an observation is a hit without distortion (i.e., $z_{ij\ell}=0$), $x_{ij\ell}$ follows a distribution parameterized by $\boldsymbol{\zeta}_{\ell}$ and $y_{\lambda_{ij}\ell}$.
Here, $\boldsymbol{\zeta}_{\ell}$ that controls the plausible range of the local variability (i.e., the hitting range) is a hyperparameter specified based on a practitioner's domain knowledge, and $y_{\lambda_{ij}\ell}$ determines the location of the hitting range.
Specific examples for defining $\boldsymbol{\zeta}_{\ell}$ are discussed in the following sections.
However, if $x_{ij\ell}$ is distorted, we consider it a random realization from the latent distribution of $y_{\lambda_{ij}\ell}$ in the absence of a known distortion pattern.
Thus, both distorted $x_{ij\ell}$ and $y_{j'\ell}$ follow $f_{\ell}$ parameterized with $\boldsymbol{\eta}_{\ell}$, and $\boldsymbol{\eta}_{\ell}$ should have its own prior distribution corresponding to $f_{\ell}$.
For instance, we may assume $f_{\ell}$ is Gaussian with $\boldsymbol{\eta}_{\ell}=\left(\frac{\mu_{\ell}}{\sigma_{\ell}},-\frac{1}{2\sigma_{\ell}}\right)$ for the $\ell$th continuous variable, where $\mu_{\ell}$ and $\sigma_{\ell}$ follow a flat and inverse gamma prior, respectively.

Next, we define the prior distribution of $\lambda_{ij}$, which links records and is the main parameter of interest.
Note that $\lambda_{ij}$ is an index variable that links an observed record $\boldsymbol{x}_{ij}$ to a latent truth $\boldsymbol{y}_{\lambda_{ij}}$.
So, it can be any integer $j'\in\{1,\cdots,N_{\max}\}$, and the linkage structure is encoded in $k$ $n_{i}$-length vectors $\boldsymbol{\lambda}_{i}=\left(\lambda_{i1},\cdots,\lambda_{in_{i}}\right)$.
Reflecting a lack of prior knowledge, we assign a diffuse prior distribution to $\boldsymbol{\Lambda}=\left(\boldsymbol{\lambda}_{1},\cdots,\boldsymbol{\lambda}_{k}\right)$.
Other priors could also be chosen for the linkage structure, such as a microclustering prior \citep{betancourt2022a,betancourt2022,lee2022}.

Lastly, a Bernoulli distribution with probability of success $\beta_{\ell}$ is chosen for the prior distribution of the distortion indicator $z_{ij\ell}$, where $\beta_{\ell}\sim\text{Beta}\left(a_{\ell},b_{\ell}\right)$.
This allows researchers to encode domain knowledge about the distortion rate of the $\ell$th field by controlling the parameters $a_{\ell}$ and $b_{\ell}$.
For example, suppose an $\ell$th field has a low distortion rate, one can set $a_{\ell}$ and $b_{\ell}$ such that $\mathbb{E}\left[\beta_{\ell}\right]$ is close to 0.

Together, we formulate the CHOMPER model as follows:
\begin{equation}
  \label{eqn:chomper-mdl}
  \begin{aligned}
    x_{ij\ell}\mid\lambda_{ij},y_{\lambda_{ij}\ell},z_{ij\ell},\boldsymbol{\eta}_{\ell},\boldsymbol{\zeta}_{\ell}
                                           & \stackrel{\text{ind}}{\sim}f_{\ell}\left(x_{ij\ell}\mid\boldsymbol{\eta}_{\ell}\right)^{z_{ij\ell}}f_{\ell}\left(x_{ij\ell}\mid\boldsymbol{\zeta}_{\ell},y_{\lambda_{ij}\ell}\right)^{1-z_{ij\ell}} \\
    z_{ij\ell}\mid\beta_{\ell}             & \stackrel{\text{ind}}{\sim}\text{Bernoulli}\left(\beta_{\ell}\right)                                                                                                                                \\
    \beta_{\ell}                           & \stackrel{\text{ind}}{\sim}\text{Beta}\left(a_{\ell},b_{\ell}\right)                                                                                                                                \\
    y_{j'\ell}\mid\boldsymbol{\eta}_{\ell} & \stackrel{\text{ind}}{\sim}f_{\ell}\left(y_{j'\ell}\mid\boldsymbol{\eta}_{\ell}\right)                                                                                                              \\
    \boldsymbol{\eta}_{\ell}               & \stackrel{\text{ind}}{\sim}\pi\left(\boldsymbol{\eta}_{\ell}\right)                                                                                                                                 \\
    \pi\left(\boldsymbol{\Lambda}\right)   & \propto 1.
  \end{aligned}
\end{equation}
In summary, the principal contribution of the CHOMPER model lies in the flexible structure obtained by using an exponential family distribution for the likelihood of $x_{ij\ell}$ when it is a hit instead of a point mass.
Further, the parameter $\boldsymbol{\zeta}_{\ell}$ allows researchers to adjust the hitting range, thereby enabling the use of variables whose observations may have multiple truths in entity resolution.

\subsection{Connection to the Traditional Hit-or-Miss Framework}
\label{sec:connection-traditional-hit-or-miss}

In this section, we show how the CHOMPER model encompasses the traditional hit-or-miss framework, which considers an observation a hit only if it matches a single truth, as a special case.
Note that we can express the likelihood of $x_{ij\ell}$ when $z_{ij\ell}=0$ as follows:
\begin{equation*}
  f_{\ell}\left(x_{ij\ell}\mid\lambda_{ij},y_{\lambda_{ij}\ell},\boldsymbol{\zeta}_{\ell},z_{ij\ell}=0\right)
  =h\left(x_{ij\ell}\right)\exp\left(\boldsymbol{\zeta}_{\ell}x_{ij\ell}-A\left(\boldsymbol{\zeta}_{\ell}\right)\right),
\end{equation*}
where $h\left(x_{ij\ell}\right)$ is the base measure and $A\left(\boldsymbol{\zeta}_{\ell}\right)$ is the log-partition function.
Then, as established in Proposition~\ref{prop:degeneracy}, proven in Appendix A, we obtain a degenerate distribution at $y_{\lambda_{ij}\ell}$ by formulating $\boldsymbol{\zeta}_{\ell}$ to satisfy $A'\left(\boldsymbol{\zeta}_{\ell}\right)=y_{\lambda_{ij}\ell}$ and $A''\left(\boldsymbol{\zeta}_{\ell}\right)\rightarrow0$ for first and second order derivatives, $A'\left(\boldsymbol{\zeta}_{\ell}\right)$ and $A''\left(\boldsymbol{\zeta}_{\ell}\right)$, respectively.
Therefore, the model considers $x_{ij\ell}$ as a hit only if it matches $y_{\lambda_{ij}\ell}$, equivalent to the traditional hit-or-miss framework.
\begin{proposition}
  \label{prop:degeneracy}
  Let \(X\) be a random variable following an exponential family distribution \(f\) with a natural parameter \(\boldsymbol{\zeta}\) such that \(f(x\mid\boldsymbol{\zeta}) = h(x)\exp\left(\boldsymbol{\zeta}T\left(x\right)-A(\boldsymbol{\zeta})\right)\) for the base measure \(h(x)\), the log-partition function \(A\left(\boldsymbol{\zeta}\right)\), and a sufficient statistic \(T\left(X\right)\).
  Then, \(T\left(X\right)\stackrel{p}{\rightarrow} \mathbb{E}\left[T\left(X\right)\right]=A'\left(\boldsymbol{\zeta}\right)\) as \(A''\left(\boldsymbol{\zeta}\right)\rightarrow0\), creating a degenerate distribution concentrated at \(A'\left(\boldsymbol{\zeta}\right)\).
\end{proposition}

\subsection{Example Distributions for Locally-Varying Hits}
\label{sec:comprehensive-hit}

In practice, it is common for an attribute of the same entity to have different values that all are truths without random distortion, which motivates the use of a \textit{locally-varying hit}.
This section discusses representative examples of variables that may admit multiple truths, recommended distributions for modeling them, and the corresponding strategy for specifying $\boldsymbol{\zeta}_{\ell}$.

Suppose researchers conducted two surveys over the course of two years, recording the highest educational level of each participant in four categories (A: High school graduate or less, B: Bachelor, C: Master, D: Ph.D.).
Participants' highest educational level can change even in the absence of measurement error.
They may finish their undergraduate coursework to earn a bachelor's degree (moved from category A to B) or earn a master's degree in two years (moved from category B to C).
Let $\boldsymbol{x}_{11}$ and $\boldsymbol{x}_{22}$ be responses from the same participant without measurement error (i.e., $x_{11\ell}=\text{A}$, $x_{22\ell}=\text{B}$, and $\lambda_{11}=\lambda_{22}=j'$).
Since it is possible for the same participant to have either A or B as their true response, the likelihood function should assign large probabilities to $x_{11\ell}$ and $x_{22\ell}$ when $y_{j'\ell}$ is either A or B.
This has been structurally impossible in previous hit-or-miss entity resolution models because the likelihood function assigns a positive probability only to the immutable single truth.
However, the proposed CHOMPER model makes it possible.
Since the highest educational level is a categorical variable, we assume a multinomial distribution for $f_{\ell}\left(x_{ij\ell}\mid\phi_{\ell},\tau_{\ell},y_{\lambda_{ij}\ell}\right)$ with $M_{\ell}$ levels in the $\ell$th field:
\begin{equation}
  \label{eqn:softmax-chomper}
  f_{\ell}\left(x_{ij\ell}\mid\phi_{\ell},\tau_{\ell},y_{\lambda_{ij}\ell}\right)
  =\frac{\phi_{\ell}^{\frac{1}{\tau_{\ell}}\mathbb{I}\left(x_{ij\ell}\in\mathcal{N}_{\varepsilon_{\ell}}\left(y_{\lambda_{ij}\ell}\right)\right)}}
  {\sum_{m=1}^{M_{\ell}}\phi_{\ell}^{\frac{1}{\tau_{\ell}}\mathbb{I}\left(m\in\mathcal{N}_{\varepsilon_{\ell}}\left(y_{\lambda_{ij}\ell}\right)\right)}},
\end{equation}
where $\mathcal{N}_{\varepsilon_{\ell}}\left(y_{\lambda_{ij}\ell}\right)$ is a closed $\varepsilon_{\ell}$-neighborhood of $y_{\lambda_{ij}\ell}$ such that $\mathcal{N}_{\varepsilon_{\ell}}\left(y_{\lambda_{ij}\ell}\right)=\{x\vert d\left(y_{\lambda_{ij}\ell},x\right)\leq\varepsilon_{\ell}\}$ for a distance function $d\left(y_{\lambda_{ij}\ell},x\right)$ with hyperparameters $\phi_{\ell}$ and $\tau_{\ell}$, which controls the temperature.
Together, $\varepsilon_{\ell}\geq0$, $\phi_{\ell}>1$, and $\tau_{\ell}>0$ define the hitting range and the intensity of the locally-varying hit for the $\ell$th field.
We use a graph shortest path distance for $d\left(y_{\lambda_{ij}\ell},x\right)$.
Namely, let $\mathcal{G}=\left(\mathcal{V},\mathcal{E}\right)$ be a graph where $\mathcal{V}$ is the set of all levels of a categorical variable and $\mathcal{E}$ is the set of all edges between the levels.
Then, we can define $d\left(y_{\lambda_{ij}\ell},x\right)=\min_{p\in\mathcal{P}\left(y_{\lambda_{ij}\ell},x\right)}\vert p\vert$, where $\mathcal{P}\left(y_{\lambda_{ij}\ell},x\right)$ is the set of all paths from $y_{\lambda_{ij}\ell}$ to $x$ in $\mathcal{G}$ and $\vert p\vert$ is the number of edges in path $p$.
For a nominal variable, we use a complete undirected graph $\mathcal{G}$; for an ordinal variable, we may use a path graph $\mathcal{G}$ with natural level ordering, yielding $d\left(y_{\lambda_{ij}\ell},x\right)=\vert y_{\lambda_{ij}\ell}-x\vert$.
Since educational level is ordinal and rapid changes are unexpected, a given category and its adjacent levels are considered as the hitting range with $\mathcal{E}=\left\{\left(\text{A},\text{B}\right),\left(\text{B},\text{C}\right),\left(\text{C},\text{D}\right)\right\}$.
Here, the model is able to assign a relatively large probability to $x_{ij\ell}$ when $x_{ij\ell}$ is $\varepsilon_{\ell}$-adjacent to $y_{\lambda_{ij}\ell}$, despite $x_{ij\ell}$ not being equal to $y_{\lambda_{ij}\ell}$.
Figure~\ref{fig:comprehensive-hit-multinomial} illustrates this example of the multinomial distribution parameterized in Equation~\eqref{eqn:softmax-chomper}.

\begin{figure}[H]
  \begin{center}
    \subfloat[$f_{\ell}$ when the model does not account local variability of a true value $y_{\lambda_{ij}\ell}$ with $\varepsilon_{\ell}=0$.\label{fig:comprehensive-hit-multinomial-1}]{\includegraphics[width=0.9\linewidth]{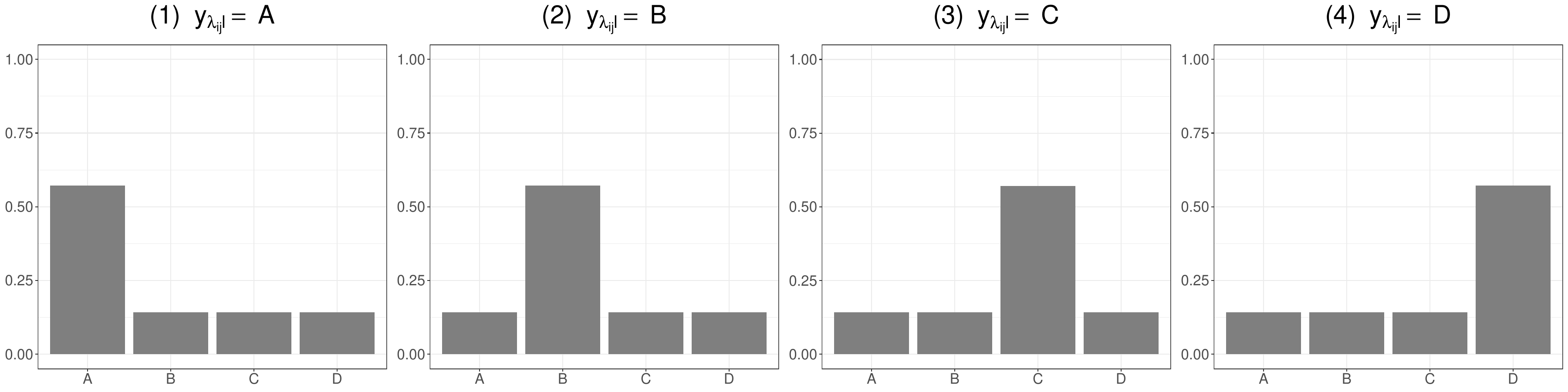} }
    \subfloat[$f_{\ell}$ when the model accounts local variability of a true value $y_{\lambda_{ij}\ell}$ with $\varepsilon_{\ell}=1$.\label{fig:comprehensive-hit-multinomial-2}]{\includegraphics[width=0.9\linewidth]{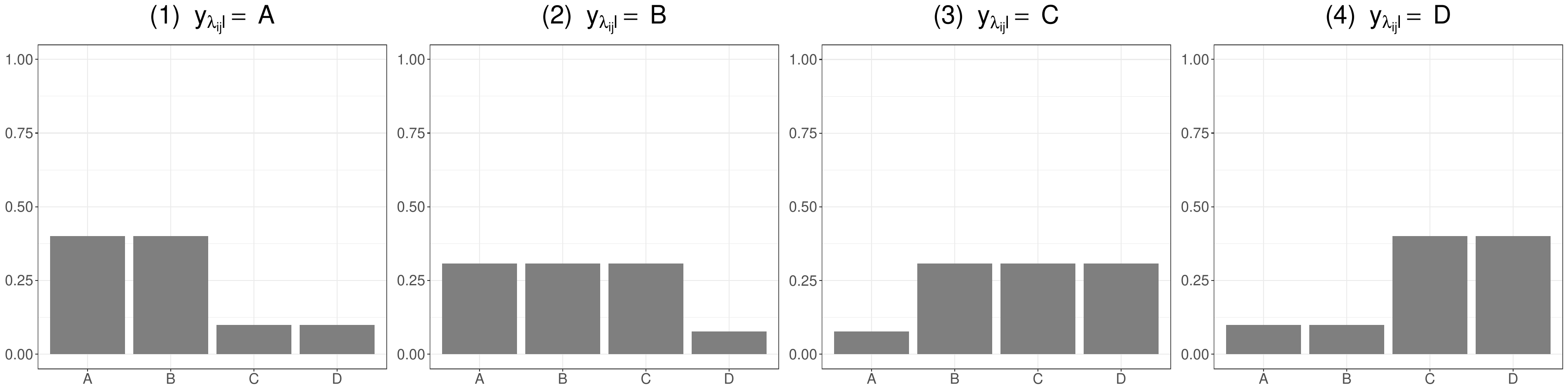} }
  \end{center}
  \caption{Probabilities of multinomial distributions according to the hitting range $\varepsilon_{\ell}$ with $\phi_{\ell}=2$ and $\tau_{\ell}=0.5$. (a) $\varepsilon_{\ell}=0$ concentrates the hitting range to a single true value by assigning a large probability to $y_{\lambda_{ij}\ell}$, while (b) $\varepsilon_{\ell}=1$ allows the observation to hit the truth that is $\varepsilon_{\ell}$-adjacent to $y_{\lambda_{ij}\ell}$ by assigning larger probabilities to these values.}
  \label{fig:comprehensive-hit-multinomial}
\end{figure}

When $\varepsilon_{\ell}=0$, the model approximates the traditional hit-or-miss assumption that an observation hits the true latent value only when $x_{ij\ell}$ exactly matches $y_{\lambda_{ij}\ell}$.
Thus, the likelihood function has a large probability when $x_{ij\ell}=y_{\lambda_{ij}\ell}$ for $z_{ij\ell}=0$ and a relatively small probability otherwise.
However, when $\varepsilon_{\ell}=1$, the likelihood function allows the observation to hit multiple values in the neighborhood of the truth without considering it random distortion by assigning large probabilities to the values where $x_{ij\ell}$ is $\varepsilon_{\ell}$-adjacent to $y_{\lambda_{ij}\ell}$.

Suppose that the researchers recorded net income, a continuous variable, in addition to educational level.
Given that net income can also vary due to changes in tax policy, it is reasonable to assume that observed values are hits if they exist within a local neighborhood of the truth.
Consequently, a locally-varying hit is justified, and the variance of $f_{\ell}\left(x_{ij\ell}\mid\boldsymbol{\zeta}_{\ell},y_{\lambda_{ij}\ell}\right)$ should be chosen to establish an appropriate hitting range.

For continuous variables, we can assume a Gaussian distribution for $f_{\ell}\left(x_{ij\ell}\mid\boldsymbol{\zeta}_{\ell},y_{\lambda_{ij}\ell}\right)$.
Setting the center of the likelihood for a locally-varying hit as $y_{\lambda_{ij}\ell}$, the natural parameter $\boldsymbol{\zeta}_{\ell}$ can be defined as follows:
\begin{equation}
  \label{eqn:gaussian-chomper}
  \boldsymbol{\zeta}_{\ell}=\left(\frac{y_{\lambda_{ij}\ell}}{\varepsilon_{\ell}},-\frac{1}{2\varepsilon_{\ell}}\right).
\end{equation}
That is, by using a large hyperparameter $\varepsilon_{\ell}$, we can widen the hitting range of $x_{ij\ell}$ and allow the model to comprehensively hit values in the neighborhood of $y_{\lambda_{ij}\ell}$.
Figure~\ref{fig:comprehensive-hit-gaussian} shows the likelihood function when $\mathbb{E}\left[x_{ij\ell}\mid z_{ij\ell}=1\right]=0$, $\text{Var}\left[x_{ij\ell}\mid z_{ij\ell}=1\right]=1$, and $y_{\lambda_{ij}\ell}=1$ with different $\varepsilon_{\ell}$'s.
From this, we observe that the hitting range of $x_{ij\ell}$ expands with an increase in $\varepsilon_{\ell}$ while choosing a smaller $\varepsilon_{\ell}$ leads $x_{ij\ell}$ to only be considered a hit when it is approximately equivalent to $y_{\lambda_{ij}\ell}$.

\begin{figure}[H]
  \begin{center}
    \includegraphics[width=0.9\linewidth]{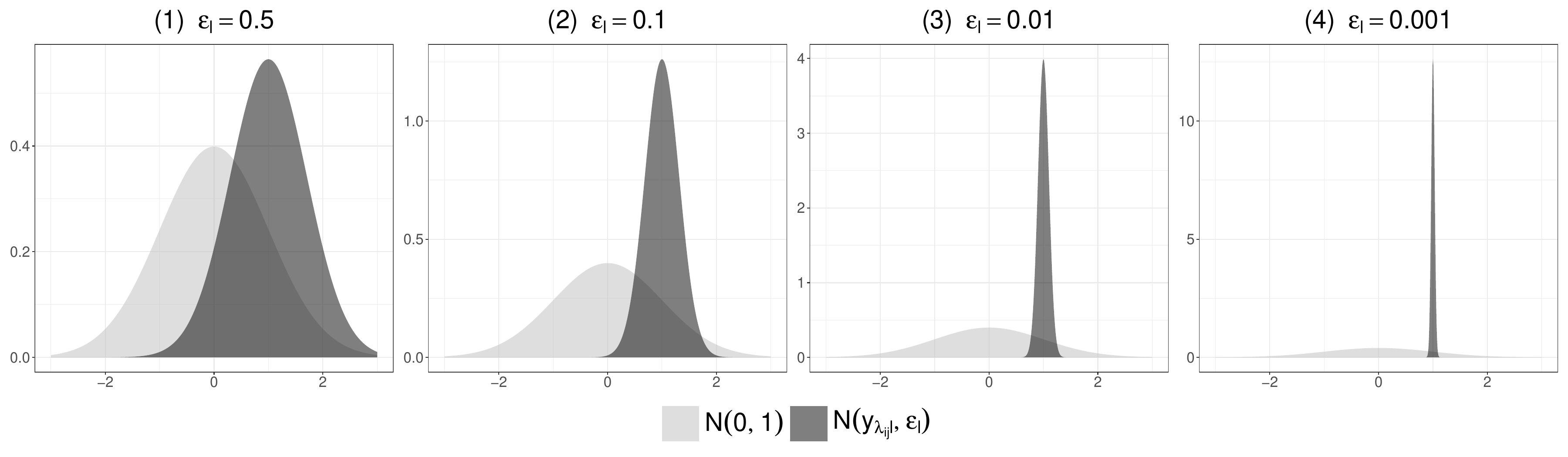}
  \end{center}
  \caption{Likelihood function with $f_{\ell}\sim\text{Gaussian}$ and $y_{\lambda_{ij}\ell}=1$ for different $\varepsilon_{\ell}$'s when the mean and variance of $x_{ij\ell}|z_{ij\ell}=1$ are $0$ and $1$, respectively.}
  \label{fig:comprehensive-hit-gaussian}
\end{figure}

\section{Inference}
\label{sec:inference}

For posterior inference, we propose two approaches: a Metropolis-Hastings within Gibbs MCMC algorithm containing a split-merge step and an evolutionary VI algorithm.
MCMC is relatively straightforward to implement when the full conditional distributions for parameters are available, and it yields asymptotically exact samples from the posterior.
However, MCMC can require substantial computation time due to slow mixing and convergence when the parameter space is large and complex.
To mitigate the slow mixing issue, we adopt a hybrid MCMC algorithm that applies a split-merge update \citep{jain2004} for the linkage parameter inside the Gibbs sampler.
VI commonly offers a shorter runtime than MCMC, but the variational approximation of the true posterior underestimates parameter uncertainty.
In addition, VI is sensitive to initialization and may converge to local optima.
To address these limitations, we develop an evolutionary VI algorithm \citep{whitley2001,givens2012}.

\subsection{Posterior Inference via Markov chain Monte Carlo}
\label{sec:mcmc}

Because the CHOMPER likelihood is a mixture of two distributions from the exponential family, it is possible to compute explicit full conditional distributions for all parameters of interest with conjugate priors for $\boldsymbol{\eta}_{\ell}$.
Nevertheless, the parameters of the posterior distribution of each $\lambda_{ij}$ are probability vectors for $N_{\max}$ categories, so the resulting computational cost scales as $\mathcal{O}\left(N_{\max}^{2}\right)$.
Additionally, the high-dimensional discrete parameter space is susceptible to poor mixing due to chains getting stuck in local modes.
Consequently, a Gibbs update may be ineffective in this setting.
Instead, we implement a split-merge algorithm that jointly updates $\lambda_{ij}$ and the corresponding $z_{ij\ell}$ and $\boldsymbol{y}_{\lambda_{ij}}$ in a single Metropolis-Hastings step (Appendix C).
This approach reduces the time complexity of the sampler and helps mixing by forcibly changing the state of records, either linked or not, further speeding up convergence \citep{jain2004}.
The full conditional distributions for all parameters and their derivations for CHOMPER are presented in Appendix B.

\subsection{Evolutionary Variational Inference}
\label{sec:evil}

In order to scale CHOMPER to large datasets commonly found in practice, we introduce a variational alternative that seeks a function $q^{\star}$ that approximates the posterior distribution.
Specifically, denoting $\boldsymbol{\xi}$ as the collection of all parameters $\{\boldsymbol{\Lambda},\boldsymbol{y},\boldsymbol{z},\boldsymbol{\beta},\boldsymbol{\eta}\}$ and $\boldsymbol{x}$ as the observed data, our goal is to solve the optimization problem:
\begin{equation}
  \label{eqn:vi-def}
  q^{\star}\left(\boldsymbol{\xi}\right)=\operatorname*{argmin}_{q\left(\boldsymbol{\xi}\right)}\mathtt{KL}\left(q\left(\boldsymbol{\xi}\right)\Vert\pi\left(\boldsymbol{\xi}\mid\boldsymbol{x}\right)\right),
\end{equation}
where $\pi\left(\boldsymbol{\xi}\mid\boldsymbol{x}\right)$ is the posterior distribution and $\mathtt{KL}\left(q\Vert\pi\right)$ is the Kullback-Leibler divergence between $q$ and $\pi$.
Searching the entire function space of $q$ is impractical, so it is recommended to traverse a prespecified smaller family of functions $\mathcal{Q}$.
We use a mean-field variational family, where the members of $\mathcal{Q}$ are mutually independent densities of each parameter \citep{bishop2006}.
Denoting each member (variational factor) as $q_{\xi_{i}}\left(\xi_{i}\right)$ for the parameter $\xi_{i}$, $q\left(\boldsymbol{\xi}\right)$ has the form
\begin{equation*}
  q\left(\boldsymbol{\xi}\right)
  =\prod_{i,j}q_{\lambda_{ij}}\left(\lambda_{ij}\right)\prod_{j',l}q_{y_{j'\ell}}\left(y_{j'\ell}\right)\prod_{i,j,l}q_{z_{ij\ell}}\left(z_{ij\ell}\right)
  \prod_{\ell}q_{\beta_{\ell}}\left(\beta_{\ell}\right)\prod_{\ell}q_{\boldsymbol{\eta}_{\ell}}\left(\boldsymbol{\eta}_{\ell}\right).
\end{equation*}

With the mean-field variational family, Coordinate Ascent VI (CAVI) is frequently used to find the optimal variational factor $q^{\star}$ in Equation~\eqref{eqn:vi-def}.
CAVI is an iterative deterministic optimization algorithm that updates the variational factors in turn until the evidence lower bound (ELBO) converges, based on the fact that the optimal factor is proportional to the exponential of the expected value of the log of the full conditional distribution \citep{bishop2006,blei2017}.
However, CAVI is sensitive to initialization and update order, and it may fall into a local optimum if the posterior distribution is multimodal \citep{hoffman2015}.
In order to help alleviate these issues, we propose a special initialization strategy for entity resolution (Appendix D.3) and an Evolutionary Variational Inference for Record Linkage (EVIL) algorithm.

Evolutionary optimization algorithms use a genetic heuristic to optimize the objective function, including crossover, selection, and mutation steps \citep{whitley2001,givens2012}.
\cite{mcnicholas2021} demonstrated that evolutionary algorithms can yield better performance in cluster analysis.
As CHOMPER performs entity resolution by clustering records based on latent true records, we expect EVIL to increase the chances of reaching the global optimum \citep{brusa2024} compared to conventional VI and ultimately improve linkage estimation.

The optimization begins by creating a generation of members with different initial values, $\boldsymbol{\Xi}^{\left(0\right)}=\{\boldsymbol{\xi}^{\left(0\right)}_{p}\}_{p=1}^{N_{P}}$, where $\boldsymbol{\xi}^{\left(0\right)}_{p}=\{\boldsymbol{\Lambda}^{\left(0\right)},\boldsymbol{y}^{\left(0\right)},\boldsymbol{z}^{\left(0\right)},\boldsymbol{\beta}^{\left(0\right)},\boldsymbol{\eta}^{\left(0\right)}\}_{p}$ and $N_{P}$ is the number of (parent) members.
In each generation, the algorithm proceeds as follows and terminates when the improvement in the ELBO is not significant.
\begin{enumerate}
  \item \textbf{CAVI Update}: The initial variational parameters of the $g$th generation, $\boldsymbol{\Xi}^{\left(0,g\right)}$, are updated via CAVI to obtain a new set of parameters, $\boldsymbol{\Xi}^{\left(1,g\right)}$.
        Since each member is independent, this step can be performed in parallel.
  \item \textbf{Crossover with Reconciliation}: Generate $N_{O}$ offspring by adopting single-point crossover to randomly selected pairs from $\boldsymbol{\Xi}^{\left(1,g\right)}$.
        Note that swapping the linkage index may break the consistency between the linkage structure and variational parameters.
  \item \textbf{Reconciliation}: Apply a \textit{reconciliation process} to correct the broken consistency between swapped linkage indices and variational factors.
        For example, if a crossover changes the linkage indices from $\lambda_{1}=\lambda_{2}=1$ and $\lambda_{3}=\lambda_{4}=2$ to $\lambda_{2}=1$ and $\lambda_{1}=\lambda_{3}=\lambda_{4}=2$, then $\boldsymbol{y}_{1},\boldsymbol{y}_{2},\boldsymbol{z}_{1},\boldsymbol{z}_{2},\boldsymbol{z}_{3},\boldsymbol{z}_{4}$, and corresponding variational parameters must be updated accordingly.
        Specifically, $\boldsymbol{y}_{1}$ and $\boldsymbol{y}_{2}$ must be updated to reflect the information of the newly linked records, and the distortion indicators must also be recalculated accordingly (See Figure~\ref{fig:evil-reconciliation}).
        Update the $N_{O}$ reconciled offspring via CAVI, obtaining $\boldsymbol{\Xi}^{\left(2,g\right)}$.
  \item \textbf{Selection}: Construct a new set of parents by selecting the top $N_{P}$ members from $\boldsymbol{\Xi}^{\left(1,g\right)}\cup\boldsymbol{\Xi}^{\left(2,g\right)}$ based on the ELBO.
  \item \textbf{Mutation with Reconciliation}: Mutate the set of variational factors based on a split-merge process \citep{jain2004} to improve the diversity of the population.
        If the new generation does not meet the convergence criterion, \textit{reconcile} the parameters and denote it as $\boldsymbol{\Xi}^{\left(4,g\right)}$.
        Reset the initial values to the mutated generation, such that $\boldsymbol{\Xi}^{\left(0,g+1\right)}=\boldsymbol{\Xi}^{\left(4,g\right)}$.
\end{enumerate}

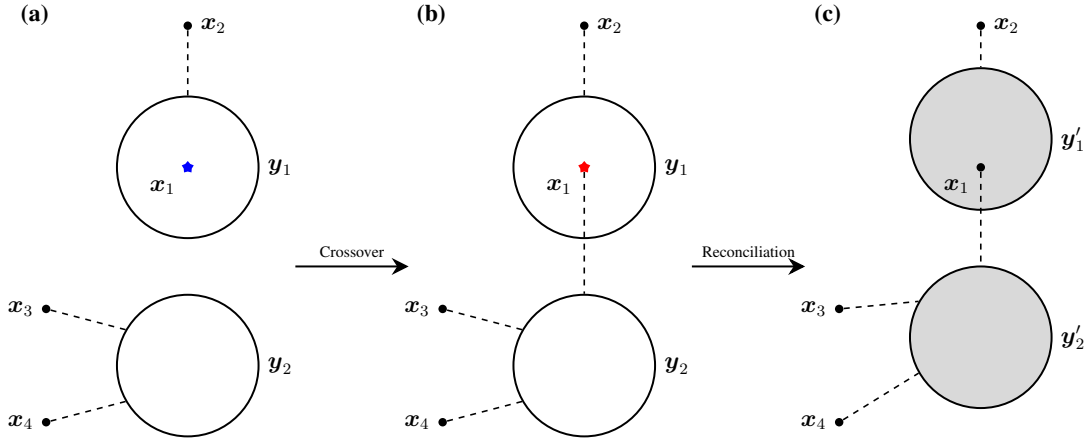
\begin{figure}[H]
  \begin{center}
    \resizebox{0.9\textwidth}{!}{
      \begin{tikzpicture}
        \tikzset{
        % dots (x1, x2, ...)
        dot/.style={circle, fill=black, inner sep=1.5pt},
        % empty circles (y1, y2)
        circ/.style={circle, draw, minimum size=2.5cm, line width=1pt},
        % filled circles (y'1, y'2)
        circ_filled/.style={circ, fill=gray!30},
        % dashed lines (linkage)
        link/.style={dashed, line width=0.8pt},
        % new arrow style
        myarrow/.style={
        -{Stealth[length=3mm, width=3mm]}, % arrow head size
        line width=1pt, % arrow line width
        draw=black, 
        postaction={decorate}, 
        decoration={
            markings,
            mark=at position 0.5 with {\node[above, font=\scriptsize, text centered] {#1};}
          }
        }
        }
        
        % --- (a) ---
        \begin{scope}
          \node at (-2.7, 2.7) {\textbf{(a)}};
          
          % y1
          \node (y1a) [circ, label=right:$\boldsymbol{y}_{1}$] at (0, 0) {};
          \node (x1a) [star, fill=blue, inner sep=1.5pt, label=below left:$\boldsymbol{x}_{1}$] at (0, 0) {};
          \node (x2a) [dot, label=right:$\boldsymbol{x}_{2}$] at (0, 2.5) {};
          \draw [link] (x2a) -- (y1a.north); 
          
          % y2
          \node (y2a) [circ, label=right:$\boldsymbol{y}_{2}$] at (0, -3.5) {}; 
          \node (x3a) [dot, label=left:$\boldsymbol{x}_{3}$] at (-2.5, -2.5) {}; 
          \node (x4a) [dot, label=left:$\boldsymbol{x}_{4}$] at (-2.5, -4.5) {}; 
          \draw [link] (x3a) -- (y2a.150); 
          \draw [link] (x4a) -- (y2a.210);
        \end{scope}
        
        % --- (a) => (b) arrow (with text) ---
        \draw[myarrow={Crossover}] (1.9, -1.75) -- (3.9, -1.75); 
        
        % --- (b) ---
        \begin{scope}[xshift=7cm]
          \node at (-2.7, 2.7) {\textbf{(b)}};
          
          % y1
          \node (y1b) [circ, label=right:$\boldsymbol{y}_{1}$] at (0, 0) {};
          \node (x1b) [star, fill=red, inner sep=1.5pt, label=below left:$\boldsymbol{x}_{1}$] at (0, 0) {};
          \node (x2b) [dot, label=right:$\boldsymbol{x}_{2}$] at (0, 2.5) {};
          \draw [link] (x2b) -- (y1b.north);
          
          % y2
          \node (y2b) [circ, label=right:$\boldsymbol{y}_{2}$] at (0, -3.5) {};
          \node (x3b) [dot, label=left:$\boldsymbol{x}_{3}$] at (-2.5, -2.5) {};
          \node (x4b) [dot, label=left:$\boldsymbol{x}_{4}$] at (-2.5, -4.5) {};
          \draw [link] (x3b) -- (y2b.150);
          \draw [link] (x4b) -- (y2b.210);
          
          % new linkage in (b): (x1 and y2)
          \draw [link] (x1b) -- (y2b.north); 
        \end{scope}
        
        % --- (b) => (c) arrow (with text) ---
        \draw[myarrow={Reconciliation}] (8.9, -1.75) -- (10.9, -1.75);
        
        % --- (c) ---
        \begin{scope}[xshift=14cm]
          \node at (-2.7, 2.7) {\textbf{(c)}};
          
          % y'1 - filled circle; updated y1
          \node (y1c) [circ_filled, label=right:$\boldsymbol{y}'_{1}$] at (0, 0.5) {}; 
          \node (x1c) [dot, label=below left:$\boldsymbol{x}_{1}$] at (0, 0) {}; 
          \node (x2c) [dot, label=right:$\boldsymbol{x}_{2}$] at (0, 2.5) {};
          \draw [link] (x2c) -- (y1c.north);
          
          % y'2 - filled circle; updated y2
          \node (y2c) [circ_filled, label=right:$\boldsymbol{y}'_{2}$] at (0, -3.0) {}; 
          \node (x3c) [dot, label=left:$\boldsymbol{x}_{3}$] at (-2.5, -2.5) {}; 
          \node (x4c) [dot, label=left:$\boldsymbol{x}_{4}$] at (-2.5, -4.5) {}; 
          \draw [link] (x3c) -- (y2c.150);
          \draw [link] (x4c) -- (y2c.210);
          
          % new linkage in (c): (x1 and y'2)
          \draw [link] (x1c) -- (y2c.north);
        \end{scope}
        
      \end{tikzpicture}
    }
  \end{center}
  \caption{A schematic diagram of the crossover and reconciliation process in the proposed EVIL algorithm. The empty and filled circles represent the current and updated latent true records, respectively. The links between the observations and the latent truths are drawn with dashed lines. Starting from (a) with $\lambda_{1}=\lambda_{2}=1$, $\lambda_{3}=\lambda_{4}=2$, and $\boldsymbol{x}_{1}$ hitting $\boldsymbol{y}_{1}$ (i.e., $\boldsymbol{z}_{1}=0$, indicated with the blue star), the crossover process switches $\lambda_{1}$ to $2$, resulting in (b) and $\boldsymbol{z}_{1}=1$ (the red star). The reconciliation process then updates $\boldsymbol{y}'_{1}$ and $\boldsymbol{y}'_{2}$ based on the new linkage structure, while corresponding distortion indicators $\boldsymbol{z}_{1},\dots,\boldsymbol{z}_{4}$ and variational parameters are updated accordingly, resulting in (c).}
  \label{fig:evil-reconciliation}
\end{figure}

The evolutionary algorithm optimizes $N_{P}+N_{O}$ randomly initialized variational factors per generation.
Given $N_{E}$ total generations, the algorithm requires $N_{P}\times N_{O}\times N_{E}$ times the cost of optimizing a single objective function.
However, since the parents and offspring are independent, the update steps with CAVI can be easily parallelized, mitigating the increase in computational cost \citep{whitley2001}.
Appendix D provides implementation details for the EVIL algorithm and derivations of the variational factors.

\subsection{Point Estimation of Linkage Structure}
\label{sec:estimation}

Through the methods presented in Sections~\ref{sec:mcmc} and~\ref{sec:evil}, we can obtain the posterior samples and approximated posterior distribution of $\boldsymbol{\Lambda}$, respectively.
However, this is not sufficient for inference, as we often need to obtain a point estimate of the linkage structure.
The most common Bayes estimator is the one that minimizes a posterior expected loss.
Since CHOMPER clusters records into the latent true record of the same entity, we need to use a loss function to solve the clustering (partitioning) problem.
Such loss functions include Binder's loss \citep{binder1978} and the variation of information loss \citep{meila2007}.
In record linkage, minimizing Binder's loss usually yields better performance, because the variation of information loss tends to underestimate the number of clusters (i.e., the number of true entities) \citep{betancourt2022a}.
Therefore, we suggest obtaining the Bayes estimate of the linkage structure by solving
\begin{equation}
  \label{eqn:point-estimator}
  \hat{\boldsymbol{\Lambda}}=\arg\min_{\boldsymbol{\Lambda}}\mathbb{E}\left[\text{L}\left(\boldsymbol{\Lambda}^{\star},\boldsymbol{\Lambda}\right)\mid\boldsymbol{\Xi}\right],
\end{equation}
where $\text{L}\left(\boldsymbol{\Lambda}^{\star},\boldsymbol{\Lambda}\right)$ is Binder's loss function and $\boldsymbol{\Lambda}^{\star}$ is the true index of the partition.

The Bayes estimate is computed using the posterior similarity matrix based on MCMC samples or variational factors from EVIL.
$\hat{\boldsymbol{\Lambda}}$ satisfying Equation~\eqref{eqn:point-estimator} is found via the \texttt{R} package \texttt{salso} \citep{dahl2022,dahl2024}.
However, if the data are too large and \texttt{salso} is computationally infeasible, we use the posterior similarity matrix directly, treating pairs whose posterior similarity exceeds a threshold as linked.
Additionally, $\hat{\boldsymbol{\Lambda}}$ implicitly estimates the true number of entities $N\leq N_{\max}$ by counting unique indices in the estimate.

\section{Real Data Application}
\label{sec:application}

In this section, we apply the CHOMPER model to real-world data containing attributes that may have multiple truths.
To showcase the flexibility of our method, we compare estimation performance across various modeling approaches commonly taken in practice, such as excluding unreliable variables or naively treating local variability as random distortion.

We estimate the linkage structure of the Italian Survey on Household Income and Wealth (ISHIW) data \citep{bankofitaly2025} to find the same respondents over repeated occurrences of the survey.
The Bank of Italy conducts the survey every 2 years, and we link the most recent data from 2020 and 2022.
Since the true identifiers are provided for each individual, we can use them to evaluate the performance of the model.
The surveys include 15,198 and 23,057 participants, respectively, with 28,000 unique individuals participating in one or both surveys.
Among them, 10,255 (36.6\%) participated in both surveys, and we refer to this as the overlap ratio such that
\begin{equation*}
  \text{Overlap Ratio} = \frac{n\left(\bigcup_{1\leq i_{1}\leq i_{2}\leq k}\left(S_{i_{1}}\cap S_{i_{2}}\right)\right)}{n\left(\bigcup_{i=1}^{k} S_{i}\right)},
\end{equation*}
where $S_{i}$ is the set of entities in the $i$th file and $n\left(S\right)$ is the number of elements in the set.

\subsection{Model Specification}
\label{sec:ishiw-model-specification}

The ISHIW dataset consists of over 300 questions, including demographic, household status, employment status, and asset-related variables.
We will use 7 attributes to perform linkage (sex, year of birth, whether Italian or not, place of birth, education level, net annual income, and estimated price of residence).
See Appendix E for the description of each attribute.

We assume that the four demographic attributes (sex, year of birth, whether Italian or not, and place of birth) have strong, immutable signals, as they are expected to remain unchanged across survey periods.
However, educational level, net annual income, and estimated residence price may have multiple truths.
As discussed in Section~\ref{sec:comprehensive-hit}, locally-varying hits are reasonable for these variables, since one can complete a degree, a tax policy may change during the survey period, or a respondent's subjective evaluation of their residence may differ.

With these 7 variables, two reasonable approaches would mimic traditional hit-or-miss methods.
The first is a Conservative approach using only the 4 strongly immutable demographic variables, as in \cite{steorts2015} and \cite{steorts2016}.
The second is a Naive approach using all 7 variables without allowing for locally-varying hits, similar to \cite{marchant2021} and \cite{menezes2024}.
Alternatively, a researcher may apply CHOMPER to all 7 variables and account for potential multiple truths (Comprehensive).
Both the Conservative and Naive approaches estimate linkage structures with the traditional hit-or-miss framework, approximated in CHOMPER by setting $\varepsilon_{\ell}\rightarrow0$ and $\tau_{\ell}\rightarrow0$ in Equation~\eqref{eqn:softmax-chomper} and~\eqref{eqn:gaussian-chomper}, respectively.
However, the Comprehensive approach controls $\varepsilon_{\ell}$ and $\tau_{\ell}$ to enable locally-variable comprehensive hits.
For instance, we set $\varepsilon_{\text{Education}}=1$ as no rapid changes are expected, while the hitting ranges for net annual income and estimated residence price require domain expertise.
Empirical results (Section~\ref{sec:simulation}) suggest erring towards a wider hitting range when eliciting a prior; see Appendix F.4 for a sensitivity analysis.
Without prior information, we set $\varepsilon_{\text{Income}}=\varepsilon_{\text{Price}}=0.1$, representing 10\% of the total variability in the sample data.
Based on the 2020 ISHIW data (standard deviations of \euro{53,117} and \euro{323,180} for net annual income and estimated residence price, respectively), this setting allows local variability of $\pm$\euro{16,797} and $\pm$\euro{102,198} per observation in 2022.

\begin{table}[H]
  \centering\caption{Model specification according to the approaches for modeling the ISHIW data. Conservative and Naive approaches use infinitesimal hitting range, while Comprehensive approach sets wide hitting ranges for local variability around the truth.}
  \label{tab:ishiw-models}
  \renewcommand{\arraystretch}{0.8}
  \begin{tabular}{cccc}
    \hline
    \multirow{2}{*}{Approach} & \multicolumn{2}{c}{Number   of Attributes} & \multirow{2}{*}{Prior Assumption}                                                  \\ \cline{2-3}
                              & Single Truth                               & Local Variability                 &                                                \\ \hline
    Conservative              & 4                                          & -                                 & $\varepsilon_{\ell},\tau_{\ell}\rightarrow0$   \\
    Naive                     & 4                                          & 3                                 & $\varepsilon_{\ell},\tau_{\ell}\rightarrow0$   \\
    Comprehensive             & 4                                          & 3                                 & $\varepsilon_{\ell}>0,\tau_{\ell}\rightarrow0$ \\ \hline
  \end{tabular}
\end{table}

The dataset has 38,255 records, making it computationally expensive to link all records at once.
For large sample sizes, researchers often partition the data using a blocking variable, performing linkage estimation for each block separately \citep{jaro1989}.
This approach is similar to stratified sampling, as blocks are expected to be mutually exclusive or have different characteristics.
However, strictly blocking records inevitably causes missing links when truly linked records contain errors in the blocking variables \citep{murray2015}.
Despite this drawback, the full dataset is too large for existing record linkage methods, including CHOMPER using MCMC (CHOMPER-MCMC).
Therefore, we will fit the model using current residence area, defined by 20 administrative regions, as a blocking variable, consistent with prior studies \citep{steorts2015,steorts2016,marchant2021,menezes2024}.
We will also present results from modeling the entire dataset together to demonstrate the scalability of CHOMPER using EVIL (CHOMPER-EVIL).

To initialize the MCMC sampler, we set $\boldsymbol{\Lambda}$, $\boldsymbol{y}$, and $\boldsymbol{z}$ assuming no records are linked.
We assume that the categorical fields follow multinomial distributions, continuous variables follow Gaussian distributions, and all fields have low distortion, setting the average distortion rate to 1\%, such that $\beta_{\ell}\sim\text{Beta}\left(N_{\max}^{r}\times0.1\times0.01,N_{\max}^{r}\times0.1\right)$, where $N_{\max}^{r}$ is the number of records in region $r$ \citep{marchant2021}.
Detailed descriptions of the other priors and their hyperparameter specifications are in Appendix E.
CHOMPER-EVIL uses the same hyperparameters and priors as CHOMPER-MCMC, with initial values set via the strategy in Appendix D.3.
In addition, we set the number of members in each generation to $N_{P}=50$ and $N_{O}=100$, and the maximum number of evolutions to $N_{E}=50$.
We use AMD Milan EPYC 7543 Processors to run all models.

For inference, we use a point estimate of the linkage structure discussed in Section~\ref{sec:estimation} using $\boldsymbol{\Lambda}$s sampled through CHOMPER-MCMC or variational factors with the largest ELBO in the last generation.
The convergence of CHOMPER-MCMC is assessed via trace plots of the number of unique entities in $\boldsymbol{\Lambda}$ \citep{aleshin-guendel2024} (Appendix E.2).
For CHOMPER-EVIL, convergence is determined if the relative change in the ELBO between successive generations is less than $10^{-5}$, the mutation step increases the ELBO, or the number of evolutions reaches $N_{E}$.

The estimated linkage structure offers four types of information to evaluate the performance of the estimation: correct linkage (TP; true positive), correct non-linkage (TN; true negative), incorrect linkage (FP; false positive), and incorrect non-linkage (FN; false negative).
F\textsubscript{1}-score, false negative rate (FNR), and false discovery rate (FDR) can be derived from these quantities and will be used to evaluate and analyze the performance of the model for each scenario, where
\begin{equation*}
  \text{F\textsubscript{1}}=\frac{2\times\text{TP}}{2\times\text{TP}+\text{FP}+\text{FN}}, \quad
  \text{FNR}=\frac{\text{FN}}{\text{FN}+\text{TP}}, \quad\text{and}\quad
  \text{FDR}=\frac{\text{FP}}{\text{FP}+\text{TP}}.
\end{equation*}

\subsection{Results}
\label{sec:application-results}

Table~\ref{tab:ishiw-regional} presents the estimation performance blocked by region for each approach.
In terms of F\textsubscript{1}-score, the Comprehensive approach using MCMC for estimation outperforms both the Conservative and Naive approaches, with the exception of region 2.
Moreover, we observe that the additional variables used in the Naive approach adversely affected performance by resulting in very few (or no) estimated links.

\begin{table}[H]
  \centering
  \caption{Results of the linkage structure estimation of regional ISHIW data using MCMC and EVIL. Each row represents the estimated performance by region, and the total row shows the aggregated performance of regional estimation. The performances are measured with F\textsubscript{1}-score (F\textsubscript{1}), false negative rate (FNR), and false discovery rate (FDR). The results with the highest F\textsubscript{1}-score for each region are bolded.}
  \label{tab:ishiw-regional}
  \renewcommand{\arraystretch}{0.8}
  \begin{tabular}{ccccccccccccc}
    \hline
    \multirow{3}{*}{Region} & \multicolumn{3}{c}{\multirow{2}{*}{Conservative}} & \multicolumn{3}{c}{\multirow{2}{*}{Naive}} & \multicolumn{6}{c}{Comprehensive}                                                                                                                \\ \cline{8-13}
                            & \multicolumn{3}{c}{}                              & \multicolumn{3}{c}{}                       & \multicolumn{3}{c}{MCMC}          & \multicolumn{3}{c}{EVIL}                                                                                     \\ \cline{2-13}
                            & F\textsubscript{1}                                & FNR                                        & FDR                               & F\textsubscript{1}       & FNR  & FDR  & F\textsubscript{1} & FNR  & FDR  & F\textsubscript{1} & FNR  & FDR  \\ \hline
    1                       & 0.17                                              & 0.72                                       & 0.87                              & 0.00                     & 1.00 & 0.00 & \textbf{0.30}      & 0.66 & 0.73 & 0.13               & 0.84 & 0.89 \\
    2                       & \textbf{0.65}                                     & 0.49                                       & 0.13                              & 0.00                     & 1.00 & -    & 0.58               & 0.59 & 0.06 & 0.63               & 0.27 & 0.44 \\
    3                       & 0.15                                              & 0.70                                       & 0.90                              & 0.01                     & 1.00 & 0.25 & 0.28               & 0.69 & 0.75 & \textbf{0.32}      & 0.71 & 0.64 \\
    4                       & 0.13                                              & 0.80                                       & 0.90                              & 0.00                     & 1.00 & -    & 0.23               & 0.77 & 0.77 & \textbf{0.27}      & 0.69 & 0.76 \\
    5                       & 0.11                                              & 0.79                                       & 0.92                              & 0.00                     & 1.00 & 1.00 & 0.26               & 0.70 & 0.78 & \textbf{0.29}      & 0.72 & 0.70 \\
    6                       & 0.36                                              & 0.65                                       & 0.63                              & 0.01                     & 1.00 & 0.00 & \textbf{0.46}      & 0.61 & 0.45 & 0.44               & 0.44 & 0.64 \\
    7                       & 0.39                                              & 0.58                                       & 0.63                              & 0.01                     & 1.00 & 0.00 & \textbf{0.52}      & 0.55 & 0.38 & 0.48               & 0.44 & 0.58 \\
    8                       & 0.13                                              & 0.75                                       & 0.92                              & 0.01                     & 0.99 & 0.33 & 0.26               & 0.70 & 0.78 & \textbf{0.27}      & 0.74 & 0.72 \\
    9                       & 0.19                                              & 0.73                                       & 0.85                              & 0.00                     & 1.00 & 0.50 & \textbf{0.33}      & 0.65 & 0.69 & 0.29               & 0.63 & 0.76 \\
    10                      & 0.22                                              & 0.73                                       & 0.82                              & 0.00                     & 1.00 & 1.00 & 0.27               & 0.74 & 0.70 & \textbf{0.30}      & 0.64 & 0.74 \\
    11                      & 0.24                                              & 0.75                                       & 0.78                              & 0.01                     & 1.00 & 0.00 & \textbf{0.39}      & 0.64 & 0.59 & 0.29               & 0.61 & 0.77 \\
    12                      & 0.19                                              & 0.68                                       & 0.86                              & 0.00                     & 1.00 & -    & \textbf{0.33}      & 0.65 & 0.68 & 0.31               & 0.58 & 0.76 \\
    13                      & 0.20                                              & 0.75                                       & 0.83                              & 0.00                     & 1.00 & -    & \textbf{0.38}      & 0.61 & 0.63 & 0.19               & 0.69 & 0.86 \\
    14                      & 0.33                                              & 0.71                                       & 0.63                              & 0.00                     & 1.00 & -    & \textbf{0.44}      & 0.61 & 0.49 & 0.35               & 0.46 & 0.74 \\
    15                      & 0.05                                              & 0.89                                       & 0.97                              & 0.00                     & 1.00 & 0.80 & \textbf{0.12}      & 0.80 & 0.92 & 0.08               & 0.95 & 0.80 \\
    16                      & 0.05                                              & 0.85                                       & 0.97                              & 0.01                     & 1.00 & 0.78 & \textbf{0.09}      & 0.81 & 0.94 & \textbf{0.09}      & 0.95 & 0.55 \\
    17                      & 0.26                                              & 0.76                                       & 0.73                              & 0.04                     & 0.98 & 0.33 & \textbf{0.37}      & 0.67 & 0.59 & 0.33               & 0.46 & 0.76 \\
    18                      & 0.20                                              & 0.78                                       & 0.82                              & 0.00                     & 1.00 & -    & \textbf{0.29}      & 0.68 & 0.73 & 0.26               & 0.68 & 0.79 \\
    19                      & 0.05                                              & 0.86                                       & 0.97                              & 0.00                     & 1.00 & 1.00 & \textbf{0.13}      & 0.76 & 0.91 & 0.11               & 0.90 & 0.88 \\
    20                      & 0.10                                              & 0.82                                       & 0.94                              & 0.00                     & 1.00 & -    & \textbf{0.23}      & 0.69 & 0.82 & 0.17               & 0.84 & 0.82 \\ \hline
    Total                   & 0.13                                              & 0.77                                       & 0.91                              & 0.00                     & 1.00 & 0.62 & 0.23               & 0.70 & 0.81 & \textbf{0.25}      & 0.74 & 0.76 \\ \hline
  \end{tabular}
\end{table}

In general, the Conservative approach has a higher overall FDR than the Comprehensive approach.
Using only 4 categorical variables resulted in insufficient information, complicating the differentiation of distinct entities and causing records with the same value to be clustered into a single entity.
With the Naive approach, the model determines that records with marginally different values do not originate from the same entity, even with a higher distortion prior, resulting in undefined FDRs due to no links (Appendix E.3).
Consequently, with the Comprehensive approach, even if the records have different values, the locally-varying hit mechanism increases the probability that they pertain to the same entity.
Figure~\ref{fig:ishiw-network} visually illustrates this behavior.

\begin{figure}[H]
  \begin{center}
    \includegraphics[width=0.9\linewidth]{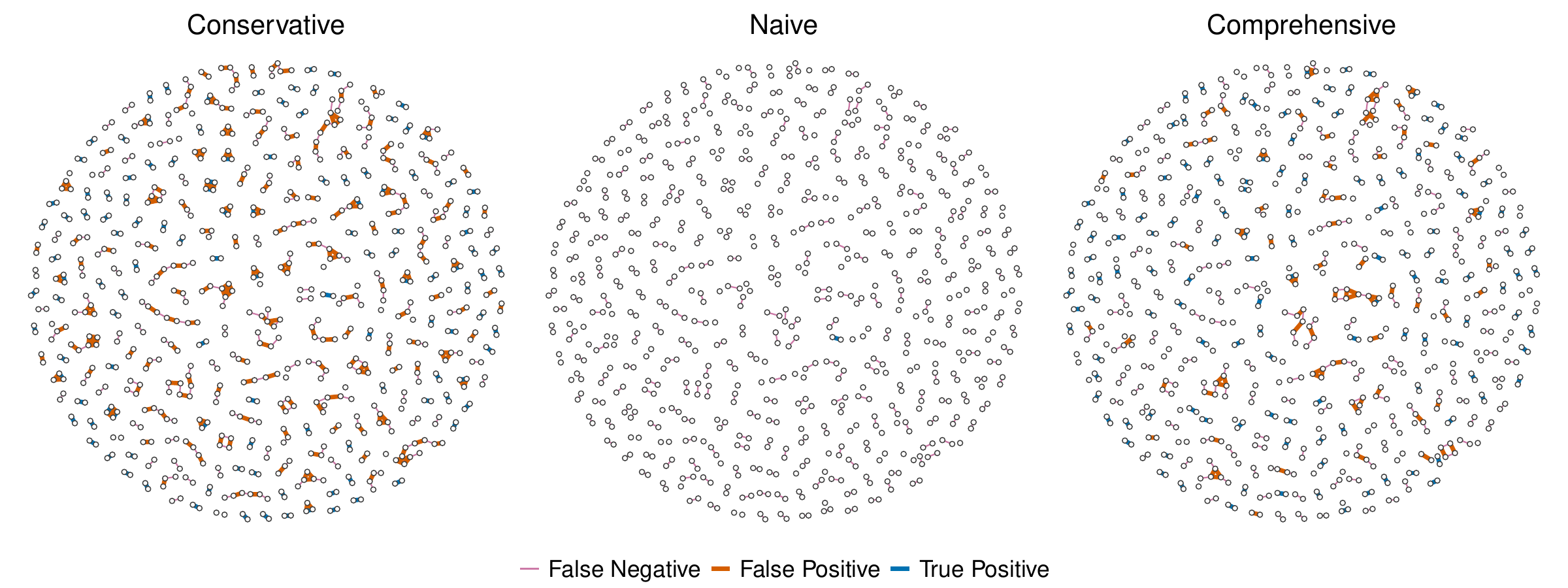}
  \end{center}
  \caption{Network graph using the linkage structure estimate for region 7. The nodes represent the records, and the edges represent the results of the estimation. Comparing the Comprehensive and Conservative approaches, we see that the Comprehensive approach reduces false positive links (thick orange edges) and increases true positive links (thick blue edges). We also observe that the Naive approach fails to identify the vast majority of links between records (thin purple edges).}
  \label{fig:ishiw-network}
\end{figure}

Figure~\ref{fig:ishiw-n-entities} shows the estimation performance of the number of unique entities for regions 3, 7, 12, and 20.
Although a considerable FDR reduction in regions 7 and 12 yielded an overestimation of the true number of unique entities, overall, the Comprehensive approach provides estimates closer to the true number of unique entities.
See Appendix E.3 for graphs for all regions.

\begin{figure}[H]
  \begin{center}
    \includegraphics[width=0.9\linewidth]{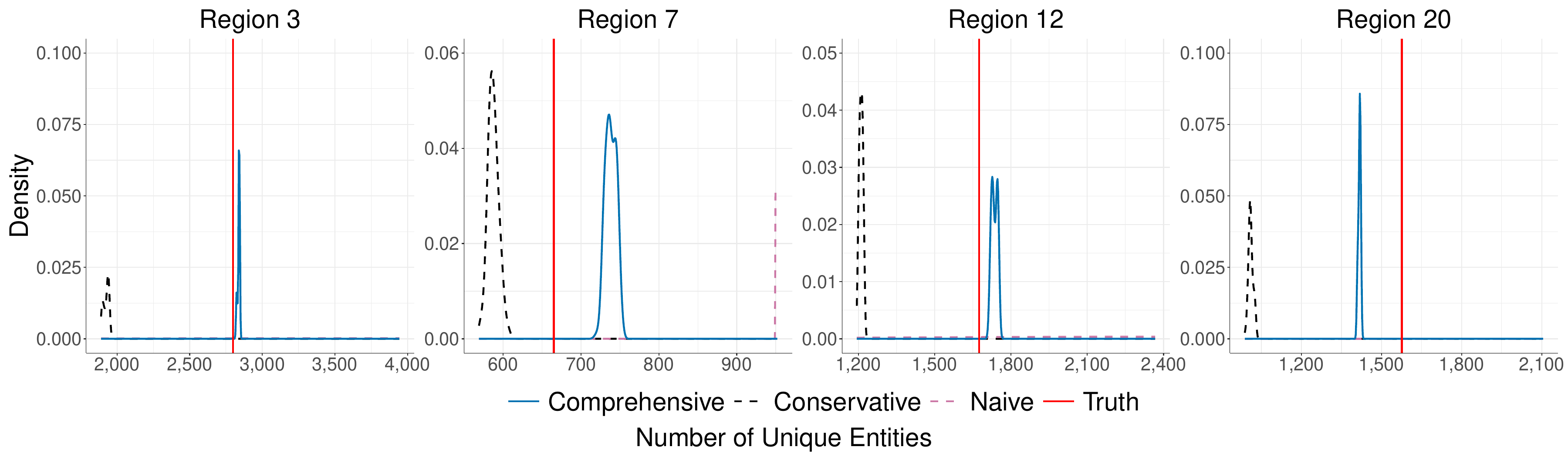}
  \end{center}
  \caption{Estimated density of the number of unique entities for regions 3, 7, 12, and 20. The red vertical line represents the true number of unique entities, and the solid blue, dashed black, and dashed purple lines represent the estimated density using MCMC samples from the Comprehensive, Conservative, and Naive approaches, respectively.}
  \label{fig:ishiw-n-entities}
\end{figure}

As mentioned in Section~\ref{sec:intro}, the CHOMPER model has the advantage of being able to estimate the linkage structure through VI.
If a Comprehensive approach is chosen, the performance of estimating $\boldsymbol{\Lambda}$ using CHOMPER-EVIL is superior to the performance using MCMC with the Conservative approach for 17 regions in terms of F\textsubscript{1}-score.
Moreover, it outperforms the Naive approach for all regions in terms of F\textsubscript{1}-score and FNR.

\begin{figure}[H]
  \begin{center}
    \includegraphics[width=0.9\linewidth]{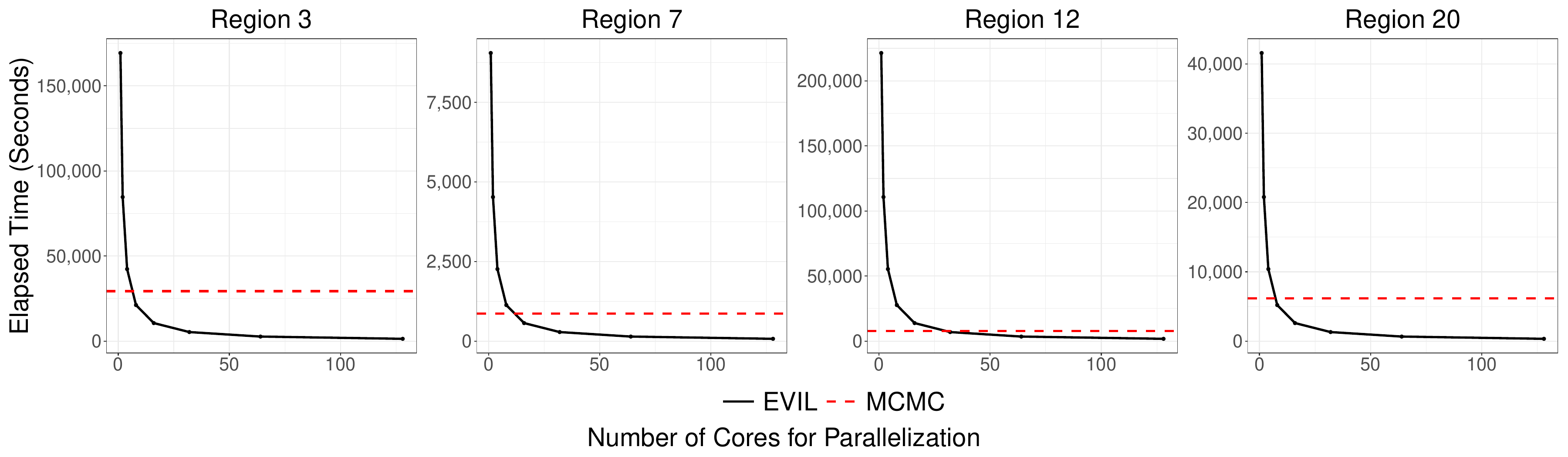}
  \end{center}
  \caption{Elapsed time in seconds of finding posterior distribution of $\boldsymbol{\Lambda}$ using CHOMPER-MCMC or CHOMPER-EVIL for the Comprehensive approach with varying number of cores used for parallelization for regions 3, 7, 12, and 20.}
  \label{fig:ishiw-timing-some}
\end{figure}

Figure~\ref{fig:ishiw-timing-some} shows that the elapsed time of fitting the model through EVIL (solid black line) decreases as the number of cores increases, whereas the speed of the MCMC algorithm (dashed red line) is not affected by the number of cores since each iteration is sequential.
While EVIL may require more time than MCMC on a single core due to up to $N_{P}\times N_{O}\times N_{E}$ independent CAVI optimizations, the $N_{P}\times N_{O}$ optimizations can be parallelized, allowing the approximate posterior distribution to be obtained faster than MCMC when a sufficient number of cores is used.
For example, in region 3, the elapsed time is shorter than that of MCMC with more than 8 cores, and in region 7, it becomes shorter with more than 16 cores.
Although both MCMC and a single CAVI optimization take longer as the data size increases, the elapsed time of EVIL may remain shorter regardless of data size, because EVIL includes randomness in its stopping criterion, stopping when the ELBO increases after the mutation step.
Appendix E.4 provides timing plots for all regions.

To further demonstrate the scalability of CHOMPER-EVIL, we fit the entire linkage structure of the ISHIW data without blocking, which, to our knowledge, has not been feasible with existing methods.
Note that since only 3 participants changed their residence area across surveys, the number of true negatives greatly increases, whereas the number of true positives remains similar to the blocked analysis.
We warm-started the algorithm using the regional results $\hat{\boldsymbol{\Lambda}}_{1},\cdots,\hat{\boldsymbol{\Lambda}}_{20}$ as the initial linkage structure, where $\hat{\boldsymbol{\Lambda}}_{r}$ is the estimate for region $r$.
We used a smaller population for the evolutionary algorithm ($N_{P}=3$ and $N_{O}=4$) to accommodate the increased memory demands of the full dataset.
To obtain a Bayes estimate, we determined linkage based on a posterior similarity matrix with a threshold of 0.5, as computing the entire point estimate with \texttt{salso} was computationally infeasible.
CHOMPER-EVIL completed in approximately 82.6 hours, resulting in F\textsubscript{1}-score$=0.19$, FNR$=0.67$, and FDR$=0.86$ by linking records across regions, similar to the aggregated regional MCMC performance values in Table~\ref{tab:ishiw-regional}.
However, CHOMPER-MCMC did not converge even after approximately 644.3 hours to generate 100,000 samples with a warm-start, as shown in Figure~\ref{fig:ishiw-traceplot-entire}.

\begin{figure}[H]
  \begin{center}
    \includegraphics[width=0.9\linewidth]{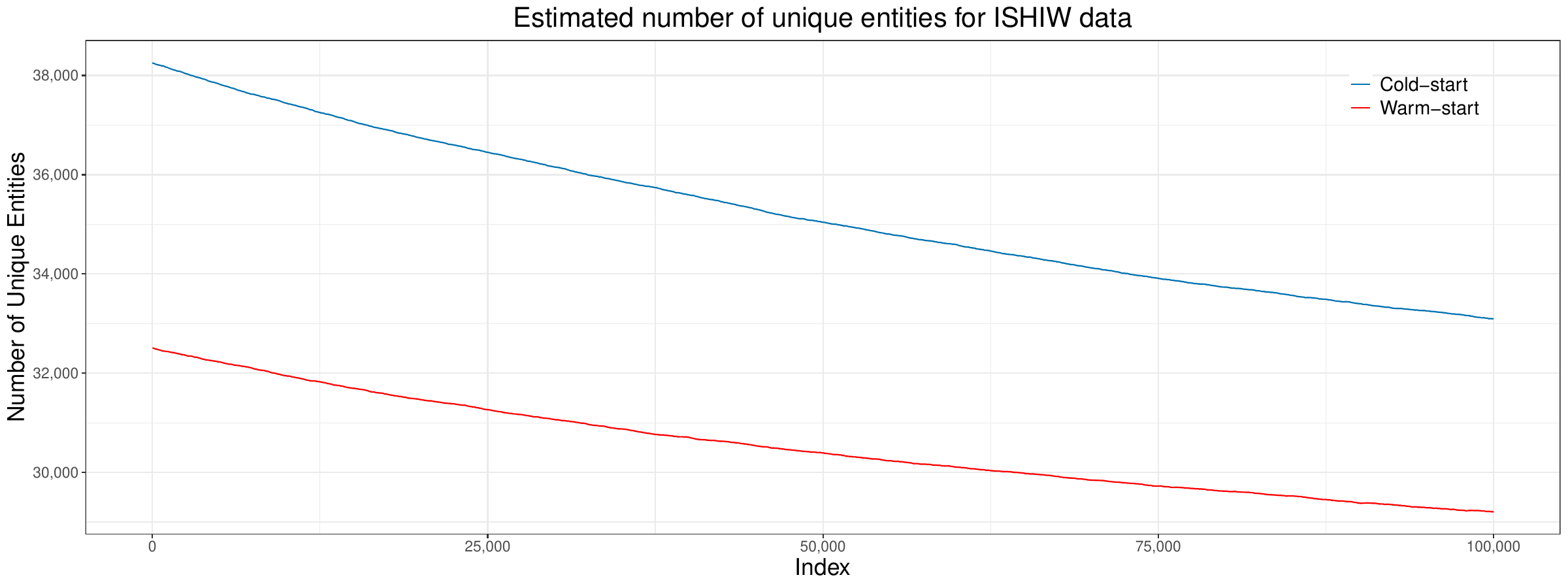}
  \end{center}
  \caption{Trace plots of the estimated number of unique $\lambda_{ij}$ for the entire ISHIW data using CHOMPER-MCMC.}
  \label{fig:ishiw-traceplot-entire}
\end{figure}

In summary, the Comprehensive approach demonstrated superior performance, while the Naive approach showed significant degradation.
The Conservative approach exhibited a high FDR by incorrectly linking distinct entities with identical values.
However, the Comprehensive approach effectively mitigated both high FDRs and FNRs.
The reduction in FDR should not be ignored if researchers plan to perform downstream analysis, as errors from a probabilistic record linkage model can bias downstream analyses in multiple ways \citep{doidge2019}.
Suppose practitioners merge one file that contains the dependent variable and another that contains the independent variables, and then run a regression analysis.
A model with a high FDR erroneously links different entities, matching a dependent variable with random independent variables, attenuating the corresponding estimated regression coefficient.
By contrast, a high FNR leaves many records unlinked, reducing the usable sample and statistical power.
Hence, the Comprehensive approach using the CHOMPER model can be a good choice for downstream analysis using the merged ISHIW data to limit attenuation in downstream inference with minimal loss in power.

\section{Simulation Study}
\label{sec:simulation}

This section aims to evaluate the effectiveness and robustness of the CHOMPER model through simulations using synthetic data for multiple cases that researchers may encounter but have not been explored in Section~\ref{sec:application}.
Although we demonstrated that our model can improve linkage structure estimation performance when adopting locally-varying hits to variables whose values differ from the truths without distortion, it is difficult to compare changes in model performance across different hitting ranges for locally-varying hits because the exact generation process of real data is unknown.

To investigate the performance of CHOMPER, we generated 30 sets of synthetic data according to the model defined in Equation~\eqref{eqn:chomper-mdl}.
The data consist of records from $N=750$ unique entities across $k=2$ files.
Each record comprises 5 discrete variables with 8 levels and 2 continuous variables, with 2 discrete and 2 continuous variables having multiple truths.
We further assume that the discrete variables follow a multinomial distribution and the continuous variables follow a Gaussian distribution.
Moreover, to analyze the effect of the data overlap ratio on estimation performance, we generated simulated data for cases with high (70\%), medium (50\%), and low (30\%) overlap ratios.

We designed simulations for two scenarios researchers might encounter: (Scenario 1) when the data contain only categorical variables and (Scenario 2) when the data contain both categorical and continuous variables.
Furthermore, researchers must choose a model that best fits their data.
We evaluate the Conservative, Naive, and Comprehensive approaches, as specified in Section~\ref{sec:ishiw-model-specification}.
Moreover, we compare to an Oracle model in which the true hitting range, which is available since the true data generating process is known, is provided to the CHOMPER model.
The true hitting range is chosen from the data-generating process so that it captures more than 95\% of the local variability.
However, this information is not available in practice.
Therefore, the selection of the hitting range depends on the researcher's domain knowledge.
In this simulation study, we evaluate the model performance by setting a slightly wider hitting range than the Oracle for the Comprehensive approach based on empirical findings.
See Appendix F.4 for a sensitivity analysis of the hitting ranges.

\begin{table}[H]
  \centering
  \caption{The average (standard deviation) of the linkage structure estimation performances of 30 simulation datasets with low overlap ratio (30\%) with the CHOMPER model using MCMC and EVIL. The performances are measured with F\textsubscript{1}-score (F\textsubscript{1}), false negative rate (FNR), and false discovery rate (FDR).}
  \label{tab:simulation-results}
  
  \begin{subtable}{\textwidth}
    \centering
    % \captionsetup{font=footnotesize}
    \caption{Scenario 1: Estimation performance when only categorical fields are used.}
    \label{tab:simulation-scenario-1}
    \scalebox{0.9}{
      \begin{tabular}{ccccccc}
        \hline
        \multirow{2}{*}{Approach} & \multicolumn{3}{c}{MCMC} & \multicolumn{3}{c}{EVIL}                                                                \\ \cline{2-7}
                                  & F\textsubscript{1}       & FNR                      & FDR         & F\textsubscript{1} & FNR         & FDR         \\ \hline
        Conservative              & 0.27 (0.02)              & 0.71 (0.03)              & 0.75 (0.02) & 0.31 (0.01)        & 0.46 (0.04) & 0.78 (0.01) \\
        Naive                     & 0.15 (0.03)              & 0.92 (0.02)              & 0.21 (0.08) & 0.23 (0.02)        & 0.86 (0.02) & 0.32 (0.07) \\
        Comprehensive             & 0.50 (0.03)              & 0.55 (0.04)              & 0.43 (0.04) & 0.57 (0.03)        & 0.43 (0.04) & 0.41 (0.03) \\
        Oracle                    & 0.46 (0.03)              & 0.66 (0.03)              & 0.31 (0.04) & 0.50 (0.03)        & 0.60 (0.03) & 0.33 (0.03) \\ \hline
      \end{tabular}
    }
  \end{subtable}
  
  \vspace{\baselineskip}
  
  \begin{subtable}{\textwidth}
    \centering
    % \captionsetup{font=footnotesize}
    \caption{Scenario 2: Estimation performance when both categorical and continuous fields are used.}
    \label{tab:simulation-scenario-2}
    \scalebox{0.9}{
      \begin{tabular}{ccccccc}
        \hline
        \multirow{2}{*}{Approach} & \multicolumn{3}{c}{MCMC} & \multicolumn{3}{c}{EVIL}                                                                \\ \cline{2-7}
                                  & F\textsubscript{1}       & FNR                      & FDR         & F\textsubscript{1} & FNR         & FDR         \\ \hline
        Conservative              & 0.27 (0.02)              & 0.71 (0.03)              & 0.75 (0.02) & 0.31 (0.01)        & 0.46 (0.04) & 0.78 (0.01) \\
        Naive                     & 0.02 (0.01)              & 0.99 (0.01)              & 0.04 (0.10) & 0.07 (0.02)        & 0.97 (0.01) & 0.14 (0.10) \\
        Comprehensive             & 0.58 (0.03)              & 0.48 (0.03)              & 0.34 (0.03) & 0.39 (0.01)        & 0.10 (0.02) & 0.75 (0.01) \\
        Oracle                    & 0.57 (0.04)              & 0.57 (0.04)              & 0.14 (0.03) & 0.69 (0.02)        & 0.09 (0.04) & 0.45 (0.03) \\ \hline
      \end{tabular}
    }
  \end{subtable}
\end{table}

Table~\ref{tab:simulation-results} shows the average performance on the 30 replicate datasets with low overlap ratio (30\%), including standard deviations in parentheses.
As in Section~\ref{sec:application}, the F\textsubscript{1}-score increases when more information is used for estimation and locally-varying hits are properly accounted for.
Note that in Scenario 2, since all the additional continuous fields contain multiple truths, the Conservative approach uses only categorical fields with a single immutable truth for estimation.
Thus, the performance is identical to that of the Conservative approach in Scenario 1.

In both scenarios, the Oracle and Comprehensive approaches outperformed the Conservative and Naive approaches in terms of F\textsubscript{1}-score.
Further, naively adding variables without considering local variability deteriorated model performance compared to the Conservative approach.
The Conservative approach created many false positives by linking records with identical values, whereas the Naive approach failed to link records with different values that originate from the same true latent record, resulting in many false negatives.
The greater performance degradation of the Naive approach relative to the Conservative approach in Scenario 2 than in Scenario 1 indicates that continuous fields are more sensitive to local variability.
By accounting for local variability, the Comprehensive and Oracle approaches linked records that the Naive approach left unlinked.
However, the changes in F\textsubscript{1}-score from the Oracle approach to the Comprehensive approach varied depending on the inferential methods and scenarios.
For CHOMPER-MCMC, the decrease in FNR exceeded the increase in FDR in both scenarios, yielding higher F\textsubscript{1}-scores for the Comprehensive approach than for the Oracle approach.
For CHOMPER-EVIL, the same pattern held in Scenario 1, but a large FDR increase due to a wider hitting range in Scenario 2 led to lower F\textsubscript{1}-scores.
The Comprehensive CHOMPER-EVIL increased both FDR and FNR at low overlap ratios but decreased FNR at medium and high overlap ratios, yet its F\textsubscript{1}-score remained below the Oracle approach.
For more detailed results on different overlap ratios and hitting range misspecification, see Appendix F.4.

In summary, the CHOMPER model enables researchers to use more information for linkage structure estimation by correctly handling attributes with multiple truths, a challenge commonly encountered but often ignored in practice.
It also demonstrates robustness to hitting range misspecification when all multiple truth variables follow a multinomial distribution.
Furthermore, using more information appropriately yields a more accurate estimate, affirming that the superior performance demonstrated in Section~\ref{sec:application} holds across various conditions in this simulation.

\section{Conclusion and Future Work}
\label{sec:conclusion}

The importance of integrating data collected from different subjects and lacking unique identifiers continues to increase, as current trends in data analysis emphasize the use of large-scale, anonymized datasets from multiple sources such as wearable devices, social media, and online platforms.
Previously existing probabilistic entity resolution models can handle only variables with a single true value per entity.
This severely limits the accuracy of entity resolution in many real world scenarios where the ``truth'' is not persistent across databases.

As a first attempt to solve this problem, we propose the CHOMPER model, which has three major contributions.
First, our model can handle cases where each entity may have variability around the single truth by proposing a new hitting mechanism in the likelihood structure for the observations.
By parameterizing the likelihood as a mixture of two distributions from the exponential family and introducing hyperparameters to control the hitting range, our model can flexibly account for both local variability and a single true value, improving entity resolution accuracy in many domains.
Second, our model is widely applicable to the data structures found in practice, as it can learn linkage structure using any fields that follow an exponential family distribution for entity resolution.
Finally, the proposed likelihood structure improves the scalability of the record linkage model by employing VI for estimation, which further extends its utility to modern data structures.

There are several directions in which the concept of multiple truths can be extended to meet research objectives.
First, we can investigate a likelihood structure that reflects the spatio-temporal tendency of variables with multiple truths \citep{robach2025}.
Furthermore, although the model proposed in this paper assumes that the fields are mutually independent, it would be advantageous to account for dependency among linkage variables.
Moreover, it may be possible to develop a model that can link records in a streaming setting, as in \cite{taylor2024} and \cite{taylor2025}, because traditional hit-or-miss entity resolution models have been limited to fixed datasets.

\bigskip
\begin{center}
  {\large\bf SUPPLEMENTARY MATERIAL}
\end{center}

\begin{description}

  \item[Supplementary Appendices:] A document including the proof of Proposition~\ref{prop:degeneracy}, derivation of full conditional distributions and variational factors, description of EVIL's initialization strategy and optimization algorithms, hyperparameters used in simulation studies and real data analysis, convergence diagnostics, and additional experimental results. (.pdf file)
  \item[Reproducible Code Base:] R code, source code (C++/Rcpp), and datasets to reproduce all results in this manuscript. (.zip file)
  \item[R-package \texttt{chomper}:] An R package implementing CHOMPER and a vignette demonstrating its use are available on CRAN.
  \\ Link: \url{https://cran.r-project.org/web/packages/chomper/index.html}
        
\end{description}

\bibliographystyle{agsm}
\bibliography{chomper}

\appendix
\clearpage % 메인 본문과 분리되도록 페이지를 넘김

% --- 1. Supplementary 타이틀 수동 생성 ---
\begin{center}
    \vspace*{2cm} % 위쪽 여백 조절
    \textbf{\LARGE Supplementary Appendices for ``A Comprehensive Bayesian Approach to Entity Resolution for Data with Multiple Truths''}
\end{center}
\vspace{1cm}

\setcounter{table}{0}
\renewcommand{\thetable}{S\arabic{table}}
\setcounter{figure}{0}
\renewcommand{\thefigure}{S\arabic{figure}}
\setcounter{equation}{0}
\renewcommand{\theequation}{S\arabic{equation}}

\section{Proof of Proposition 1}
\label{app:degeneracy}

In this section, we restate Proposition 1 in Section 2.1, and provide a proof.

\begin{proposition}
  \label{prop:app-degeneracy}
  Let \(X\) be a random variable following an exponential family distribution \(f\) with a natural parameter \(\boldsymbol{\zeta}\) such that \(f(x\mid\boldsymbol{\zeta}) = h(x)\exp\left(\boldsymbol{\zeta}T\left(x\right)-A(\boldsymbol{\zeta})\right)\) for the base measure \(h(x)\), the log-partition function \(A\left(\boldsymbol{\zeta}\right)\), and a sufficient statistic \(T\left(X\right)\).
  Then, \(T\left(X\right)\stackrel{p}{\rightarrow} \mathbb{E}\left[T\left(X\right)\right]=A'\left(\boldsymbol{\zeta}\right)\) as \(A''\left(\boldsymbol{\zeta}\right)\rightarrow0\), creating a degenerate distribution concentrated at \(A'\left(\boldsymbol{\zeta}\right)\).
\end{proposition}

\begin{proof}
  By the properties of the exponential family, the mean and variance of the sufficient statistic can be obtained as
  \begin{equation*}
    \begin{aligned}
      \mathbb{E}\left[T\left(X\right)\right] & = \frac{d}{d\boldsymbol{\zeta}}A\left(\boldsymbol{\zeta}\right)=A'\left(\boldsymbol{\zeta}\right) \text{ and } \\
      \text{Var}\left[T\left(X\right)\right] & = \frac{d^{2}}{d\boldsymbol{\zeta}^{2}}A\left(\boldsymbol{\zeta}\right)=A''\left(\boldsymbol{\zeta}\right).
    \end{aligned}
  \end{equation*}
  Then for an arbitrary $\epsilon>0$,
  \begin{equation*}
    \mathbb{P}\left(\lvert T\left(X\right)-A'\left(\boldsymbol{\zeta}\right)\rvert\geq\epsilon\right)\leq\frac{\text{Var}\left[T\left(X\right)\right]}{\epsilon^{2}}
  \end{equation*}
  by Chebyshev's inequality. As $\text{Var}\left[T\left(X\right)\right]=A''\left(\boldsymbol{\zeta}\right)\rightarrow0$ by assumption,
  \begin{equation*}
    \lim_{A''\left(\boldsymbol{\zeta}\right)\rightarrow0}\mathbb{P}\left(\lvert T\left(X\right)-A'\left(\boldsymbol{\zeta}\right)\rvert\geq\epsilon\right)=0.
  \end{equation*}
  Thus, the result above satisfies the definition of convergence in probability, which results in
  \begin{equation*}
    T\left(X\right)\stackrel{p}{\rightarrow}A'\left(\boldsymbol{\zeta}\right)
  \end{equation*}
  by Slutsky's theorem.
  This implies that as the variance vanishes, the distribution of the sufficient statistic concentrates entirely around its expectation.
  Thus, the distribution degenerates to a point mass at $A'\left(\boldsymbol{\zeta}\right)$.
\end{proof}

\section{Full Conditional Distributions}
\label{app:full-conditional}

This section provides derivations for the full conditional distributions of the CHOMPER model defined in Equation (2) of Section 2 in the main manuscript.
The full conditional distributions in this section are used in both CHOMPER-MCMC and CHOMPER-EVIL, because the CAVI updates require the log of the full conditional distributions \citep{bishop2006,blei2017}.

We assume that there are $p$ variables, of which $\ell_{1}$ follow multinomial distributions and $p-\ell_{1}+1$ follow Gaussian distributions.
Furthermore, we use conjugate priors for each likelihood.
That is, we assume a CHOMPER model of the following form:

\begin{equation}
  \label{eqn:chomper-mdl-appendix}
  \begin{aligned}
    x_{ij\ell}\mid\lambda_{ij},y_{\lambda_{ij}\ell},z_{ij\ell}
                                                                       & \stackrel{\text{ind}}{\sim}
    \begin{cases}
      \text{MN}\left(1,\boldsymbol{\theta}_{\ell}\right)^{z_{ij\ell}}\text{MN}\left(1,\tilde{\phi}_{\ell}\left(x_{ij\ell},y_{\lambda_{ij}\ell},\tau_{\ell},\varepsilon_{\ell}\right)\right)^{1-z_{ij\ell}} & \text{for } \ell=1,\cdots,\ell_{1}    \\
      \text{N}\left(\eta_{\ell},\sigma_{\ell}\right)^{z_{ij\ell}}\text{N}\left(y_{\lambda_{ij}\ell},\varepsilon_{\ell}\right)^{1-z_{ij\ell}}                                                               & \text{for } \ell=\ell_{1}+1,\cdots,p
    \end{cases} \\
    z_{ij\ell}\mid\beta_{\ell}                                         & \stackrel{\text{ind}}{\sim}\text{Bernoulli}\left(\beta_{\ell}\right)                                                                                                                                 \\
    y_{j'\ell}\mid\boldsymbol{\theta}_{\ell},\eta_{\ell},\sigma_{\ell} & \stackrel{\text{ind}}{\sim}
    \begin{cases}
      \text{MN}\left(1,\boldsymbol{\theta}_{\ell}\right) & \text{for } \ell=1,\cdots,\ell_{1}    \\
      \text{N}\left(\eta_{\ell},\sigma_{\ell}\right)     & \text{for } \ell=\ell_{1}+1,\cdots,p
    \end{cases}                                                                                                                                                                             \\
    \boldsymbol{\theta}_{\ell}                                         & \stackrel{\text{ind}}{\sim}\text{Dirichlet}\left(\boldsymbol{\mu}_{\ell}\right)                                                                                                                      \\
    \pi\left(\eta_{\ell}\right)                                        & \propto 1                                                                                                                                                                                            \\
    \sigma_{\ell}                                                      & \stackrel{\text{ind}}{\sim}\text{Inverse-Gamma}\left(a^{\sigma}_{\ell},b^{\sigma}_{\ell}\right)                                                                                                      \\
    \beta_{\ell}                                                       & \stackrel{\text{ind}}{\sim}\text{Beta}\left(a_{\ell},b_{\ell}\right)                                                                                                                                 \\
    \pi\left(\boldsymbol{\Lambda}\right)                               & \propto 1,
  \end{aligned}
\end{equation}
where
\begin{equation*}
  \tilde{\phi}_{\ell}\left(x_{ij\ell},y_{\lambda_{ij}\ell},\tau_{\ell},\varepsilon_{\ell}\right)
  =\frac{\phi_{\ell}^{\frac{1}{\tau_{\ell}}\mathbb{I}\left(\lvert x_{ij\ell}-y_{\lambda_{ij}\ell}\rvert\leq\varepsilon_{\ell}\right)}}{\sum_{m=1}^{M_{\ell}}\phi_{\ell}^{\frac{1}{\tau_{\ell}}\mathbb{I}\left(\lvert x_{ij\ell}-m\rvert\leq\varepsilon_{\ell}\right)}}
\end{equation*}
for $\varepsilon_{\ell}\geq0$, $\tau_{\ell}>0$, and $\phi_{\ell}>1$.

\subsection{Full Conditional Distributions of the Model}
\label{app:full-conditional-summary}

The CHOMPER model defined in Equation~\eqref{eqn:chomper-mdl-appendix} has full conditional distributions of
\begin{equation*}
  \begin{aligned}
    \boldsymbol{\theta}_{\ell}\mid\boldsymbol{\Xi}_{-\boldsymbol{\theta}_{\ell}}
     & \sim\text{Dirichlet}\left(\left\{\alpha'_{\ell m}\right\}_{m=1}^{M_{\ell}}\right)                                                                                                                                                                                                                         \\
    \eta_{\ell}\mid\boldsymbol{\Xi}_{-\eta_{\ell}}
     & \sim\text{N}\left(\frac{\sum_{j'}y_{j'\ell}+\sum_{i,j}z_{ij\ell}x_{ij\ell}}{N_{\max}+\sum_{i,j}z_{ij\ell}},\frac{\sigma_{\ell}}{N_{\max}+\sum_{i,j}z_{ij\ell}}\right)                                                                                                                                     \\
    \sigma_{\ell}\mid\boldsymbol{\Xi}_{-\sigma_{\ell}}
     & \sim\text{Inverse-Gamma}\left(a^{\sigma}_{\ell}+\frac{1}{2}\left[N_{\max}+\sum_{i,j}z_{ij\ell}\right],b'^{\sigma}_{\ell}\right)                                                                                                                                                                           \\
    \beta_{\ell}\mid\boldsymbol{\Xi}_{-\beta_{\ell}}
     & \sim\text{Beta}\left(a_{\ell}+\sum_{i=1}^{k}\sum_{j=1}^{n_{i}}z_{ij\ell},b_{\ell}+\sum_{i=1}^{k}\sum_{j=1}^{n_{i}}\left(1-z_{ij\ell}\right)\right)                                                                                                                                                        \\
    \lambda_{ij}\mid\boldsymbol{\Xi}_{-\lambda_{ij}}
     & \sim\text{MN}\left(1,\left\{\frac{\nu'_{ij,j'}}{\sum_{u=1}^{N_{\max}}\nu'_{ij,u}}\right\}_{j'=1}^{N_{\max}}\right)                                                                                                                                                                                        \\
    y_{j'\ell}\mid\boldsymbol{\Xi}_{-y_{j'\ell}}
     & \sim\begin{cases}
             \text{MN}\left(1,\left\{\frac{\gamma'_{j'\ell,m}}{\sum_{r=1}^{M_{\ell}}\gamma'_{j'\ell,r}}\right\}_{m=1}^{M_{\ell}}\right) & \text{for }\ell=1,\cdots,\ell_{1}    \\
             \text{N}\left(\eta'_{j'\ell},\sigma'_{j'\ell}\right)                                                                       & \text{for }\ell=\ell_{1}+1,\cdots,p
           \end{cases} \\
    z_{ij\ell}\mid\boldsymbol{\Xi}_{-z_{ij\ell}}
     & \sim\text{Bernoulli}\left(p_{ij\ell}\right),
  \end{aligned}
\end{equation*}
where
\begin{align*}
  \alpha'_{\ell m}   & = \mu_{\ell m}+\sum_{j'=1}^{N_{\max}}\mathbb{I}\left(y_{j'\ell}=m\right)+\sum_{i=1}^{k}\sum_{j=1}^{n_{i}}z_{ij\ell}\mathbb{I}\left(x_{ij\ell}=m\right)                                                                                                                                                  \\
  b'^{\sigma}_{\ell} & = b^{\sigma}_{\ell}+\frac{1}{2}\left(\sum_{j'}\left(y_{j'\ell}-\eta_{\ell}\right)^{2}+\sum_{i,j}z_{ij\ell}\left(x_{ij\ell}-\eta_{\ell}\right)^{2}\right)                                                                                                                                                \\
  \nu'_{ij,j'}       & = \prod_{\ell=1}^{\ell_{1}}\phi_{\ell}^{\frac{1}{\tau_{\ell}}\left(1-z_{ij\ell}\right)\mathbb{I}\left(\lvert y_{j'\ell}-x_{ij\ell}\rvert\leq\varepsilon_{\ell}\right)}\prod_{\ell=\ell_{1}+1}^{p}\exp\left(-\frac{\left(1-z_{ij\ell}\right)}{2\varepsilon_{\ell}}\left(x_{ij\ell}-y_{j'\ell}\right)^{2}\right) \\
  \gamma'_{j'\ell,m} & = \phi_{\ell}^{\frac{1}{\tau_{\ell}}\sum_{\{i,j:\lvert x_{ij\ell}-m\rvert\leq\varepsilon_{\ell}\}}\mathbb{I}\left(\lambda_{ij}=j'\right)\left(1-z_{ij\ell}\right)}\theta_{\ell m}                                                                                                                       \\
  \eta'_{j'\ell}     & = \frac{\eta_{\ell}\varepsilon_{\ell}+\sigma_{\ell}\sum_{i,j}\mathbb{I}\left(\lambda_{ij}=j'\right)\left(1-z_{ij\ell}\right)x_{ij\ell}}{\varepsilon_{\ell}+\sigma_{\ell}\sum_{i,j}\mathbb{I}\left(\lambda_{ij}=j'\right)\left(1-z_{ij\ell}\right)}                                                      \\
  \sigma'_{j'\ell}   & = \frac{\varepsilon_{\ell}\sigma_{\ell}}{\varepsilon_{\ell}+\sigma_{\ell}\sum_{i,j}\mathbb{I}\left(\lambda_{ij}=j'\right)\left(1-z_{ij\ell}\right)}
\end{align*}
and
\begin{align*}
  p_{ij\ell}=
  \begin{cases}
    \frac{
    \beta_{\ell}\prod_{m}\theta_{\ell m}^{\mathbb{I}\left(x_{ij\ell}=m\right)}
    }{
    \beta_{\ell}\prod_{m}\theta_{\ell m}^{\mathbb{I}\left(x_{ij\ell}=m\right)}+c_{ij\ell}^{-1}\left(1-\beta_{\ell}\right)\left(\phi_{\ell}^{\frac{1}{\tau_{\ell}}\mathbb{I}\left(\lvert y_{j'\ell}-x_{ij\ell}\rvert\leq\varepsilon_{\ell}\right)\mathbb{I}\left(\lambda_{ij}=j'\right)}\right)
    } & \text{if }\ell=1,\cdots,\ell_{1}    \\
    \frac{
      \beta_{\ell}\frac{1}{\sqrt{2\pi\sigma_{\ell}}}\exp\left(-\frac{1}{2\sigma_{\ell}}\left(x_{ij\ell}-\eta_{\ell}\right)^{2}\right)
    }{
      \splitfrac{
        \beta_{\ell}\frac{1}{\sqrt{2\pi\sigma_{\ell}}}\exp\left(-\frac{1}{2\sigma_{\ell}}\left(x_{ij\ell}-\eta_{\ell}\right)^{2}\right)
      }{
        +\left(1-\beta_{\ell}\right)\frac{1}{\sqrt{2\pi\varepsilon_{\ell}}}\exp\left(-\frac{1}{2\varepsilon_{\ell}}\left(x_{ij\ell}-\mathbb{I}\left(\lambda_{ij}=j'\right)y_{j'\ell}\right)^{2}\right)
      }
    } & \text{if }\ell=\ell_{1}+1,\cdots,p
  \end{cases}
\end{align*}
for
\begin{equation*}
  c_{ij\ell}=\sum_{m=1}^{M_{\ell}}\phi_{\ell}^{\frac{1}{\tau_{\ell}}\mathbb{I}\left(\lvert x_{ij\ell}-m\rvert\leq\varepsilon_{\ell}\right)}.
\end{equation*}

\subsection{Detailed Derivation of the Full Conditional Distributions}
\label{app:detailed-derivation-full-conditionals}

Based on the definition of the model in Equation~\eqref{eqn:chomper-mdl-appendix}, we can find the joint posterior distribution of the model:
\begin{equation}
  \label{eqn:joint-posterior}
  \begin{aligned}
    \pi\left(\boldsymbol{\Lambda},\boldsymbol{y},\boldsymbol{z},\boldsymbol{\theta},\boldsymbol{\beta},\boldsymbol{\eta},\boldsymbol{\sigma}\mid\boldsymbol{x}\right)
    \propto & \prod_{i,j,\ell,m}^{\ell=1,\cdots,\ell_{1}}\left[\theta_{\ell m}^{\mathbb{I}\left(x_{ij\ell}=m\right)}\right]^{z_{ij\ell}}
    \times \prod_{j',\ell,m}^{\ell=1,\cdots,\ell_{1}}\theta_{\ell m}^{\mathbb{I}\left(y_{j'\ell}=m\right)}
    \times \prod_{\ell ,m}^{\ell=1,\cdots,\ell_{1}}\theta_{\ell m}^{\mu_{\ell m}-1}                                                                                                                                                \\
            & \times \prod_{i,j,\ell}^{\ell=1,\cdots,\ell_{1}}\tilde{\phi}_{\ell}\left(x_{ij\ell},y_{\lambda_{ij}\ell},\tau_{\ell},\varepsilon_{\ell}\right)^{1-z_{ij\ell}}                                                        \\
            & \times \prod_{i,j,\ell}^{\ell=1,\cdots,p}\beta_{\ell}^{z_{ij\ell}}\left(1-\beta_{\ell}\right)^{1-z_{ij\ell}}
    \times \prod_{\ell}^{\ell=1,\cdots,p}\beta_{\ell}^{a_{\ell}-1}\left(1-\beta_{\ell}\right)^{b_{\ell}-1}                                                                                                                         \\
            & \times \prod_{i,j,\ell}^{\ell=\ell_{1}+1,\cdots,p}\left[\frac{1}{\sqrt{2\pi\sigma_{\ell}}}\exp\left(-\frac{1}{2\sigma_{\ell}}\left(x_{ij\ell}-\eta_{\ell}\right)^{2}\right)\right]^{z_{ij\ell}}                      \\
            & \times \prod_{i,j,\ell}^{\ell=\ell_{1}+1,\cdots,p}\left[\frac{1}{\sqrt{2\pi\varepsilon_{\ell}}}\exp\left(-\frac{1}{2\varepsilon_{\ell}}\left(x_{ij\ell}-y_{\lambda_{ij}\ell}\right)^{2}\right)\right]^{1-z_{ij\ell}} \\
            & \times \prod_{j',\ell}^{\ell=\ell_{1}+1,\cdots,p}\frac{1}{\sqrt{2\pi\sigma_{\ell}}}\exp\left(-\frac{1}{2\sigma_{\ell}}\left(y_{j'\ell}-\eta_{\ell}\right)^{2}\right)                                                 \\
            & \times \prod_{\ell}^{\ell=\ell_{1}+1,\cdots,p}\sigma_{\ell}^{-a^{\sigma}_{\ell}-1}\exp\left(-\frac{b^{\sigma}_{\ell}}{\sigma_{\ell}}\right)
  \end{aligned}
\end{equation}
for $\boldsymbol{\Lambda}=\left\{\lambda_{ij}\right\}_{i,j}$, $\boldsymbol{y}=\left(\boldsymbol{y}_{1},\cdots,\boldsymbol{y}_{N_{\max}}\right)$, $\boldsymbol{z}=\left\{z_{ij\ell}\right\}_{i,j,\ell}$, $\boldsymbol{\theta}=\left(\boldsymbol{\theta}_{1},\cdots,\boldsymbol{\theta}_{\ell_{1}}\right)$, $\boldsymbol{\beta}=\left(\beta_{1},\cdots,\beta_{p}\right)$, $\boldsymbol{\eta}=\left(\eta_{\ell_{1}+1},\cdots,\eta_{p}\right)$, $\boldsymbol{\sigma}=\left(\sigma_{\ell_{1}+1},\cdots,\sigma_{p}\right)$, and $\boldsymbol{x}=\left\{\boldsymbol{x}_{ij}\right\}_{i,j}$.
Hereafter, we will derive the full conditional distributions of the parameters and the latent variables from Equation~\eqref{eqn:joint-posterior} so that we can draw inference of the linkage structure, $\boldsymbol{\Lambda}$.

Before the derivation, we require a trick for bringing out $\lambda_{ij}$ from the subscripts of $y_{\lambda_{ij}\ell}$.
Note that $y_{\lambda_{ij}\ell}=\mathbb{I}\left(\lambda_{ij}=j'\right)y_{j'\ell}$ so that $x_{ij\ell}-y_{\lambda_{ij}\ell}=x_{ij\ell}$ for all $i$ and $j$ where $\lambda_{ij}\neq j'$, which implies that we can bring out the terms related to $y_{j'\ell}$ in the following manner:
\begin{equation}
  \label{eqn:trick_for_extracting_lambda}
  \begin{aligned}
    \prod_{i,j}\left[\frac{1}{\sqrt{2\pi\varepsilon_{\ell}}}\exp\left(-\frac{1}{2\varepsilon_{\ell}}\left(x_{ij\ell}-y_{\lambda_{ij}\ell}\right)^{2}\right)\right]^{z_{ij\ell}}
    \propto & \prod_{i,j}\left[\exp\left(-\frac{1}{2\varepsilon_{\ell}}\left(x_{ij\ell}-y_{\lambda_{ij}\ell}\right)^{2}\right)\right]^{z_{ij\ell}}                               \\
    =       & \prod_{\{i,j:\lambda_{ij}=j'\}}\left[\exp\left(-\frac{1}{2\varepsilon_{\ell}}\left(x_{ij\ell}-y_{j'\ell}\right)^{2}\right)\right]^{z_{ij\ell}}                     \\
            & \times\prod_{\{i,j:\lambda_{ij}\neq j'\}}\left[\exp\left(-\frac{1}{2\varepsilon_{\ell}}x_{ij\ell}^{2}\right)\right]^{z_{ij\ell}}                                   \\
    \propto & \prod_{i,j}\left[\exp\left(-\frac{1}{2\varepsilon_{\ell}}\left(x_{ij\ell}-y_{j'\ell}\right)^{2}\right)\right]^{z_{ij\ell}\mathbb{I}\left(\lambda_{ij}=j'\right)}.
  \end{aligned}
\end{equation}

For a simpler notation, let $\boldsymbol{\Xi}_{-\xi}$ be a set of parameters and data excluding $\xi$.
If we only consider the kernel related to $\boldsymbol{\theta}_{\ell}$, we observe
\begin{equation*}
  \begin{aligned}
    \pi\left(\boldsymbol{\theta}_{\ell}\mid\boldsymbol{\Xi}_{-\boldsymbol{\theta}_{\ell}}\right)
    \propto & \prod_{i,j,m}\theta_{\ell m}^{z_{ij\ell}\mathbb{I}\left(x_{ij\ell}=m\right)}
    \times \prod_{j',m}\theta_{\ell m}^{\mathbb{I}\left(y_{j'\ell}=m\right)}
    \times \prod_{m}\theta_{\ell m}^{\mu_{\ell m}-1}                                                                                                                        \\
    \propto & \prod_{m=1}^{M_{\ell}}\theta_{\ell m}^{\mu_{\ell m}+\sum_{j'}\mathbb{I}\left(y_{j'\ell}=m\right)+\sum_{i,j}z_{ij\ell}\mathbb{I}\left(x_{ij\ell}=m\right)-1}.
  \end{aligned}
\end{equation*}
Thus, $\boldsymbol{\theta}_{\ell}\mid\boldsymbol{\Xi}_{-\boldsymbol{\theta}_{\ell}}\sim\text{Dirichlet}\left(\left\{\alpha'_{\ell m}\right\}_{m=1}^{M_{\ell}}\right)$ for
\begin{equation*}
  \alpha'_{\ell m}=\mu_{\ell m}+\sum_{j'=1}^{N_{\max}}\mathbb{I}\left(y_{j'\ell}=m\right)+\sum_{i=1}^{k}\sum_{j=1}^{n_{i}}z_{ij\ell}\mathbb{I}\left(x_{ij\ell}=m\right).
\end{equation*}

Next, if one collects the terms related to $\eta_{\ell}$ and $\sigma_{\ell}$, it is observed that
\begin{equation*}
  \begin{aligned}
    \pi\left(\eta_{\ell},\sigma_{\ell}\mid\boldsymbol{\Xi}_{-\eta_{\ell},-\sigma_{\ell}}\right)
    \propto & \prod_{j'}\frac{1}{\sqrt{2\pi\sigma_{\ell}}}\exp\left(-\frac{1}{2\sigma_{\ell}}\left(y_{j'\ell}-\eta_{\ell}\right)^{2}\right)                                                                                                                                   \\
            & \times \prod_{i,j}\frac{1}{\sqrt{2\pi\sigma_{\ell}}}\exp\left(-\frac{1}{2\sigma_{\ell}}\left(x_{ij\ell}-\eta_{\ell}\right)^{2}\right)\times \sigma_{\ell}^{-a^{\sigma}_{\ell}-1}\exp\left(-\frac{b^{\sigma}_{\ell}}{\sigma_{\ell}}\right)                       \\
    \propto & \sigma_{\ell}^{-\frac{1}{2}N_{\max}}\exp\left(-\frac{1}{2\sigma_{\ell}}\sum_{j'}\left(\eta_{\ell}-y_{j'\ell}\right)^{2}\right)                                                                                                                                  \\
            & \times \sigma_{\ell}^{-\frac{1}{2}\sum_{i,j}z_{ij\ell}}\exp\left(-\frac{1}{2\sigma_{\ell}}\sum_{i,j}z_{ij\ell}\left(x_{ij\ell}-\eta_{\ell}\right)^{2}\right)\times\sigma_{\ell}^{-a^{\sigma}_{\ell}-1}\exp\left(-\frac{b^{\sigma}_{\ell}}{\sigma_{\ell}}\right) \\
    \propto & \sigma_{\ell}^{-\left(a^{\sigma}_{\ell}+\frac{1}{2}\left\{N_{\max}+\sum_{i,j}z_{ij\ell}\right\}\right)-1}                                                                                                                                                       \\
            & \times \exp\left(-\frac{1}{\sigma_{\ell}}\left[b^{\sigma}_{\ell}+\frac{1}{2}\left\{\sum_{j'}\left(\eta_{\ell}-y_{j'\ell}\right)^{2}+\sum_{i,j}z_{ij\ell}\left(\eta_{\ell}-x_{ij\ell}\right)^{2}\right\}\right]\right).
  \end{aligned}
\end{equation*}
Therefore,
\begin{equation*}
  \begin{aligned}
    \eta_{\ell}\mid\boldsymbol{\Xi}_{-\eta_{\ell}}
     & \sim\text{N}\left(\frac{\sum_{j'}y_{j'\ell}+\sum_{i,j}z_{ij\ell}x_{ij\ell}}{N_{\max}+\sum_{i,j}z_{ij\ell}},\frac{\sigma_{\ell}}{N_{\max}+\sum_{i,j}z_{ij\ell}}\right)\text{ and}                                                                                      \\
    \sigma_{\ell}\mid\boldsymbol{\Xi}_{-\sigma_{\ell}}
     & \sim\text{Inverse-Gamma}\left(a^{\sigma}_{\ell}+\frac{1}{2}\left[N_{\max}+\sum_{i,j}z_{ij\ell}\right],b^{\sigma}_{\ell}+\frac{1}{2}\left[\sum_{j'}\left(y_{j'\ell}-\eta_{\ell}\right)^{2}+\sum_{i,j}z_{ij\ell}\left(x_{ij\ell}-\eta_{\ell}\right)^{2}\right]\right).
  \end{aligned}
\end{equation*}

The full conditional distribution of $\beta_{\ell}$ is obtained similarly:
\begin{equation*}
  \beta_{\ell}\mid\boldsymbol{\Xi}_{-\beta_{\ell}}
  \sim\text{Beta}\left(a_{\ell}+\sum_{i=1}^{k}\sum_{j=1}^{n_{i}}z_{ij\ell},b_{\ell}+\sum_{i=1}^{k}\sum_{j=1}^{n_{i}}\left(1-z_{ij\ell}\right)\right).
\end{equation*}

The next step is deriving the full conditional distribution of $\lambda_{ij}$, where the trick in Equation~\eqref{eqn:trick_for_extracting_lambda} is needed.
If we take out $\lambda_{ij}$ from the subscripts of $y_{\lambda_{ij}\ell}$ and consider the terms related to $\lambda_{ij}$ from Equation~\eqref{eqn:joint-posterior}, we find
\begin{equation*}
  \begin{aligned}
    \pi\left(\lambda_{ij}\mid\boldsymbol{\Xi}_{-\lambda_{ij}}\right)
    \propto & \left[\prod_{\ell=1}^{\ell_{1}}\phi_{\ell}^{\frac{1}{\tau_{\ell}}\mathbb{I}\left(\lvert y_{j'\ell}-x_{ij\ell}\rvert\leq\varepsilon_{\ell}\right)\left(1-z_{ij\ell}\right)}
    \times \prod_{\ell=\ell_{1}+1}^{p}\exp\left(-\frac{\left(1-z_{ij\ell}\right)}{2\varepsilon_{\ell}}\left(x_{ij\ell}-y_{j'\ell}\right)^{2}\right)\right]^{\mathbb{I}\left(\lambda_{ij}=j'\right)}.
  \end{aligned}
\end{equation*}
That is, the full conditional distribution of $\lambda_{ij}$ follows a multinomial distribution such that
\begin{equation*}
  \lambda_{ij}\mid\boldsymbol{\Xi}_{-\lambda_{ij}}
  \sim\text{MN}\left(1,\left\{\frac{\nu'_{ij,j'}}{\sum_{u=1}^{N_{\max}}\nu'_{ij,u}}\right\}_{j'=1}^{N_{\max}}\right),
\end{equation*}
where
\begin{equation*}
  \nu'_{ij,j'}=\prod_{\ell=1}^{\ell_{1}}\phi_{\ell}^{\frac{1}{\tau_{\ell}}\mathbb{I}\left(\lvert y_{j'\ell}-x_{ij\ell}\rvert\leq\varepsilon_{\ell}\right)\left(1-z_{ij\ell}\right)}\prod_{\ell=\ell_{1}+1}^{p}\exp\left(-\frac{\left(1-z_{ij\ell}\right)}{2\varepsilon_{\ell}}\left(x_{ij\ell}-y_{j'\ell}\right)^{2}\right).
\end{equation*}

Applying the trick in~\eqref{eqn:trick_for_extracting_lambda} again and considering the terms that contain $y_{j'\ell}$, we obtain
\begin{equation*}
  \pi\left(y_{j'\ell}\mid\boldsymbol{\Xi}_{-y_{j'\ell}}\right)
  \propto \phi_{\ell}^{\frac{1}{\tau_{\ell}}\sum_{\left\{i,j:\lvert x_{ij\ell}-y\rvert\leq\varepsilon_{\ell}\right\}}\mathbb{I}\left(\lambda_{ij}=j'\right)\left(1-z_{ij\ell}\right)}\prod_{m=1}^{M_{\ell}}\theta_{\ell m}^{\mathbb{I}\left(y_{j'\ell}=y\right)}
\end{equation*}
for $\ell=1,\cdots,\ell_{1}$, and
\begin{equation*}
  \begin{aligned}
    \pi\left(y_{j'\ell}\mid\boldsymbol{\Xi}_{-y_{j'\ell}}\right)
    \propto & \prod_{i,j}\exp\left(-\frac{\left(1-z_{ij\ell}\right)}{2\varepsilon_{\ell}}\left(y_{j'\ell}^{2}\mathbb{I}\left(\lambda_{ij}=j'\right)-2y_{j'\ell}\mathbb{I}\left(\lambda_{ij}=j'\right)x_{ij\ell}\right)\right) \\
            & \times \exp\left(-\frac{1}{2\sigma_{\ell}}\left(y_{j'\ell}^{2}-2y_{j'\ell}\eta_{\ell}\right)^{2}\right)
  \end{aligned}
\end{equation*}
for $\ell=\ell_{1}+1,\cdots,p$, which yields
\begin{equation*}
  \begin{aligned}
    \pi\left(y_{j'\ell}\mid\boldsymbol{\Xi}_{-y_{j'\ell}}\right)
    \propto & \exp\left(-\frac{1}{2\varepsilon_{\ell}\sigma_{\ell}}y_{j'\ell}^{2}\left[\varepsilon_{\ell}+\sigma_{\ell}\sum_{i,j}\mathbb{I}\left(\lambda_{ij}=j'\right)\left(1-z_{ij\ell}\right)\right]\right)                                          \\
            & \times \exp\left(-\frac{1}{2\varepsilon_{\ell}\sigma_{\ell}}\left(-2y_{j'\ell}\right)\left[\eta_{\ell}\varepsilon_{\ell}+\sigma_{\ell}\sum_{i,j}\mathbb{I}\left(\lambda_{ij}=j'\right)\left(1-z_{ij\ell}\right)x_{ij\ell}\right]\right).
  \end{aligned}
\end{equation*}
Therefore, for the continuous fields, the full conditional distribution of $y_{j'\ell}$ follows a normal distribution:
\begin{equation*}
  \pi\left(y_{j'\ell}\mid\boldsymbol{\Xi}_{-y_{j'\ell}}\right)
  \propto\text{N}\left(\frac{\eta_{\ell}\varepsilon_{\ell}+\sigma_{\ell}\sum_{i,j}\mathbb{I}\left(\lambda_{ij}=j'\right)\left(1-z_{ij\ell}\right)x_{ij\ell}}{\varepsilon_{\ell}+\sigma_{\ell}\sum_{i,j}\mathbb{I}\left(\lambda_{ij}=j'\right)\left(1-z_{ij\ell}\right)},
  \frac{\varepsilon_{\ell}\sigma_{\ell}}{\varepsilon_{\ell}+\sigma_{\ell}\sum_{i,j}\mathbb{I}\left(\lambda_{ij}=j'\right)\left(1-z_{ij\ell}\right)}\right).
\end{equation*}
Hence,
\begin{equation*}
  y_{j'\ell}\mid\boldsymbol{\Xi}_{-y_{j'\ell}}
  \sim
  \begin{cases}
    \text{MN}\left(\left\{\frac{\gamma'_{j'\ell,m}}{\sum_{r=1}^{M_{\ell}}\gamma'_{j'\ell,r}}\right\}_{m=1}^{M_{\ell}}\right) & \text{ for }\ell=1,\cdots,\ell_{1}     \\
    \text{N}\left(\eta'_{j'\ell},\sigma'_{j'\ell}\right)                                                                     & \text{ for }\ell=\ell_{1}+1,\cdots,p,
  \end{cases}
\end{equation*}
where
\begin{equation*}
  \begin{aligned}
    \gamma'_{j'\ell,m} & = \phi_{\ell}^{\frac{1}{\tau_{\ell}}\sum_{\left\{i,j:\lvert x_{ij\ell}-m\rvert\leq\varepsilon_{\ell}\right\}}\mathbb{I}\left(\lambda_{ij}=j'\right)\left(1-z_{ij\ell}\right)}\theta_{\ell m}                                                       \\
    \eta'_{j'\ell}     & = \frac{\eta_{\ell}\varepsilon_{\ell}+\sigma_{\ell}\sum_{i,j}\mathbb{I}\left(\lambda_{ij}=j'\right)\left(1-z_{ij\ell}\right)x_{ij\ell}}{\varepsilon_{\ell}+\sigma_{\ell}\sum_{i,j}\mathbb{I}\left(\lambda_{ij}=j'\right)\left(1-z_{ij\ell}\right)} \\
    \sigma'_{j'\ell}   & = \frac{\varepsilon_{\ell}\sigma_{\ell}}{\varepsilon_{\ell}+\sigma_{\ell}\sum_{i,j}\mathbb{I}\left(\lambda_{ij}=j'\right)\left(1-z_{ij\ell}\right)}.
  \end{aligned}
\end{equation*}

Finally, as
\begin{equation*}
  \begin{aligned}
    \pi\left(z_{ij\ell}\mid\boldsymbol{\Xi}_{-z_{ij\ell}}\right)
    \propto & \prod_{m=1}^{M_{\ell}}\left[\theta_{\ell m}^{\mathbb{I}\left(x_{ij\ell}=m\right)}\right]^{z_{ij\ell}}
    \times\tilde{\phi}_{\ell}\left(x_{ij\ell},y_{\lambda_{ij}\ell},\tau_{\ell},\varepsilon_{\ell}\right)^{1-z_{ij\ell}}
    \times\beta_{\ell}^{z_{ij\ell}}\left(1-\beta_{\ell}\right)^{1-z_{ij\ell}}                                                                                                                                         \\
            & \times \left[\frac{1}{\sqrt{2\pi\sigma_{\ell}}}\exp\left(-\frac{1}{2\sigma_{\ell}}\left(x_{ij\ell}-\eta_{\ell}\right)^{2}\right)\right]^{z_{ij\ell}}                                                    \\
            & \times \left[\frac{1}{\sqrt{2\pi\varepsilon_{\ell}}}\exp\left(-\frac{1}{2\varepsilon_{\ell}}\left(x_{ij\ell}-\mathbb{I}\left(\lambda_{ij}=j'\right)y_{j'\ell}\right)^{2}\right)\right]^{1-z_{ij\ell}},
  \end{aligned}
\end{equation*}
we obtain
\begin{equation*}
  z_{ij\ell}\mid\boldsymbol{\Xi}_{-z_{ij\ell}}\sim\text{Bernoulli}\left(\rho_{ij\ell}\right),
\end{equation*}
where
\begin{equation*}
  \rho_{ij\ell}=
  \begin{cases}
    \frac{\beta_{\ell}\prod_{m}\theta_{\ell m}^{\mathbb{I}\left(x_{ij\ell}=m\right)}}
    {\beta_{\ell}\prod_{m}\theta_{\ell m}^{\mathbb{I}\left(x_{ij\ell}=m\right)}+\frac{1}{c_{ij\ell}}\left(1-\beta_{\ell}\right)\left(\phi_{\ell}^{\frac{1}{\tau_{\ell}}\mathbb{I}\left(\lvert y_{j'\ell}-x_{ij\ell}\rvert\leq\varepsilon_{\ell}\right)\mathbb{I}\left(\lambda_{ij}=j'\right)}\right)} & \text{ if }\ell=1,\cdots,\ell_{1}    \\
    \frac{\beta_{\ell}\frac{1}{\sqrt{2\pi\sigma_{\ell}}}\exp\left(-\frac{1}{2\sigma_{\ell}}\left(x_{ij\ell}-\eta_{\ell}\right)^{2}\right)}
    {\splitfrac{\beta_{\ell}\frac{1}{\sqrt{2\pi\sigma_{\ell}}}\exp\left(-\frac{1}{2\sigma_{\ell}}\left(x_{ij\ell}-\eta_{\ell}\right)^{2}\right)}
    {+\left(1-\beta_{\ell}\right)\frac{1}{\sqrt{2\pi\varepsilon_{\ell}}}\exp\left(-\frac{1}{2\varepsilon_{\ell}}\left(x_{ij\ell}-\mathbb{I}\left(\lambda_{ij}=j'\right)y_{j'\ell}\right)^{2}\right)}}                                                                                                 & \text{ if }\ell=\ell_{1}+1,\cdots,p
  \end{cases}
\end{equation*}
for
\begin{equation*}
  c_{ij\ell}=\sum_{m=1}^{M_{\ell}}\phi_{\ell}^{\frac{1}{\tau_{\ell}}\mathbb{I}\left(\lvert x_{ij\ell}-m\rvert\leq\varepsilon_{\ell}\right)}.
\end{equation*}

\section{Hybrid MCMC with Split-Merge Algorithm}
\label{app:mcmc-split-merge}

This section explains the hybrid MCMC algorithm that applies a split-merge update for the linkage parameter inside the Gibbs sampler.
As full conditional distributions for all parameters are available if conjugate priors are used as in Appendix~\ref{app:full-conditional}, the explanation uses the CHOMPER model defined in Equation (2) in the main manuscript, assuming $\pi\left(\boldsymbol{\eta}_{\ell}\right)$ is a conjugate prior for $f_{\ell}\left(x_{ij\ell}\mid\boldsymbol{\eta}_{\ell}\right)$ for $\ell=1,\cdots,p$.
With $N_{\text{MCMC}}$ representing the number of MCMC iterations and $N_{\text{Split-Merge}}$ the number of split-merge updates within each MCMC iteration, the algorithm is as follows:

\begin{algorithm}[H]
  \caption{Hybrid MCMC with Split-Merge Algorithm for CHOMPER}
  \label{alg:mcmc}
  
  \begin{algorithmic}[1]
    
    \Procedure{Hybrid MCMC}{$\boldsymbol{x}$}
    \State Initialize $\boldsymbol{\xi}=\left\{\boldsymbol{\Lambda},\boldsymbol{y},\boldsymbol{z},\boldsymbol{\beta},\boldsymbol{\eta}\right\}$.
    \For{$m=1,\cdots,N_{\text{MCMC}}$}
    \For{$s=1,\cdots,N_{\text{Split-Merge}}$}
    \State Randomly select a pair of records, $\boldsymbol{x}_{i_{1}j_{1}}$ and $\boldsymbol{x}_{i_{2}j_{2}}$, from $\boldsymbol{x}$.
    \If{$\lambda_{i_{1}j_{1}}=\lambda_{i_{2}j_{2}}$}\Comment{Split step}
    \State Collect a set $\boldsymbol{\mathcal{C}}$ of $\boldsymbol{x}_{ij}$ such that $\lambda_{ij}=\lambda_{i_{1}j_{1}}=\lambda_{i_{2}j_{2}}$.
    \State Randomly sample $j'$ from $\left\{1,\cdots,N_{\max}\right\}\setminus\boldsymbol{\Lambda}$.
    \State Randomly assign the records in $\boldsymbol{\mathcal{C}}$ to either $j'$ or $\lambda_{ij}$.
    \State Replace $\boldsymbol{y}_{j'}$ and $\boldsymbol{y}_{\lambda_{ij}}$ with one of the assigned records.
    \Else\Comment{Merge step}
    \State Collect a set $\boldsymbol{\mathcal{C}}$ of $\boldsymbol{x}_{ij}$ such that $\lambda_{ij}=\lambda_{i_{1}j_{1}}$ or $\lambda_{ij}=\lambda_{i_{2}j_{2}}$.
    \State Randomly sample $j'$ from $\left\{\lambda_{i_{1}j_{1}},\lambda_{i_{2}j_{2}}\right\}$.
    \State Assign the records in $\boldsymbol{\mathcal{C}}$ to $j'$.
    \State Replace $\boldsymbol{y}_{j'}$ with one of the assigned records.
    \EndIf
    \State Generate $z_{ij\ell}^{\star}$ from the corresponding full conditional distribution.
    \State Accept proposal if $r\leq\rho$ for an acceptance probability $\rho$ and $r\sim\text{Unif}\left(0,1\right)$.
    \EndFor
    \State Sample $\boldsymbol{\beta}$ and $\boldsymbol{\eta}$ from the corresponding full conditional distributions.
    \EndFor
    \EndProcedure
    
  \end{algorithmic}
\end{algorithm}
The acceptance probability $\rho$ for the split-merge steps is defined as:
\begin{equation*}
  \label{eqn:acceptance-probability}
  \rho=
  \begin{cases}
    \left(\frac{1}{2}\right)^{2-n_{\boldsymbol{\mathcal{C}}}}\frac{\pi\left(\boldsymbol{\Lambda}^{\star},\boldsymbol{y}^{\star},\boldsymbol{z}^{\star},\boldsymbol{\beta},\boldsymbol{\eta}\mid\boldsymbol{x}\right)}{\pi\left(\boldsymbol{\Lambda},\boldsymbol{y},\boldsymbol{z},\boldsymbol{\beta},\boldsymbol{\eta}\mid\boldsymbol{x}\right)}, & \text{for split}  \\
    \left(\frac{1}{2}\right)^{n_{\boldsymbol{\mathcal{C}}}-2}\frac{\pi\left(\boldsymbol{\Lambda}^{\star},\boldsymbol{y}^{\star},\boldsymbol{z}^{\star},\boldsymbol{\beta},\boldsymbol{\eta}\mid\boldsymbol{x}\right)}{\pi\left(\boldsymbol{\Lambda},\boldsymbol{y},\boldsymbol{z},\boldsymbol{\beta},\boldsymbol{\eta}\mid\boldsymbol{x}\right)}, & \text{for merge,}
  \end{cases}
\end{equation*}
where $n_{\boldsymbol{\mathcal{C}}}$ is the number of records in $\boldsymbol{\mathcal{C}}$ and $\boldsymbol{\Lambda}^{\star},\boldsymbol{z}^{\star}$, and $\boldsymbol{y}^{\star}$ are proposal samples from a split-merge update.

\section{Variational Inference}
\label{app:vi}

The proposed EVIL algorithm uses Coordinate Ascent Variational Inference (CAVI) to approximate the posterior distribution of the model for each member of a generation.
The variational factors are obtained by calculating the right-hand side of the following equation \citep{bishop2006,blei2017}:
\begin{equation*}
  q_{j'}^{\star}\left(\xi_{j'}\right)\propto\exp\left(\mathbb{E}_{-j'}\left[\log\pi\left(\xi_{j'}\mid\boldsymbol{\Xi}_{-j'},\boldsymbol{x}\right)\right]\right)
\end{equation*}
for
\begin{equation*}
  \mathbb{E}_{-j'}\left[\pi\left(\xi_{j'}\mid\boldsymbol{\Xi}_{-j'},\boldsymbol{x}\right)\right]
  =\int\cdots\int\log\pi\left(\xi_{j'}\mid\boldsymbol{\Xi}_{-j'},\boldsymbol{x}\right)\prod_{i\neq j'}q_{i}\left(\xi_{i}\right)d\xi_{i}.
\end{equation*}

\subsection{Variational Factors}
\label{app:vi-factors}

Note that the full conditional distributions have explicit distributional forms.
So, the expectation of the log of the full conditional distribution can be calculated analytically.
We obtain the variational factors as follows:
\begin{equation*}
  \begin{aligned}
    q_{\boldsymbol{\theta}_{\ell}}\left(\boldsymbol{\theta}_{\ell}\right)
     & \sim\text{Dirichlet}\left(\left\{\alpha_{\ell m}\right\}_{m=1}^{M_{\ell}}\right)                                                                                                                                                                                                                      \\
    q_{\eta_{\ell}}\left(\eta_{\ell}\right)
     & \sim\text{N}\left(\frac{\sum_{j'=1}^{N_{\max}}\mathbb{E}_{-\eta_{\ell}}\left[y_{j'\ell}\right]+\sum_{i,j}x_{ij\ell}\mathbb{E}_{-\eta_{\ell}}\left[z_{ij\ell}\right]}{N_{\max}+\sum_{i,j}\mathbb{E}_{-\eta_{\ell}}\left[z_{ij\ell}\right]},
    \frac{\mathbb{E}_{-\eta_{\ell}}^{-1}\left[\frac{1}{\sigma_{\ell}}\right]}{N_{\max}+\sum_{i,j}\mathbb{E}_{-\eta_{\ell}}\left[z_{ij\ell}\right]}\right)                                                                                                                                                    \\
    q_{\sigma_{\ell}}\left(\sigma_{\ell}\right)
     & \sim\text{Inverse-Gamma}\left(a^{\sigma}_{\ell}+\frac{1}{2}\left[N_{\max}+\sum_{i,j}\mathbb{E}_{-\sigma_{\ell}}\left[z_{ij\ell}\right]\right],\tilde{b}^{\sigma}_{\ell}\right)                                                                                                                        \\
    q_{\beta_{\ell}}\left(\beta_{\ell}\right)
     & \sim\text{Beta}\left(a_{\ell}+\sum_{i,j}\mathbb{E}_{-\beta_{\ell}}\left[z_{ij\ell}\right],b_{\ell}+\sum_{i,j}\sum_{i,j}\left(1-\mathbb{E}_{-\beta_{\ell}}\left[z_{ij\ell}\right]\right)\right)                                                                                                        \\
    q_{\lambda_{ij}}\left(\lambda_{ij}\right)
     & \sim\text{MN}\left(1,\left\{\frac{\nu_{ij,j'}}{\sum_{u=1}^{N_{\max}}\nu_{ij,u}}\right\}_{j'=1}^{N_{\max}}\right)                                                                                                                                                                                      \\
    q_{y_j'\ell}\left(y_{j'\ell}\right)
     & \sim
    \begin{cases}
      \text{MN}\left(1,\left\{\frac{\gamma_{j'\ell,m}}{\sum_{r=1}^{M_{\ell}}\gamma_{j'\ell,r}}\right\}_{m=1}^{M_{\ell}}\right) & \text{for }\ell=1,\cdots,\ell_{1}    \\
      \text{N}\left(\tilde{\eta}_{j'\ell},\tilde{\sigma}_{j'\ell}\right)                                                       & \text{for }\ell=\ell_{1}+1,\cdots,p
    \end{cases} \\
    q_{z_{ij\ell}}\left(z_{ij\ell}\right)
     & \sim\text{Bernoulli}\left(\rho_{ij\ell}\right),
  \end{aligned}
\end{equation*}
where
\begin{equation*}
  \alpha_{\ell m}=\mu_{\ell m}+\sum_{j'=1}^{N_{\max}}\mathbb{E}_{-\boldsymbol{\theta}_{\ell}}\left[\mathbb{I}\left(y_{j'\ell}=m\right)\right]+\sum_{i=1}^{k}\sum_{j=1}^{n_{i}}\mathbb{E}_{-\boldsymbol{\theta}_{\ell}}\left[z_{ij\ell}\right]\mathbb{I}\left(x_{ij\ell}=m\right)
\end{equation*}
\begin{equation*}
  \begin{aligned}
    \tilde{b}^{\sigma}_{\ell}=b^{\sigma}_{\ell}+\frac{1}{2}\Bigg[ & \sum_{j'}\left(\mathbb{E}_{-\sigma_{\ell}}\left[y_{j'\ell}^{2}\right]-2\mathbb{E}_{-\sigma_{\ell}}\left[y_{j'\ell}\right]\mathbb{E}_{-\sigma_{\ell}}\left[\eta_{\ell}\right]+\mathbb{E}_{-\sigma_{\ell}}\left[\eta_{\ell}^{2}\right]\right) \\
                                                                  & +\sum_{i,j}\left(\mathbb{E}_{-\sigma_{\ell}}\left[z_{ij\ell}\right]\left\{x_{ij\ell}^{2}-2x_{ij\ell}\mathbb{E}_{-\sigma_{\ell}}\left[\eta_{\ell}\right]+\mathbb{E}_{-\sigma_{\ell}}\left[\eta_{\ell}^{2}\right]\right\}\right)\Bigg]
  \end{aligned}
\end{equation*}
\begin{equation*}
  \begin{aligned}
    \nu_{ij,j'}=\exp\Bigg[ & \sum_{\ell=1}^{\ell_{1}}\mathbb{E}_{-\lambda_{ij}}\left[1-z_{ij\ell}\right]\mathbb{E}_{-\lambda_{ij}}\left[\mathbb{I}\left(\lvert y_{j'\ell}-x_{ij\ell}\rvert\leq\varepsilon_{\ell}\right)\right]\log\phi_{\ell}^{\frac{1}{\tau_{\ell}}}                           \\
                           & -\sum_{\ell=\ell_{1}+1}^{p}\frac{1}{2\varepsilon_{\ell}}\mathbb{E}_{-\lambda_{ij}}\left[1-z_{ij\ell}\right]\left\{x_{ij\ell}^{2}-2x_{ij\ell}\mathbb{E}_{-\lambda_{ij}}\left[y_{j'\ell}\right]+\mathbb{E}_{-\lambda_{ij}}\left[y_{j'\ell}^{2}\right]\right\}\Bigg]
  \end{aligned}
\end{equation*}
\begin{equation*}
  \gamma_{j'\ell,m}=\phi_{\ell}^{\frac{1}{\tau_{\ell}}\sum_{\left\{i,j:\lvert x_{ij\ell}-m\rvert\leq\varepsilon_{\ell}\right\}}\mathbb{E}_{-y_{j'\ell}}\left[\mathbb{I}\left(\lambda_{ij}=j'\right)\right]\mathbb{E}_{-y_{j'\ell}}\left[1-z_{ij\ell}\right]}\exp\left(\mathbb{E}_{-y_{j'\ell}}\left[\log\theta_{\ell m}\right]\right)
\end{equation*}
\begin{equation*}
  \tilde{\eta}_{j'\ell}=
  \frac{\mathbb{E}_{-y_{j'\ell}}\left[\eta_{\ell}\right]\mathbb{E}_{-y_{j'\ell}}\left[\frac{1}{\sigma_{\ell}}\right]+\frac{1}{\varepsilon_{\ell}}\sum_{i,j}x_{ij\ell}\mathbb{E}_{-y_{j'\ell}}\left[\mathbb{I}\left(\lambda_{ij}=j'\right)\right]\mathbb{E}_{-y_{j'\ell}}\left[1-z_{ij\ell}\right]}
  {\mathbb{E}_{-y_{j'\ell}}\left[\frac{1}{\sigma}_{\ell}\right]+\frac{1}{\varepsilon_{\ell}}\sum_{i,j}\mathbb{E}_{-y_{j'\ell}}\left[\mathbb{I}\left(\lambda_{ij}=j'\right)\right]\mathbb{E}_{-y_{j'\ell}}\left[1-z_{ij\ell}\right]}
\end{equation*}
\begin{equation*}
  \tilde{\sigma}_{j'\ell}=\left(\mathbb{E}_{-y_{j'\ell}}\left[\frac{1}{\sigma}_{\ell}\right]+\frac{1}{\varepsilon_{\ell}}\sum_{i,j}\mathbb{E}_{-y_{j'\ell}}\left[\mathbb{I}\left(\lambda_{ij}=j'\right)\right]\mathbb{E}_{-y_{j'\ell}}\left[1-z_{ij\ell}\right]\right)^{-1}
\end{equation*}
and
\begin{equation*}
  \rho_{ij\ell}=\frac{\rho_{ij\ell ,1}}{\rho_{ij\ell ,1}+\rho_{ij\ell ,2}}
\end{equation*}
for
\begin{equation*}
  \rho_{ij\ell ,1}
  =\begin{cases}
    \exp\left(\mathbb{E}_{-z_{ij\ell}}\left[\log\beta_{\ell}\right]+\sum_{m=1}^{M_{\ell}}\mathbb{I}\left(x_{ij\ell}=m\right)\mathbb{E}_{-z_{ij\ell}}\left[\log\theta_{\ell m}\right]\right) \quad\text{if }\ell=1,\cdots,\ell_{1} \\
    \begin{aligned}
      \exp & \Bigg(\mathbb{E}_{-z_{ij\ell}}\left[\log\beta_{\ell}\right]-\frac{1}{2}\log2\pi-\frac{1}{2}\mathbb{E}_{-z_{ij\ell}}\left[\log\sigma_{\ell}\right]                                                                            \\
           & \quad-\frac{1}{2}\mathbb{E}_{-z_{ij\ell}}\left[\frac{1}{\sigma_{\ell}}\right]\Big(x_{ij\ell}^{2}-2x_{ij\ell}\mathbb{E}_{-z_{ij\ell}}\left[\eta_{\ell}\right]+\mathbb{E}_{-z_{ij\ell}}\left[\eta_{\ell}^{2}\right]\Big)\Bigg) \\
           & \qquad\qquad\qquad\qquad\qquad\qquad\qquad\qquad\qquad\qquad\qquad\qquad\text{if }\ell=\ell_{1}+1,\cdots,p
    \end{aligned}
  \end{cases}
\end{equation*}
and
\begin{equation*}
  \rho_{ij\ell ,2}
  =\begin{cases}
    \begin{aligned}
      \exp & \Bigg(\mathbb{E}_{-z_{ij\ell}}\left[\log\left(1-\beta_{\ell}\right)\right]                                                                                                                                                                                         \\
           & \quad+\mathbb{E}_{-z_{ij\ell}}\left[\mathbb{I}\left(\lvert y_{j'\ell}-x_{ij\ell}\rvert\leq\varepsilon_{\ell}\right)\right]\mathbb{E}_{-z_{ij\ell}}\left[\mathbb{I}\left(\lambda_{ij}=j'\right)\right]\log\phi_{\ell}^{\frac{1}{\tau_{\ell}}}-\log c_{ij\ell}\Bigg) \\
           & \qquad\qquad\qquad\qquad\qquad\qquad\qquad\qquad\qquad\qquad\qquad\qquad\text{if }\ell=1,\cdots,\ell_{1}
    \end{aligned} \\
    \begin{aligned}
      \exp & \Bigg(\mathbb{E}_{-z_{ij\ell}}\left[\log\left(1-\beta_{\ell}\right)\right]-\frac{1}{2}\log2\pi\varepsilon_{\ell}                                                                             \\
           & \quad-\frac{1}{2\varepsilon_{\ell}}\Big(x_{ij\ell}^{2}-2x_{ij\ell}\mathbb{E}_{-z_{ij\ell}}\left[\mathbb{I}\left(\lambda_{ij}=j'\right)\right]\mathbb{E}_{-z_{ij\ell}}\left[y_{j'\ell}\right] \\
           & \qquad\qquad\qquad\qquad\qquad\qquad\quad+\mathbb{E}_{-z_{ij\ell}}\left[\mathbb{I}\left(\lambda_{ij}=j'\right)\right]\mathbb{E}_{-z_{ij\ell}}\left[y_{j'\ell}^{2}\right]\Big)\Bigg)          \\
           & \qquad\qquad\qquad\qquad\qquad\qquad\qquad\qquad\qquad\qquad\qquad\qquad\text{if }\ell=\ell_{1}+1,\cdots,p.
    \end{aligned}
  \end{cases}
\end{equation*}

\subsection{Evidence Lower Bound}
\label{app:vi-evidence}

The Evidence Lower Bound (ELBO) is defined as the difference between the expected value of the log of the current variational factor and the expectation of the log of the joint distribution of the parameters of interest, $\boldsymbol{\xi}$.
Defining $q^{\left(t\right)}\left(\boldsymbol{\xi}\right)$ as the variational factor at $t$th iteration, we can derive the ELBO as:
\begin{equation*}
  \label{eqn:elbo}
  \begin{aligned}
    \text{ELBO}\left(q^{\left(t\right)}\right)
    = & \mathbb{E}_{q^{\left(t\right)}}\left[\log\pi\left(\boldsymbol{\xi},\boldsymbol{x}\right)\right]-\mathbb{E}_{q^{\left(t\right)}}\left[\log q^{\left(t\right)}\left(\boldsymbol{\xi}\right)\right]                                         \\
    \propto
      & \sum_{i,j,\ell}\mathbb{E}_{q^{\left(t\right)}}\left[\log\pi\left(x_{ij\ell}\mid\boldsymbol{\xi}\right)+\log\pi\left(z_{ij\ell}\mid\boldsymbol{\beta}\right)-\log q_{z_{ij\ell}}^{\left(t\right)}\left(z_{ij\ell}\right)\right]           \\
      & +\sum_{j',\ell}\mathbb{E}_{q^{\left(t\right)}}\left[\log\pi\left(y_{j'\ell}\mid\boldsymbol{\theta},\boldsymbol{\eta},\boldsymbol{\sigma}\right)-\log q_{y_{j'\ell}}^{\left(t\right)}\left(y_{j'\ell}\right)\right]                       \\
      & +\sum_{\ell}\mathbb{E}_{q^{\left(t\right)}}\left[\log\pi\left(\beta_{\ell}\right)-\log q_{\beta_{\ell}}^{\left(t\right)}\left(\beta_{\ell}\right)\right]                                                                                 \\
      & +\sum_{\ell=1}^{\ell_{1}}\sum_{m=1}^{M_{\ell}}\mathbb{E}_{q^{\left(t\right)}}\left[\log\pi\left(\theta_{\ell m}\right)-\log q_{\theta_{\ell m}}^{\left(t\right)}\left(\theta_{\ell m}\right)\right]                                      \\
      & +\sum_{\ell=\ell_{1}+1}^{p}\mathbb{E}_{q^{\left(t\right)}}\left[\log\pi\left(\sigma_{\ell}\right)-\log q_{\sigma_{\ell}}^{\left(t\right)}\left(\sigma_{\ell}\right)-\log q_{\eta_{\ell}}^{\left(t\right)}\left(\eta_{\ell}\right)\right] \\
      & -\sum_{i,j}\mathbb{E}_{q^{\left(t\right)}}\left[\log q_{\lambda_{ij}}^{\left(t\right)}\left(\lambda_{ij}\right)\right].
  \end{aligned}
\end{equation*}

\subsection{Initialization Strategy and Update Process of CAVI}
\label{app:vi-init}

As we have mentioned in Section 3 and Section 3.2 in the main manuscript, the initialization of the variational factors and the update procedure of CAVI are crucial for optimization.
We start by setting initial values for $\boldsymbol{\Lambda}$.

First, we identify the identically matching records.
Then, we take a sample of them with probability $p_{\text{Init}}$ and assign indices $j'$ to the corresponding $\lambda_{ij}$.
If we were to use a single CAVI for posterior estimation, we could set the initial values to the matching records; however, we sample those records to introduce randomness for the members in the evolutionary algorithm.
Next, we initialize $\boldsymbol{\nu}_{ij}=\left(\nu_{ij,1},\cdots,\nu_{ij,N_{\max}}\right)$ by placing a high probability on the $j'$th element for $\lambda_{ij}=j'$.
Next, $\gamma_{j'\ell}$, $\tilde{\eta}_{j'\ell}$, and $\tilde{\sigma}_{j'\ell}$ are initialized.
$\gamma_{j'\ell}$ is initialized similar to the $\boldsymbol{\nu}_{ij}$'s, giving a high probability to the assigned $m$th element for $y_{j'\ell}=m$.
Additionally, $\tilde{\eta}_{j'\ell}$ and $\tilde{\sigma}_{j'\ell}$ are initialized with the mean of allocated records and the squared range of the records, respectively.

We complete the initialization by setting the initial values of $\rho_{ij\ell}$ based on the mean of the prior distribution of $\beta_{\ell}$.
Let $\beta_{\ell}^{\star}$ be the distortion rate derived from the researchers' domain knowledge.
Then, following the suggestion of \cite{marchant2021}, we set the initial $\rho_{ij\ell}$ as $\mathbb{E}\left[\beta_{\ell}\right]$ for
\begin{equation*}
  \beta_{\ell}\sim\text{Beta}\left(N_{\max}\times0.1\times\beta_{\ell}^{\star},N_{\max}\times0.1\right).
\end{equation*}

With those initial values, the variational factors are updated in the order of $\boldsymbol{\theta}_{\ell}$, $\eta_{\ell}$, $\sigma_{\ell}$, $\beta_{\ell}$, $\lambda_{ij}$, $y_{j'\ell}$, and $z_{ij\ell}$.
Finally, the update of the variational factors is repeated until it satisfies the convergence criterion.

The detailed CAVI update process is
\begin{enumerate}
  \item Update $\boldsymbol{\theta}_{\ell}$:
        \begin{equation*}
          q_{\boldsymbol{\theta}_{\ell}}^{\left(t+1\right)}\left(\boldsymbol{\theta}_{\ell}\right)
          \sim\text{Dirichlet}\left(\left\{\alpha_{\ell m}^{\left(t+1\right)}\right\}_{m=1}^{M_{\ell}}\right),
        \end{equation*}
        where
        \begin{equation*}
          \alpha_{\ell m}^{\left(t+1\right)}=\mu_{\ell m}+\sum_{j'=1}^{N_{\max}}\frac{\gamma_{j'\ell,m}^{\left(t\right)}}{\sum_{r=1}^{M_{\ell}}\gamma_{j'\ell,r}^{\left(t\right)}}+\sum_{i,j}\mathbb{I}\left(x_{ij\ell}=m\right)\rho_{ij\ell}^{\left(t\right)}.
        \end{equation*}
  \item Update $\eta_{\ell}$:
        \begin{equation*}
          q_{\eta_{\ell}}^{\left(t+1\right)}\left(\eta_{\ell}\right)
          \sim\text{N}\left(\frac{\sum_{j'}\tilde{\eta}_{j'\ell}^{\left(t\right)}+\sum_{i,j}x_{ij\ell}\rho_{ij\ell}^{\left(t\right)}}{N_{\max}+\sum_{i,j}\rho_{ij\ell}^{\left(t\right)}},
          \frac{\tilde{b}_{\ell}^{\sigma,\left(t\right)}}{\left(N_{\max}+\sum_{i,j}\rho_{ij\ell}^{\left(t\right)}\right)\left(a_{\ell}^{\sigma}+\frac{1}{2}\left[N_{\max}+\sum_{i,j}\rho_{ij\ell}^{\left(t\right)}\right]\right)}\right).
        \end{equation*}
  \item Update $\sigma_{\ell}$:
        \begin{equation*}
          q_{\sigma_{\ell}}^{\left(t+1\right)}\left(\sigma_{\ell}\right)
          \sim\text{Inverse-Gamma}\left(a^{\sigma}_{\ell}+\frac{1}{2}\left[N_{\max}+\sum_{i,j}\rho_{ij\ell}^{\left(t\right)}\right],
          \tilde{b}^{\sigma,\left(t+1\right)}_{\ell}\right),
        \end{equation*}
        where
        \begin{equation*}
          \tilde{b}^{\sigma,\left(t+1\right)}_{\ell}
          =\tilde{b}^{\sigma}_{\ell}
          +\frac{1}{2}\Bigg[
            \sum_{j'}\left(\tilde{\sigma}_{j'\ell}^{\left(t\right)}+\left(\tilde{\eta}_{j'\ell}^{\left(t\right)}\right)^{2}-2\tilde{\eta}_{j'\ell}^{\left(t\right)}\zeta_{\ell 1}^{\left(t\right)}+\zeta_{\ell 2}^{\left(t\right)}\right)
            +\sum_{i,j}\left(\rho_{ij\ell}^{\left(t\right)}\left\{x_{ij\ell}^{2}-2x_{ij\ell}\zeta_{\ell 1}^{\left(t\right)}+\zeta_{\ell 2}^{\left(t\right)}\right\}\right)\Bigg]
        \end{equation*}
        for
        \begin{equation*}
          \begin{aligned}
            \zeta_{\ell 1}^{\left(t\right)}
             & = \frac{\sum_{j'}\tilde{\eta}_{j'\ell}^{\left(t\right)}+\sum_{i,j}x_{ij\ell}\rho_{ij\ell}^{\left(t\right)}}{N_{\max}+\sum_{i,j}\rho_{ij\ell}^{\left(t\right)}}\text{ and }                                                                                                                                                                                                                                \\
            \zeta_{\ell 2}^{\left(t\right)}
             & = \left(\frac{\sum_{j'}\tilde{\eta}_{j'\ell}^{\left(t\right)}+\sum_{i,j}x_{ij\ell}\rho_{ij\ell}^{\left(t\right)}}{N_{\max}+\sum_{i,j}\rho_{ij\ell}^{\left(t\right)}}\right)^{2}+\frac{\tilde{b}_{\ell}^{\sigma,\left(t\right)}}{\left(N_{\max}+\sum_{i,j}\rho_{ij\ell}^{\left(t\right)}\right)\left(a_{\ell}^{\sigma}+\frac{1}{2}\left[N_{\max}+\sum_{i,j}\rho_{ij\ell}^{\left(t\right)}\right]\right)}.
          \end{aligned}
        \end{equation*}
  \item Update $\beta_{\ell}$:
        \begin{equation*}
          q_{\beta_{\ell}}^{\left(t+1\right)}\left(\beta_{\ell}\right)
          \sim\text{Beta}\left(a_{\ell}+\sum_{i,j}\rho_{ij\ell}^{\left(t\right)},b_{\ell}+\sum_{i,j}\left(1-\rho_{ij\ell}^{\left(t\right)}\right)\right).
        \end{equation*}
  \item Update $\lambda_{ij}$:
        \begin{equation*}
          q_{\lambda_{ij}}^{\left(t+1\right)}\left(\lambda_{ij}\right)
          \sim\text{MN}\left(1,\left\{\frac{\nu_{ij,j'}^{\left(t+1\right)}}{\sum_{u=1}^{N_{\max}}\nu_{ij,u}^{\left(t+1\right)}}\right\}_{j'=1}^{N_{\max}}\right),
        \end{equation*}
        where
        \begin{equation*}
          \begin{aligned}
            \nu_{ij,j'}^{\left(t+1\right)}=
            \exp\Bigg[ & \sum_{\ell=1}^{\ell_{1}}\left(1-\rho_{ij\ell}^{\left(t\right)}\right)\left(\frac{\sum_{\left\{r':\lvert x_{ij\ell}-r'\rvert\leq\varepsilon_{\ell}\right\}}\gamma_{j'\ell,r'}^{\left(t\right)}}{\sum_{r=1}^{M_{\ell}}\gamma_{j'\ell,r}^{\left(t\right)}}\right)\log\phi_{\ell}^{\frac{1}{\tau_{\ell}}} \\
                       & -\sum_{\ell=\ell_{1}+1}^{p}\frac{1}{2\varepsilon_{\ell}}\left(1-\rho_{ij\ell}^{\left(t\right)}\right)\left(x_{ij\ell}^{2}-2x_{ij\ell}\tilde{\eta}_{j'\ell}^{\left(t\right)}+\tilde{\sigma}_{j'\ell}^{\left(t\right)}+\left(\tilde{\eta}_{j'\ell}^{\left(t\right)}\right)^{2}\right) \Bigg].
          \end{aligned}
        \end{equation*}
  \item Update $y_{j'\ell}$:
        \begin{equation*}
          q_{y_{j'\ell}}^{\left(t+1\right)}\left(y_{j'\ell}\right)
          \sim\begin{cases}
            \text{MN}\left(1,\left\{\frac{\gamma_{j'\ell,m}^{\left(t+1\right)}}{\sum_{r=1}^{M_{\ell}}\gamma_{j'\ell,r}^{\left(t+1\right)}}\right\}_{m=1}^{M_{\ell}}\right) & \text{for }\ell=1,\cdots,\ell_{1}     \\
            \text{N}\left(\tilde{\eta}_{j'\ell}^{\left(t+1\right)},\tilde{\sigma}_{j'\ell}^{\left(t+1\right)}\right)                                                       & \text{for }\ell=\ell_{1}+1,\cdots,p,
          \end{cases}
        \end{equation*}
        where
        \begin{equation*}
          \gamma_{j'\ell,m}^{\left(t+1\right)}
          =\phi_{\ell}^{\frac{1}{\tau_{\ell}}\sum_{\left\{i,j:\lvert x_{ij\ell}-m\rvert\leq\varepsilon_{\ell}\right\}}\frac{\nu_{ij,j'}^{\left(t+1\right)}}{\sum_{u=1}^{N_{\max}}\nu_{ij,u}^{\left(t+1\right)}}\left(1-\rho_{ij\ell}^{\left(t\right)}\right)}\exp\left[\psi\left(\alpha_{\ell m}^{\left(t+1\right)}\right)-\psi\left(\sum_{r=1}^{M_{\ell}}\alpha_{\ell r}^{\left(t+1\right)}\right)\right]
        \end{equation*}
        and
        \begin{equation*}
          \begin{aligned}
            \tilde{\sigma}_{j'\ell}^{\left(t+1\right)}
            = & \left(\frac{1}{\tilde{b}_{\ell}^{\sigma,\left(t+1\right)}}\left(a_{\ell}^{\sigma}+\frac{1}{2}\left[N_{\max}+\sum_{i,j}\rho_{ij\ell}^{\left(t\right)}\right]\right)+\frac{1}{\varepsilon_{\ell}}\left[\sum_{i,j}\frac{\nu_{ij,j'}^{\left(t+1\right)}}{\sum_{u=1}^{N_{\max}}\nu_{ij,u}^{\left(t+1\right)}}\left(1-\rho_{ij\ell}^{\left(t\right)}\right)\right]\right)^{-1}                \\
            \tilde{\eta}_{j'\ell}^{\left(t+1\right)}
            = & \sigma_{j'\ell}^{\left(t+1\right)}\left[\frac{\left(\sum_{j'}\tilde{\eta}_{j'\ell}^{\left(t\right)}+\sum_{i,j}x_{ij\ell}\rho_{ij\ell}^{\left(t\right)}\right)\left(a_{\ell}^{\sigma}+\frac{1}{2}\left[N_{\max}+\sum_{i,j}\rho_{ij\ell}^{\left(t\right)}\right]\right)}{\left(N_{\max}+\sum_{i,j}\rho_{ij\ell}^{\left(t\right)}\right)\tilde{b}_{\ell}^{\sigma,\left(t+1\right)}}\right] \\
              & +\sigma_{j'\ell}^{\left(t+1\right)}\left[\frac{1}{\varepsilon_{\ell}}\left[\sum_{i,j}x_{ij\ell}\frac{\nu_{ij,j'}^{\left(t+1\right)}}{\sum_{u=1}^{N_{\max}}\nu_{ij,u}^{\left(t+1\right)}}\left(1-\rho_{ij\ell}^{\left(t\right)}\right)\right]\right],
          \end{aligned}
        \end{equation*}
        with the digamma function $\psi$.
  \item Update $z_{ij\ell}$:
        \begin{equation*}
          q_{\rho_{ij\ell}}^{\left(t+1\right)}\left(\rho_{ij\ell}\right)\sim\text{Bernoulli}\left(\rho_{ij\ell}^{\left(t+1\right)}\right) \text{ for }\rho_{ij\ell}^{\left(t+1\right)}=\frac{\rho_{ij\ell ,1}^{\left(t+1\right)}}{\rho_{ij\ell ,1}^{\left(t+1\right)}+\rho_{ij\ell ,2}^{\left(t+1\right)}},
        \end{equation*}
        where
        \begin{equation*}
          \begin{aligned}
            \rho_{ij\ell ,1}^{\left(t+1\right)}
             & =\begin{cases}
                  \begin{aligned}
                \exp\Bigg[ & \psi\left(a_{\ell}+\sum_{i,j}\rho_{ij\ell}^{\left(t\right)}\right)-\psi\left(a_{\ell}+b_{\ell}+N_{\max}\right)                                                                                             \\
                           & +\sum_{m=1}^{M_{\ell}}\mathbb{I}\left(x_{ij\ell}=m\right)\left\{\psi\left(\alpha_{\ell m}^{\left(t+1\right)}\right)-\psi\left(\sum_{r=1}^{M_{\ell}}\alpha_{\ell r}^{\left(t+1\right)}\right)\right\}\Bigg] \\
                           & \qquad\qquad\qquad\qquad\qquad\qquad\qquad\qquad\qquad\qquad\text{if }\ell=1,\cdots,\ell_{1}
              \end{aligned} \\
                  \begin{aligned}
                \exp\Bigg[ & \psi\left(a_{\ell}+\sum_{i,j}\rho_{ij\ell}^{\left(t\right)}\right)-\psi\left(a_{\ell}+b_{\ell}+N_{\max}\right)                                                                                                                                                   \\
                           & -\frac{1}{2}\log2\pi-\frac{1}{2}\left(\log\tilde{b}^{\sigma,\left(t+1\right)}_{\ell}-\psi\left(a^{\sigma}_{\ell}+\frac{1}{2}\left[N_{\max}+\sum_{i,j}\rho_{ij\ell}^{\left(t\right)}\right]\right)\right)                                                         \\
                           & -\frac{a^{\sigma}_{\ell}+\frac{1}{2}\left[N_{\max}+\sum_{i,j}\rho_{ij\ell}^{\left(t\right)}\right]}{2\tilde{b}_{\ell}^{\sigma,\left(t+1\right)}}\left(x_{ij\ell}^{2}-2x_{ij\ell}\zeta_{\ell 1}^{\left(t+1\right)}+\zeta_{\ell 2}^{\left(t+1\right)}\right)\Bigg] \\
                           & \qquad\qquad\qquad\qquad\qquad\qquad\qquad\qquad\qquad\qquad\text{if }\ell=\ell_{1},\cdots,p
              \end{aligned}
                \end{cases}
          \end{aligned}
        \end{equation*}
        and
        \begin{equation*}
          \begin{aligned}
            \rho_{ij\ell ,2}^{\left(t+1\right)}
             & =\begin{cases}
                  \begin{aligned}
                \exp\Bigg[ & \psi\left(b_{\ell}+\sum_{i,j}\left(1-\rho_{ij\ell}^{\left(t\right)}\right)\right)-\psi\left(a_{\ell}+b_{\ell}+N_{\max}\right)                                                                                                                                                                                                                                 \\
                           & +\log\phi_{\ell}^{\frac{1}{\tau_{\ell}}}\sum_{j'}\left[\frac{\sum_{\left\{r':\lvert x_{ij\ell}-r'\rvert\leq\varepsilon_{\ell}\right\}}\gamma_{j'\ell,r'}^{\left(t+1\right)}}{\sum_{r=1}^{M_{\ell}}\gamma_{j'\ell,r}^{\left(t+1\right)}}\frac{\nu_{ij,j'}^{\left(t+1\right)}}{\sum_{u=1}^{N_{\max}}\nu_{ij,u}^{\left(t+1\right)}}\right]-\log c_{ij\ell}\Bigg] \\
                           & \qquad\qquad\qquad\qquad\qquad\qquad\qquad\qquad\qquad\qquad\text{if }\ell=1,\cdots,\ell_{1}
              \end{aligned} \\
                  \begin{aligned}
                \exp\Bigg[ & \psi\left(b_{\ell}+\sum_{i,j}\left(1-\rho_{ij\ell}^{\left(t\right)}\right)\right)-\psi\left(a_{\ell}+b_{\ell}+N_{\max}\right)                                                                                                                                               \\
                           & -\frac{1}{2}\log2\pi\varepsilon_{\ell}-\frac{1}{2\varepsilon_{\ell}}\Big(x_{ij\ell}^{2}-2x_{ij\ell}\sum_{j'}\left[\frac{\nu_{ij,j'}^{\left(t+1\right)}}{\sum_{u=1}^{N_{\max}}\nu_{ij,u}^{\left(t+1\right)}}\tilde{\eta}_{j'\ell}^{\left(t+1\right)}\right]                  \\
                           & \qquad\qquad\qquad\qquad\quad+\sum_{j'}\left[\frac{\nu_{ij,j'}^{\left(t+1\right)}}{\sum_{u=1}^{N_{\max}}\nu_{ij,u}^{\left(t+1\right)}}\left\{\tilde{\sigma}_{j'\ell}^{\left(t+1\right)}+\left(\tilde{\eta}_{j'\ell}^{\left(t+1\right)}\right)^{2}\right\}\right]\Big)\Bigg] \\
                           & \qquad\qquad\qquad\qquad\qquad\qquad\qquad\qquad\qquad\qquad\text{if }\ell=\ell_{1},\cdots,p
              \end{aligned}
                \end{cases}
          \end{aligned}
        \end{equation*}
        with
        \begin{equation*}
          \begin{aligned}
            \zeta_{\ell 1}^{\left(t+1\right)}
            = & \frac{\sum_{j'}\tilde{\eta}_{j'\ell}^{\left(t+1\right)}+\sum_{i,j}x_{ij\ell}\rho_{ij\ell}^{\left(t\right)}}{N_{\max}+\sum_{i,j}\rho_{ij\ell}^{\left(t\right)}}                                                                \\
            \zeta_{\ell 2}^{\left(t+1\right)}
            = & \left(\frac{\sum_{j'}\tilde{\eta}_{j'\ell}^{\left(t+1\right)}+\sum_{i,j}x_{ij\ell}\rho_{ij\ell}^{\left(t\right)}}{N_{\max}+\sum_{i,j}\rho_{ij\ell}^{\left(t\right)}}\right)^{2}                                               \\
              & + \frac{\tilde{b}_{\ell}^{\sigma,\left(t+1\right)}}{\left(N_{\max}+\sum_{i,j}\rho_{ij\ell}^{\left(t\right)}\right)\left(a_{\ell}^{\sigma}+\frac{1}{2}\left[N_{\max}+\sum_{i,j}\rho_{ij\ell}^{\left(t\right)}\right]\right)}.
          \end{aligned}
        \end{equation*}
\end{enumerate}

\subsection{Evolutionary Variational Inference for Record Linkage}
\label{app:vi-evil}

This section revisits the Evolutionary Variational Inference for Record Linkage (EVIL) algorithm described in Section 3.2 in the main manuscript and provides a summarized algorithm to aid future implementations and extensions.

Evolutionary algorithms solve optimization problems by mimicking biological evolution, searching for optimal solutions through selection, crossover, mutation, and iterative optimization \citep{whitley2001,givens2012}.
Before the main optimization process starts, researchers must generate the initial population and set the hyperparameters.
The initial generation consists of $N_{P}$ members with different initial values, where the initialization strategy described in Appendix~\ref{app:vi-init} is used.
If we denote the initial values of member $p$ as $\boldsymbol{\xi}_{p}^{\left(0\right)}=\{\boldsymbol{\Lambda}^{\left(0\right)},\boldsymbol{y}^{\left(0\right)},\boldsymbol{z}^{\left(0\right)},\boldsymbol{\beta}^{\left(0\right)},\boldsymbol{\eta}^{\left(0\right)}\}_{p}$, the first generation can be represented as a set $\boldsymbol{\Xi}^{\left(0\right)}=\{\boldsymbol{\xi}_{p}^{\left(0\right)}\}_{p=1}^{N_{P}}$.
Additionally, the number of offspring generated during the evolutionary process, $1\leq N_{O}\leq\binom{N_{P}}{2}$, the maximum number of generations $N_{E}$, and the tolerance level $\epsilon_{\text{tol}}$ serving as the criterion for terminating the evolution must be set.

EVIL begins with $\boldsymbol{\Xi}^{\left(0\right)}$ and finds the optimal approximate posterior distribution of $\boldsymbol{\Lambda}$ through the iteration of the following process.
For a summary, see Algorithm~\ref{alg:evil}.
\begin{enumerate}
  \item \textbf{Update $\boldsymbol{\Xi}^{\left(0\right)}$ via CAVI to create $\boldsymbol{\Xi}^{\left(1\right)}$:}
        Each member $\boldsymbol{\xi}_{p}^{\left(0\right)}$ is updated to find approximate posterior distributions using CAVI.
        The algorithm stops if the relative change in the ELBO is less than the prespecified tolerance, $\epsilon_{\text{tol}}$, such that
        \begin{equation*}
          \left|
          \frac
          {\text{ELBO}\left(q^{\left(t\right)}\right)-\text{ELBO}\left(q^{\left(t-1\right)}\right)}
          {\text{ELBO}\left(q^{\left(t-1\right)}\right)}\right|<\epsilon_{\text{tol}},
        \end{equation*}
        or the number of iterations reaches the maximum number of iterations for CAVI.
        The updated population is denoted as $\boldsymbol{\Xi}^{\left(1\right)}$.
  \item \textbf{Single-point crossover with $\boldsymbol{\Xi}^{\left(1\right)}$ to create $N_{O}$ offspring:}
        From $\boldsymbol{\Xi}^{\left(1\right)}$, randomly select $N_{O}$ pairs of distinct members.
        Implement a single-point crossover on the set of $N_{O}$ pairs.
        That is, suppose that $\boldsymbol{\xi}_{1}^{\left(1\right)}$ and $\boldsymbol{\xi}_{2}^{\left(1\right)}$ belong to the selected pair.
        Randomly choose an index $\left(i',j'\right)$ and swap the linkage indices corresponding to the records after $\boldsymbol{x}_{i'j'}$.
        Denote the set of offspring as $\boldsymbol{\Xi}_{C}^{\left(1\right)}=\{\boldsymbol{\xi}_{p}^{C}\}_{p=1}^{N_{O}}$.
  \item \textbf{Reconcile broken linkage structure of $N_{O}$ offspring:}
        As the linkage indices are swapped in the crossover step, the consistency between the linkage structure and variational factors may be broken.
        Apply a reconciliation process to correct this by recalculating the variational parameters based on the new linkage structure.
        For instance, assuming $\lambda_{1}=\lambda_{2}=1$ and $\lambda_{3}=\lambda_{4}=2$ as the linkage structure before crossover, $\boldsymbol{x}_{1}$ and $\boldsymbol{x}_{2}$ are linked through $\boldsymbol{y}_{1}$, and $\boldsymbol{x}_{3}$ and $\boldsymbol{x}_{4}$ are linked through $\boldsymbol{y}_{2}$.
        If the value of $\lambda_{1}$ is switched to $2$, then $\boldsymbol{x}_{1}$, $\boldsymbol{x}_{3}$, and $\boldsymbol{x}_{4}$ are now linked through $\boldsymbol{y}_{2}$, and $\boldsymbol{x}_{2}$ considers $\boldsymbol{y}_{1}$ as a true latent record, but no other linked records exist.
        Accordingly, we need to update $\boldsymbol{y}_{1}$ and $\boldsymbol{y}_{2}$ based on the new linkage structure and modify corresponding distortion indicators $\boldsymbol{z}_{1},\cdots,\boldsymbol{z}_{4}$.
        Specifically, variational parameters corresponding to $\boldsymbol{y}_{1}$, $\boldsymbol{y}_{2}$, and $\boldsymbol{z}_{1},\cdots,\boldsymbol{z}_{4}$ are updated based on $\lambda_{2}=1$ and $\lambda_{1}=\lambda_{3}=\lambda_{4}=2$.
        Moreover, the remaining variational parameters $\boldsymbol{\beta}$ and $\boldsymbol{\eta}$ also need to be updated based on the new linkage structure.
        Denote the set of reconciled offspring as $\boldsymbol{\Xi}_{R}^{\left(1\right)}=\{\boldsymbol{\xi}_{p}^{R}\}_{p=1}^{N_{O}}$.
  \item \textbf{Update $N_{O}$ offspring via CAVI to create $\boldsymbol{\Xi}^{\left(2\right)}$:}
        Each member $\boldsymbol{\xi}_{p}^{R}$ is updated to estimate the approximate posterior distributions using CAVI.
        The algorithm stops if it meets the same condition as in the first step.
        The updated population is denoted as $\boldsymbol{\Xi}^{\left(2\right)}$.
  \item \textbf{Select $N_{P}$ members from $\boldsymbol{\Xi}^{\left(1\right)}\cup\boldsymbol{\Xi}^{\left(2\right)}$ based on the ELBO:}
        Select the top $N_{P}$ best-performing members from $\boldsymbol{\Xi}^{\left(2\right)}$ based on the ELBO.
        Namely, create an ordered set $S=\{\tilde{\boldsymbol{\xi}}_{\left(1\right)},\cdots,\tilde{\boldsymbol{\xi}}_{\left(N_{P}+N_{O}\right)}\}$ from $\boldsymbol{\Xi}^{\left(1\right)}\cup\boldsymbol{\Xi}^{\left(2\right)}=\{\tilde{\boldsymbol{\xi}}_{p}\}_{p=1}^{N_{P}+N_{O}}$, where
        \begin{equation*}
          \text{ELBO}\left(\tilde{\boldsymbol{\xi}}_{\left(1\right)}\right)\geq\text{ELBO}\left(\tilde{\boldsymbol{\xi}}_{\left(2\right)}\right)\geq\cdots\geq\text{ELBO}\left(\tilde{\boldsymbol{\xi}}_{\left(N_{P}+N_{O}\right)}\right).
        \end{equation*}
        Then, select the first $N_{P}$ members as the new parents such that $\tilde{\boldsymbol{\Xi}}=\{\tilde{\boldsymbol{\xi}}_{\left(p\right)}\}_{p=1}^{N_{P}}$.
        Set $\text{ELBO}^{\left(t\right)}=\text{ELBO}\left(\tilde{\boldsymbol{\xi}}_{\left(1\right)}\right)$ for the $t$th generation.
        If the relative difference between $\text{ELBO}^{\left(t\right)}$ and $\text{ELBO}^{\left(t-1\right)}$ is less than $\epsilon_{\text{tol}}$ or $\text{ELBO}^{\left(t\right)}\leq\text{ELBO}^{\left(t-1\right)}$ or the number of generations reaches the maximum number of generations, $N_{E}$, terminate the optimization.
        Note that the second condition can occur because of the following mutation step.
        If the optimization is terminated, $\tilde{\boldsymbol{\xi}}_{\left(1\right)}$ is considered as the optimal solution.
        Otherwise, go to the next step.
  \item \textbf{Mutate $\tilde{\boldsymbol{\Xi}}$ and reconcile broken linkage structure:}
        Mutate $\tilde{\boldsymbol{\Xi}}$ through a split-merge process \citep{jain2004}.
        Specifically, randomly select two records $\boldsymbol{x}_{i_{1}j_{1}}$ and $\boldsymbol{x}_{i_{2}j_{2}}$ from a member $\tilde{\boldsymbol{\xi}}_{\left(p\right)}$ in $\tilde{\boldsymbol{\Xi}}$.
        If they are linked, i.e., $\lambda_{i_{1}j_{1}}=\lambda_{i_{2}j_{2}}$, split both records by assigning either $\lambda_{i_{1}j_{1}}$ or $\lambda_{i_{2}j_{2}}$ to an index that is not used yet.
        Otherwise, merge them by making $\lambda_{i_{1}j_{1}}=\lambda_{i_{2}j_{2}}$.
        Then, the consistency between the linkage structure and variational factors is broken similar to the crossover step.
        Apply the reconciliation process in the third step to correct this.
        Replace the parameters after the mutation and reconciliation with $\boldsymbol{\Xi}^{\left(0\right)}$ and go to the next generation.
\end{enumerate}

\begin{algorithm}[H]
  \caption{\textbf{E}volutionary Algorithm of \textbf{V}ariational \textbf{I}nference for Record \textbf{L}inkage (EVIL)}
  \label{alg:evil}
  
  \begin{algorithmic}[1]
    
    \Procedure{EVIL}{$\boldsymbol{x}$}
    \State Initialize $\boldsymbol{\Xi}^{\left(0\right)}$ with $N_{P}$ sets of initial values.
    \State $t\leftarrow1$
    \While{$\text{Relative change of ELBO}>\epsilon_{\text{tol}}$}  \Comment{$\epsilon_{\text{tol}}=\text{convergence criterion}$}
    
    \State Update $\boldsymbol{\Xi}^{\left(0\right)}\rightarrow\boldsymbol{\Xi}^{\left(1\right)}$ via CAVI.
    \State Crossover $\boldsymbol{\Xi}^{\left(1\right)}\rightarrow\boldsymbol{\Xi}_{C}^{\left(1\right)}$. \Comment{Crossover $N_{P}$ parents to generate $N_{O}$ offspring}
    \State Reconcile broken linkage structure of members in $\boldsymbol{\Xi}_{C}^{\left(1\right)}\rightarrow\boldsymbol{\Xi}_{R}^{\left(1\right)}$.
    \State Update $\boldsymbol{\Xi}_{R}^{\left(1\right)}\rightarrow\boldsymbol{\Xi}^{\left(2\right)}$ via CAVI.
    \State Select $N_{P}$ members from $\boldsymbol{\Xi}^{\left(1\right)}\cup\boldsymbol{\Xi}^{\left(2\right)}$ based on ELBO$\rightarrow\tilde{\boldsymbol{\Xi}}$.
    \State Set $\text{ELBO}^{\left(t\right)}\leftarrow\max_{\text{ELBO}}\left(\tilde{\boldsymbol{\Xi}}\right)$.
    
    \If{$\text{Relative change from }\text{ELBO}^{\left(t-1\right)}\text{ to }\text{ELBO}^{\left(t\right)}\leq\epsilon_{\text{tol}}$}
    \State Break the loop.
    \EndIf
    
    \State Mutate and reconcile $\tilde{\boldsymbol{\Xi}}$.
    \State Replace $\tilde{\boldsymbol{\Xi}}\rightarrow\boldsymbol{\Xi}^{\left(0\right)}$.
    \State $t\leftarrow t+1$
    \EndWhile  
    \EndProcedure
    
  \end{algorithmic}
\end{algorithm}

\section{Details for the Application Study}
\label{app:ishiw}

In this section, we present the detailed specifications of the model used to fit the Italian Survey on Household Income and Wealth (ISHIW) data \citep{bankofitaly2025}, convergence diagnostics, and model results.
Additionally, we compare the elapsed times of CHOMPER-MCMC and CHOMPER-EVIL to demonstrate the scalability of EVIL.
Table~\ref{tab:ishiw-variables} presents the variables used for the linkage structure estimation.

\begin{table}[H]
  \centering\caption{Variables used for the linkage structure estimation in the ISHIW data.}
  \label{tab:ishiw-variables}
  \begin{tabular}{clc}
    \hline
    Name      & \multicolumn{1}{c}{Description}      & Type        \\ \hline
    SEX       & Sex                                  & Categorical \\
    ANASC     & Year of birth                        & Categorical \\
    CIT       & Nationality (whether Italian or not) & Categorical \\
    NASREG    & Place of birth                       & Categorical \\
    STUDIO    & Education level                      & Categorical \\
    NETINCOME & Net annual income                    & Continuous  \\
    VALABIT   & Estimated price of residence         & Continuous  \\ \hline
  \end{tabular}
\end{table}

\subsection{Prior and Hyperparameters}
\label{app:ishiw-prior}

In this section, we describe the prior and hyperparameter specifications assumed for the CHOMPER model defined in Equation~\eqref{eqn:chomper-mdl-appendix} in various settings.
The Conservative approach uses only 4 categorical variables, so $p=\ell_{1}=4$.
Since the Conservative approach assumes that all variables have strong, immutable signals, we set the hitting range to $\varepsilon_{\ell}=0$ and $\tau_{\ell}=10^{-2}$.
The Naive and Comprehensive approaches use an additional categorical field with multiple truths and 2 continuous fields with multiple truths, so $p=7$ and $\ell_{1}=5$.
We set $\phi_{\ell}=2$ for all approaches.
We defined the fifth field as education level, and the sixth and seventh fields as net annual income and estimated price of residence, respectively.
Both continuous fields were standardized to have a mean of 0 and a variance of 1 before analysis.

The Naive approach does not consider locally-varying hits, so we set $\varepsilon_{5}=0$, $\tau_{5}=10^{-2}$, and $\varepsilon_{6}=\varepsilon_{7}=10^{-4}$.
However, the Comprehensive approach accounts for multiple truths in the attributes.
As education level is not expected to change rapidly within two years, we set $\varepsilon_{5}=1$ and $\tau_{5}=10^{-2}$ to include adjacent categories of the education level within the hitting range.
For continuous fields such as net annual income and estimated price of residence, defining the threshold for what changes are considered ``noise'' is challenging.
In particular, it is not possible to measure the local variability of a variable, even when the true linkage structure is known.
Therefore, we recommend following the opinion of domain experts in practice and selecting a hitting range that errs toward a wider range, as CHOMPER demonstrates robustness with a slightly broader hitting range and occasionally outperforms models with the true range (See Appendix~\ref{app:simulation-results}).

We selected $\varepsilon_{6}=\varepsilon_{7}=0.1$ which is interpreted as 10\% of the total variability for each attribute in the sample data.
Based on ISHIW data from 2020, the standard deviations of net annual income and estimated price of residence are \euro{53,117} and \euro{323,180}, respectively.
Therefore, the choice of $\varepsilon_{6}=\varepsilon_{7}=0.1$ implies that the model with the Comprehensive approach captures local variability of $\pm$\euro{16,797} and $\pm$\euro{102,198} for each observation in 2022, respectively.

\begin{table}[!ht]
  \centering\caption{Hitting range parameters according to the approaches for fitting the ISHIW data. Conservative and Naive approaches use infinitesimal hitting range, while Comprehensive approach sets wide hitting range for multiple truths.}
  \label{tab:ishiw-hitting-range}
  \renewcommand{\arraystretch}{0.8}
  \begin{tabular}{cccc}
    \hline
    \multirow{2}{*}{Approach} & \multirow{2}{*}{Single Truth}                 & \multicolumn{2}{c}{Local Variability}                                        \\ \cline{3-4}
                              &                                               & Multinomial                                   & Gaussian                     \\ \hline
    Conservative              & $\varepsilon_{\ell}=0$, $\tau_{\ell}=10^{-2}$ & -                                             & -                            \\
    Naive                     & $\varepsilon_{\ell}=0$, $\tau_{\ell}=10^{-2}$ & $\varepsilon_{\ell}=0$, $\tau_{\ell}=10^{-2}$ & $\varepsilon_{\ell}=10^{-4}$ \\
    Comprehensive             & $\varepsilon_{\ell}=0$, $\tau_{\ell}=10^{-2}$ & $\varepsilon_{\ell}=1$, $\tau_{\ell}=10^{-2}$ & $\varepsilon_{\ell}=10^{-1}$ \\ \hline
  \end{tabular}
\end{table}

We used the parameterization method proposed in \cite{marchant2021} as described in Section 4.1 of the main manuscript to set the prior distortion rate for each attribute to 1\%.
We set the hyperparameter for the probability of each level in categorical fields to a length-$M_{\ell}$ one vector, $\boldsymbol{\mu}_{\ell}=\boldsymbol{1}_{\ell}$, and the hyperparameters for the variance of continuous fields were set to 0.01 for both shape and scale parameters.
In other words, the following prior distributions were assumed:
\begin{equation*}
  \begin{aligned}
    \beta_{\ell}               & \sim\text{Beta}\left(N_{\max}^{r}\times0.1\times0.01,N_{\max}^{r}\times0.1\right)        \\
    \boldsymbol{\theta}_{\ell} & \sim\text{Dirichlet}\left(\boldsymbol{1}_{\ell}\right)\text{ for }\ell=1,\cdots,\ell_{1} \\
    \sigma_{\ell}              & \sim\text{Inverse-Gamma}\left(0.01,0.01\right)\text{ for }\ell=\ell_{1}+1,\cdots,p,
  \end{aligned}
\end{equation*}
where $N_{\max}^{r}$ is the number of records in region $r$.

Next, we describe the initialization strategy for CHOMPER-MCMC.
As there is no prior information about the linkage structure, $\lambda_{ij}$ are assigned indices from 1 to $N_{\max}$ without duplication, assuming that there is no linkage in the data, and $\boldsymbol{y}_{j'}$ are initialized to $\boldsymbol{x}_{ij}$ corresponding to each $\lambda_{ij}$.
Consequently, all initial $z_{ij\ell}$ are set to 0 because observed records and initial true records are the same according to the initialized $\boldsymbol{\Lambda}$.
$\boldsymbol{\theta}_{\ell}$ is initialized as a vector consisting of $M_{\ell}$ elements, each valued with $\frac{1}{M_{\ell}}$, and $\eta_{\ell}$ and $\sigma_{\ell}$ are initialized as the mean and variance of $\boldsymbol{y}_{\cdot\ell}$, respectively.
Finally, $\beta_{\ell}$ is initialized as the mean of the beta distribution, $\frac{0.01}{1.01}\approx0.01$.

Finally, we specify the hyperparameters for CHOMPER-EVIL's evolutionary process and the number of samples for CHOMPER-MCMC.
When fitting the ISHIW data with living region as a blocking variable using CHOMPER-EVIL, we set the number of parents to $N_{P}=50$, the number of offspring to $N_{O}=100$, and the maximum number of evolutions to $N_{E}=50$.
Additionally, evolution was set to stop when the relative change in the ELBO between generations was less than $10^{-5}$.
For CHOMPER-MCMC, we performed 1,000 split-merge updates per Gibbs iteration and used the last 1,000 converged samples for inference, based on the convergence diagnostics in Appendix~\ref{app:ishiw-convergence}.
However, when fitting the entire ISHIW dataset together, we set $N_{P}=3$, $N_{O}=4$, and $N_{E}=50$, with the same stopping condition applied.

\subsection{Convergence Diagnostics}
\label{app:ishiw-convergence}

This section provides convergence diagnostics for CHOMPER-MCMC.
The convergence diagnostics were performed using trace plots of the number of unique $\lambda_{ij}$ in each MCMC sample, which represents the number of unique entities in the data \citep{aleshin-guendel2024}.
Figure~\ref{fig:ishiw-regional-convergence} shows the trace plots of the number of unique $\lambda_{ij}$ in each region for each approach.
We determined the number of burn-in samples visually from Figure~\ref{fig:ishiw-regional-convergence}, and 1,000 MCMC samples after burn-in were used for inference.
The number of burn-in samples for each region is provided in Table~\ref{tab:ishiw-burn-in-samples}.

\begin{figure}[!ht]
  \begin{center}
    \includegraphics[width=0.9\linewidth]{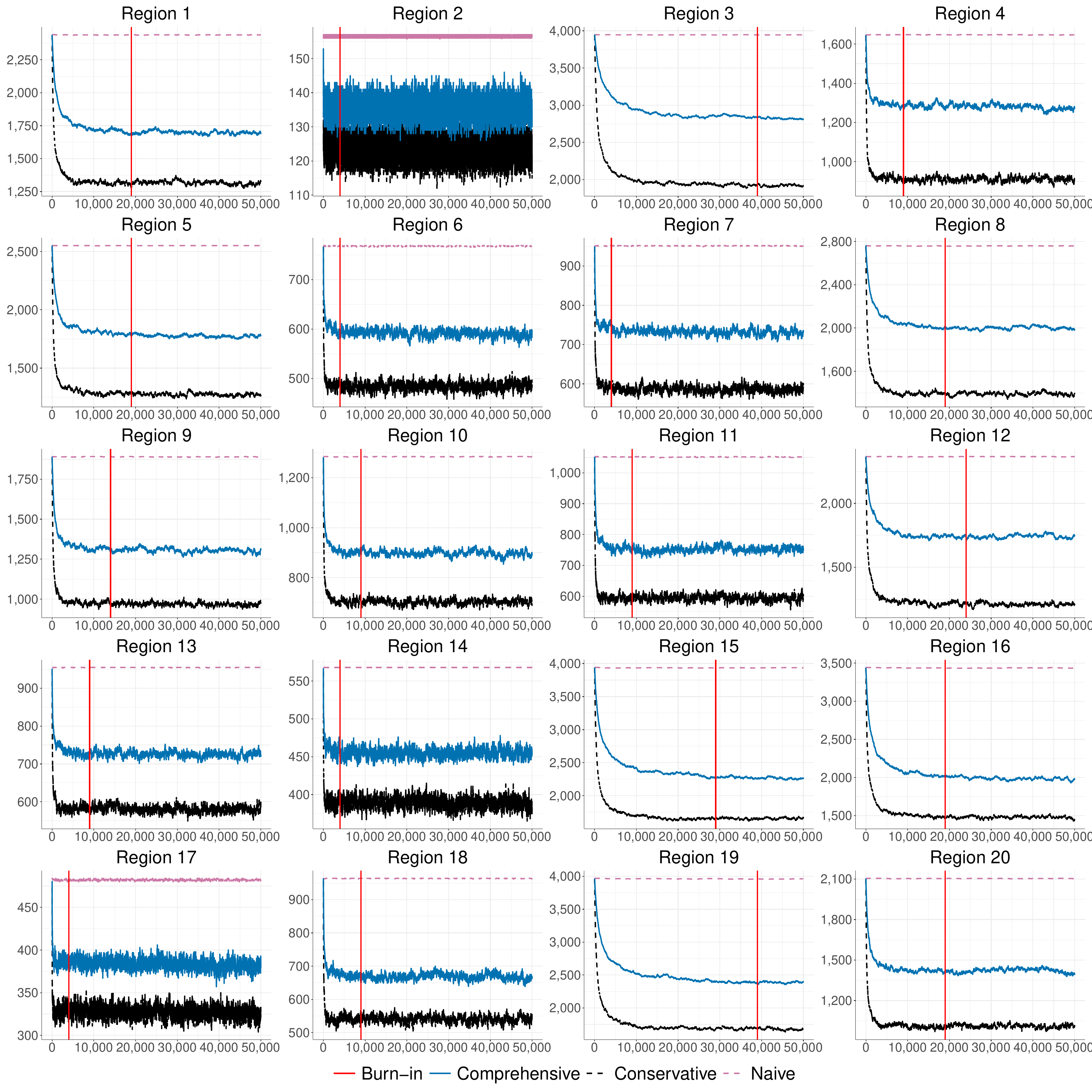}
  \end{center}
  \caption{Trace plots of the number of unique $\lambda_{ij}$ in each region for each approach. Different numbers of burn-in samples were chosen based on the trace plot for each region.}
  \label{fig:ishiw-regional-convergence}
\end{figure}

\begin{table}[!ht]
  \centering\caption{Number of burn-in samples for each region. 1,000 MCMC samples after burn-in were used for inference.}
  \label{tab:ishiw-burn-in-samples}
  \renewcommand{\arraystretch}{0.8}
  \begin{tabular}{clllll}
    \hline
    Region      & \multicolumn{1}{c}{1}  & \multicolumn{1}{c}{2}  & \multicolumn{1}{c}{3}  & \multicolumn{1}{c}{4}  & \multicolumn{1}{c}{5}  \\ \hline
    \#(Burn-in) & 19,000                 & 4,000                  & 39,000                 & 9,000                  & 19,000                 \\ \hline
    Region      & \multicolumn{1}{c}{6}  & \multicolumn{1}{c}{7}  & \multicolumn{1}{c}{8}  & \multicolumn{1}{c}{9}  & \multicolumn{1}{c}{10} \\ \hline
    \#(Burn-in) & 4,000                  & 4,000                  & 19,000                 & 14,000                 & 9,000                  \\ \hline
    Region      & \multicolumn{1}{c}{11} & \multicolumn{1}{c}{12} & \multicolumn{1}{c}{13} & \multicolumn{1}{c}{14} & \multicolumn{1}{c}{15} \\ \hline
    \#(Burn-in) & 9,000                  & 24,000                 & 9,000                  & 4,000                  & 29,000                 \\ \hline
    Region      & \multicolumn{1}{c}{16} & \multicolumn{1}{c}{17} & \multicolumn{1}{c}{18} & \multicolumn{1}{c}{19} & \multicolumn{1}{c}{20} \\ \hline
    \#(Burn-in) & 19,000                 & 4,000                  & 9,000                  & 39,000                 & 19,000                 \\ \hline
  \end{tabular}
\end{table}

\subsection{Supplementary Results}
\label{app:ishiw-results}

The results of the estimated linkage structure using CHOMPER-MCMC and CHOMPER-EVIL are summarized in Table 2 in Section 4.2 in the main manuscript.
In Section 4.2, we showed that the Comprehensive approach generally performs better than alternative approaches.
Using the linkage structure estimate, we can not only determine which entity each record belongs to but also estimate the number of unique entities in the data.
We calculated the number of unique $\lambda_{ij}$ for each iteration of the 1,000 MCMC samples used for inference.
We present an estimated density plot using these 1,000 values (Figure~\ref{fig:ishiw-n-entities-supplementary}).
Except for regions 2 and 6, the estimated density from the Comprehensive approach is closer to the true number of unique entities than that from the Conservative approach.
As shown in the linkage structure estimate results, the Naive approach fails to find links and therefore cannot accurately estimate the number of unique entities.

\begin{figure}[!ht]
  \begin{center}
    \includegraphics[width=0.9\linewidth]{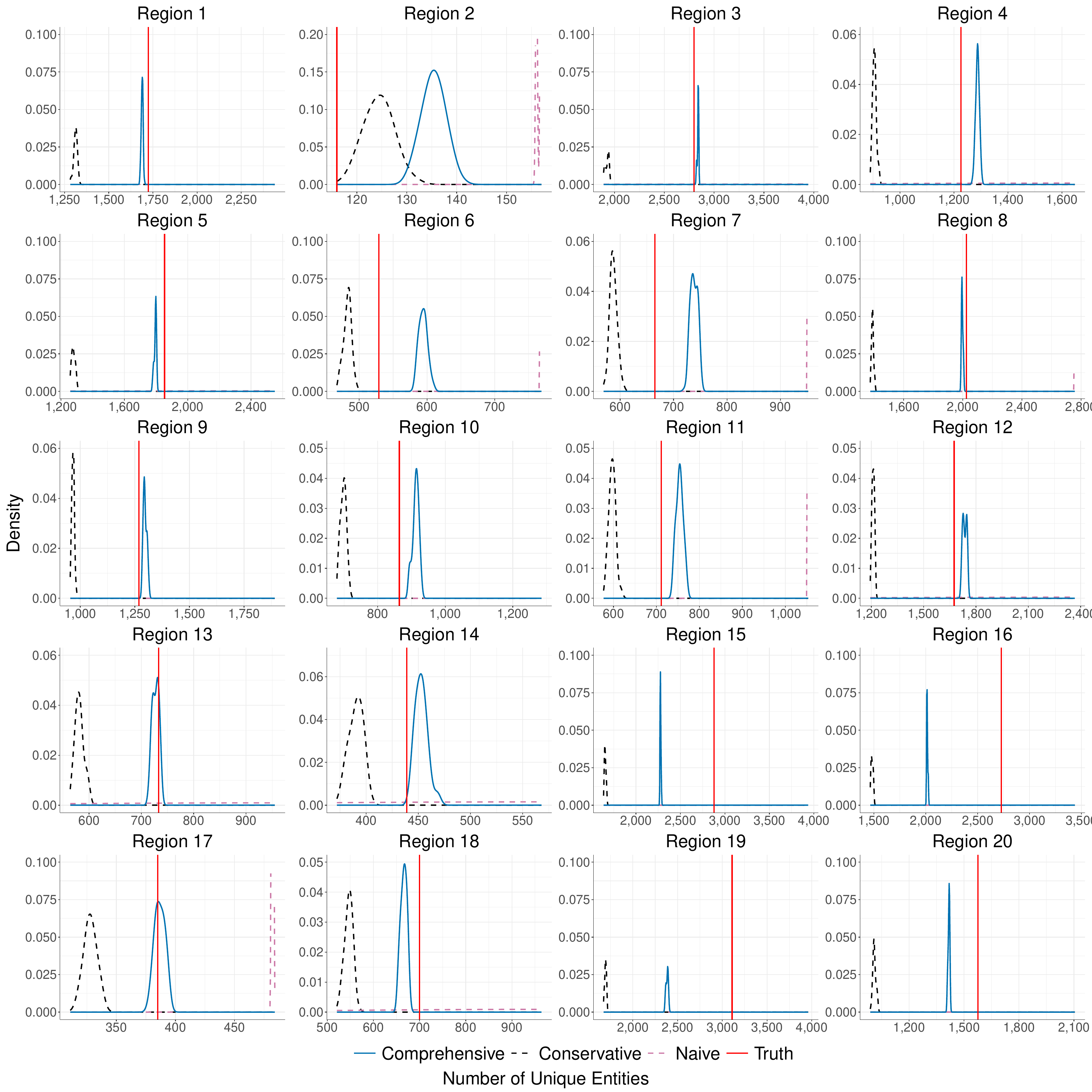}
  \end{center}
  \caption{Estimated densities for the number of unique entities across all regions. The red vertical line represents the true number of unique entities, and the solid blue, dashed black, and dashed purple lines represent the estimated density using MCMC samples from the Comprehensive, Conservative, and Naive approaches, respectively.}
  \label{fig:ishiw-n-entities-supplementary}
\end{figure}

Section 4 of the main manuscript compares the performance of 3 different approaches.
Among these, the Naive approach estimates the linkage structure without considering the potential multiple truths of the variables.
However, in practice, researchers may misinterpret these multiple truths as a distortion.
Therefore, we additionally examined the results obtained from CHOMPER-MCMC, assuming that the variables (education level, net annual income, and estimated price of residence) were highly distorted.
We assumed that 10\% of the total records for these variables contained distortion such that
\begin{equation*}
  \beta_{\ell}=\text{Beta}\left(N_{\max}^{r}\times0.1\times0.1,N_{\max}^{r}\times0.1\right)
\end{equation*}
Table~\ref{tab:ishiw-naive-distorted} presents comparisons between the results from the Naive approach and those from the assumption of highly distorted variables.
Through these comparisons, we confirmed that even when assuming a higher distortion rate for variables, the local variability of the variables cannot be adequately captured without using locally-varying hits.

\begin{table}[H]
  \centering\caption{Results of the linkage structure estimation of regional ISHIW data using CHOMPER-MCMC with the Naive approach across different distortion rates. 1\% Distortion column represents the results from the model with 1\% distortion rate, while 10\% Distortion column represents the results from the model with 10\% distortion rate for the multiple truths variables. The performances are measured with F\textsubscript{1}-score (F\textsubscript{1}), false negative rate (FNR), and false discovery rate (FDR).}
  \label{tab:ishiw-naive-distorted}
  \renewcommand{\arraystretch}{0.8}
  \begin{tabular}{cccccccccccccc}
    \hline
    \multirow{2}{*}{\small Region} & \multicolumn{3}{c}{1\% Distortion} & \multicolumn{3}{c}{10\% Distortion} & \multirow{2}{*}{\small Region} & \multicolumn{3}{c}{1\% Distortion} & \multicolumn{3}{c}{10\% Distortion}                                                                                   \\ \cline{2-7} \cline{9-14}
                                   & F\textsubscript{1}                 & FNR                                 & FDR                            & F\textsubscript{1}                 & FNR                                 & FDR  &    & F\textsubscript{1} & FNR  & FDR  & F\textsubscript{1} & FNR  & FDR  \\ \hline
    1                              & 0.00                               & 1.00                                & 0.00                           & 0.00                               & 1.00                                & 0.50 & 11 & 0.01               & 1.00 & 0.00 & 0.00               & 1.00 & -    \\
    2                              & 0.00                               & 1.00                                & -                              & 0.00                               & 1.00                                & -    & 12 & 0.00               & 1.00 & -    & 0.00               & 1.00 & 1.00 \\
    3                              & 0.01                               & 1.00                                & 0.25                           & 0.00                               & 1.00                                & 0.33 & 13 & 0.00               & 1.00 & -    & 0.00               & 1.00 & -    \\
    4                              & 0.00                               & 1.00                                & -                              & 0.00                               & 1.00                                & 1.00 & 14 & 0.00               & 1.00 & -    & 0.00               & 1.00 & -    \\
    5                              & 0.00                               & 1.00                                & 1.00                           & 0.00                               & 1.00                                & 0.00 & 15 & 0.00               & 1.00 & 0.80 & 0.00               & 1.00 & 1.00 \\
    6                              & 0.01                               & 1.00                                & 0.00                           & 0.01                               & 1.00                                & 0.00 & 16 & 0.01               & 1.00 & 0.78 & 0.00               & 1.00 & 0.83 \\
    7                              & 0.01                               & 1.00                                & 0.00                           & 0.00                               & 1.00                                & -    & 17 & 0.04               & 0.98 & 0.33 & 0.04               & 0.98 & 0.00 \\
    8                              & 0.01                               & 0.99                                & 0.33                           & 0.00                               & 1.00                                & 1.00 & 18 & 0.00               & 1.00 & -    & 0.00               & 1.00 & -    \\
    9                              & 0.00                               & 1.00                                & 0.50                           & 0.01                               & 1.00                                & 0.00 & 19 & 0.00               & 1.00 & 1.00 & 0.00               & 1.00 & 0.90 \\
    10                             & 0.00                               & 1.00                                & 1.00                           & 0.00                               & 1.00                                & -    & 20 & 0.00               & 1.00 & -    & 0.00               & 1.00 & 1.00 \\ \hline
  \end{tabular}
\end{table}

\subsection{Computation Time Comparison between MCMC and EVIL}
\label{app:ishiw-time}

This section compares the run times of CHOMPER-MCMC and CHOMPER-EVIL.
Both models were run on AMD Milan EPYC 7543 32-Core Processors with a 2.8 GHz base clock.
For CHOMPER-MCMC, we measured the time required to generate the region-specific burn-in samples presented in Table~\ref{tab:ishiw-burn-in-samples} and 1,000 additional MCMC samples for the inference.
The elapsed time for CHOMPER-EVIL decreases as more cores are used because parallelization is possible.
The dashed red line represents the elapsed time of CHOMPER-MCMC.
Unlike CHOMPER-EVIL, the elapsed time is not affected by the number of cores.

\begin{figure}[!ht]
  \begin{center}
    \includegraphics[width=0.9\linewidth]{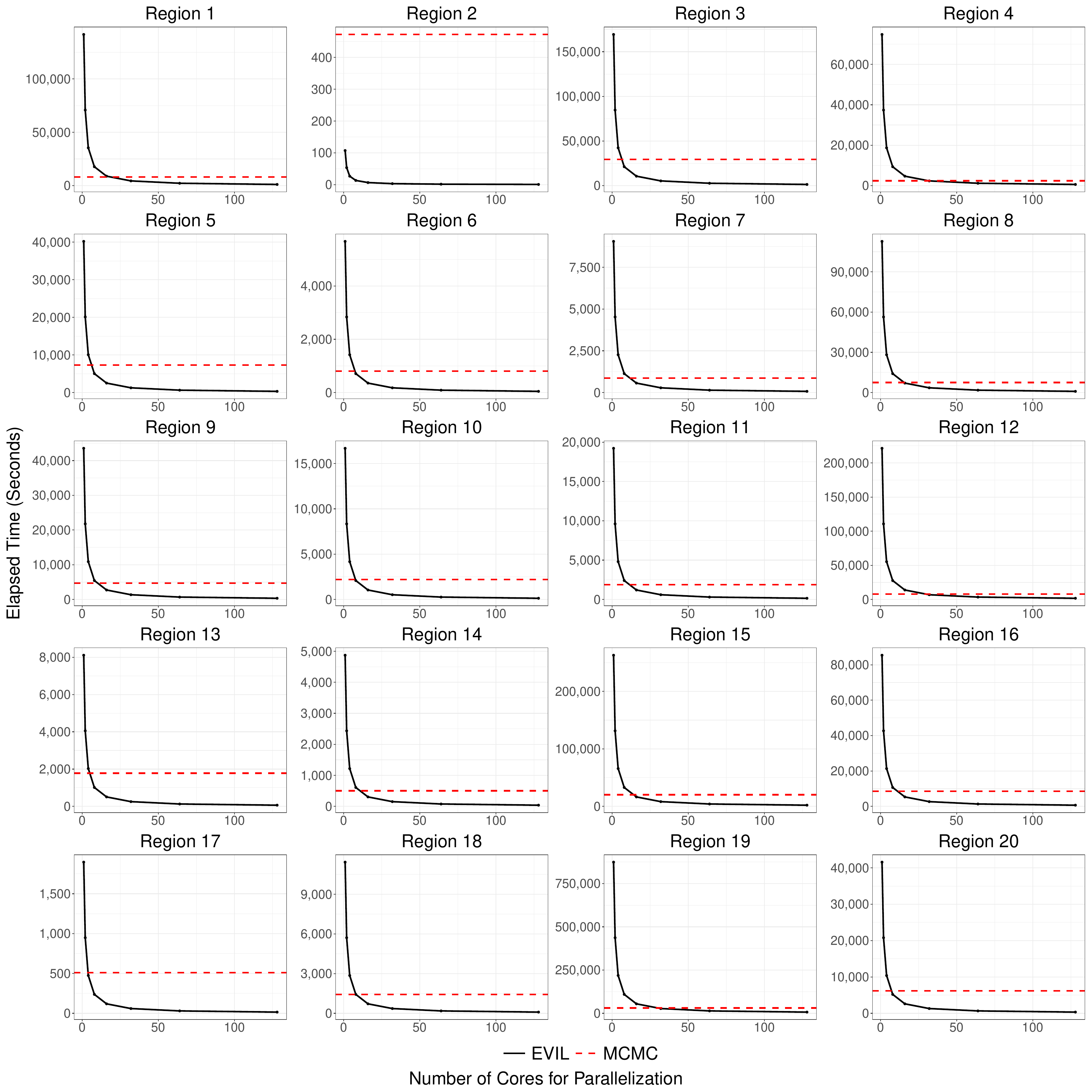}
  \end{center}
  \caption{Elapsed time of running CHOMPER-MCMC and CHOMPER-EVIL for all regions.}
  \label{fig:ishiw-timing-all-regions}
\end{figure}

\section{Details for the Simulation Study}
\label{app:simulation}

In this section, we present the detailed specifications of the models used in the simulation study in Section 5 of the main manuscript, the convergence diagnostics, and the results for all overlap ratios: high, medium, and low.
In addition, we conduct a sensitivity analysis of the hitting range when the attributes follow Gaussian distributions and provide fitted results from the SMERED (Split and Merge Record Linkage and De-duplication; \cite{steorts2016}) model, which uses the traditional hit-or-miss framework, for comparison.

\subsection{Data Generation Process}
\label{app:simulation-data-generation}

We generated 30 sets of synthetic data for each overlap ratio (high: 70\%, medium: 50\%, and low: 30\%) according to the model defined in Equation~\eqref{eqn:chomper-mdl-appendix}.
The records were stored in $k=2$ files, each originating from $N=750$ unique entities.
Each record consisted of 5 multinomial variables with 8 levels, each with equal probability, and 2 Gaussian variables.
We defined the continuous variables to follow Gaussian distributions with means of 0 and 10 and a variance of 10.
However, we normalized the data to have a mean of 0 and a variance of 1 before analysis.
Among these, 2 categorical and 2 continuous variables have multiple truths.
The local variability of the categorical variables was set by randomly sampling 1-adjacent categories to the true value $y_{\lambda_{ij}\ell}$ with equal probability.
The local variability for the continuous variables was applied through random sampling from $N\left(y_{\lambda_{ij}\ell}, 0.5^{2}\right)$.
In addition, we set the distortion rate $\beta_{\ell}$ for each variable to 1\%.

\subsection{Prior and Hyperparameters}
\label{app:simulation-prior}

In addition to the models in Section 4, the simulation study included a model with hyperparameter misspecification in which the researcher chose an over-specified hitting range, referred to as a Vague approach.
The Vague approach sets the hitting range wider than both the Oracle model and the Comprehensive approach, making it difficult for CHOMPER to distinguish between distortion and signal.
Because we showed in Section 5 that CHOMPER is robust to categorical variables with multiple truths, we further evaluate the sensitivity of the hitting range using the Vague approach for Scenario 2, which uses Gaussian fields with multiple truths.
See Table~\ref{tab:simulation-hitting-range} for detailed hyperparameter combinations.

\begin{table}[H]
  \centering\caption{Hitting range parameters according to the different approaches for modeling the simulation data. Conservative and Naive approaches use infinitesimal hitting range, while Comprehensive approach sets a wide hitting range for multiple truths. Vague approach refers to over-specifying the hitting range.}
  \label{tab:simulation-hitting-range}
  \begin{subtable}{\textwidth}
    \centering
    % \captionsetup{font=footnotesize}
    \caption{Scenario 1: Hitting range parameters when only categorical fields are used: 3 categorical variables have a single truth, and 2 categorical variables have multiple truths.}
    \label{tab:hitting-range-scenario-1}
    \scalebox{0.9}{
      \begin{tabular}{ccc}
        \hline
        Approach      & Single Truth                                  & Local Variability                             \\ \hline
        Conservative  & $\varepsilon_{\ell}=0$, $\tau_{\ell}=10^{-2}$ & -                                             \\
        Naive         & $\varepsilon_{\ell}=0$, $\tau_{\ell}=10^{-2}$ & $\varepsilon_{\ell}=0$, $\tau_{\ell}=10^{-2}$ \\
        Oracle        & $\varepsilon_{\ell}=0$, $\tau_{\ell}=10^{-2}$ & $\varepsilon_{\ell}=1$, $\tau_{\ell}=10^{-2}$ \\
        Comprehensive & $\varepsilon_{\ell}=0$, $\tau_{\ell}=10^{-2}$ & $\varepsilon_{\ell}=2$, $\tau_{\ell}=10^{-2}$ \\ \hline
      \end{tabular}
    }
  \end{subtable}
  
  \vspace{\baselineskip}
  
  \begin{subtable}{\textwidth}
    \centering
    % \captionsetup{font=footnotesize}
    \caption{Scenario 2: Hitting range parameters when both categorical and continuous fields are used: 3 categorical variables have a single truth, and 2 continuous variables have multiple truths.}
    \label{tab:hitting-range-scenario-2}
    \scalebox{0.9}{
      \begin{tabular}{ccc}
        \hline
        Approach      & Single Truth                                  & Local Variability            \\ \hline
        Conservative  & $\varepsilon_{\ell}=0$, $\tau_{\ell}=10^{-2}$ & -                            \\
        Naive         & $\varepsilon_{\ell}=0$, $\tau_{\ell}=10^{-2}$ & $\varepsilon_{\ell}=10^{-3}$ \\
        Oracle        & $\varepsilon_{\ell}=0$, $\tau_{\ell}=10^{-2}$ & $\varepsilon_{\ell}=0.1$     \\
        Comprehensive & $\varepsilon_{\ell}=0$, $\tau_{\ell}=10^{-2}$ & $\varepsilon_{\ell}=0.5$     \\
        Vague         & $\varepsilon_{\ell}=0$, $\tau_{\ell}=10^{-2}$ & $\varepsilon_{\ell}=1.0$     \\ \hline
      \end{tabular}
    }
  \end{subtable}
\end{table}

We used the same prior, hyperparameters, and initialization strategy for all models as Appendix~\ref{app:ishiw-prior}.
Specifically, we set the prior distortion rate to 1\% for all variables, and the hyperparameters for the probability of each level in categorical fields to a length-$M_{\ell}$ one vector, and $\text{Inverse-Gamma}\left(0.01,0.01\right)$ for the variance of continuous fields.
For CHOMPER-MCMC, 1,000 split-merge updates were performed for each Gibbs sampling iteration, and for CHOMPER-EVIL, $N_{P}=50$, $N_{O}=100$, and $N_{E}=50$ were chosen.

\subsection{Convergence Diagnostics}
\label{app:simulation-convergence}

Similar to Appendix~\ref{app:ishiw-convergence}, we determine convergence of CHOMPER-MCMC using trace plots of the number of unique $\lambda_{ij}$.
The convergence of CHOMPER-EVIL is determined by the relative change of the ELBO between generations.
For all scenarios, approaches, and overlap ratios, we considered the first 9,000 samples as burn-in samples, and the remaining 1,000 samples were used for inference.
Trace plots for the diagnostics are provided in Figures~\ref{fig:simulation-traceplot-high-overlap},~\ref{fig:simulation-traceplot-medium-overlap}, and~\ref{fig:simulation-traceplot-low-overlap}.

\begin{figure}[H]
  \begin{center}
    \subfloat[Scenario 1: 3 multinomial fields with a single truth and 2 multinomial fields with multiple truths.\label{fig:simulation-traceplot-scenario-1-high}]{\includegraphics[width=0.9\linewidth]{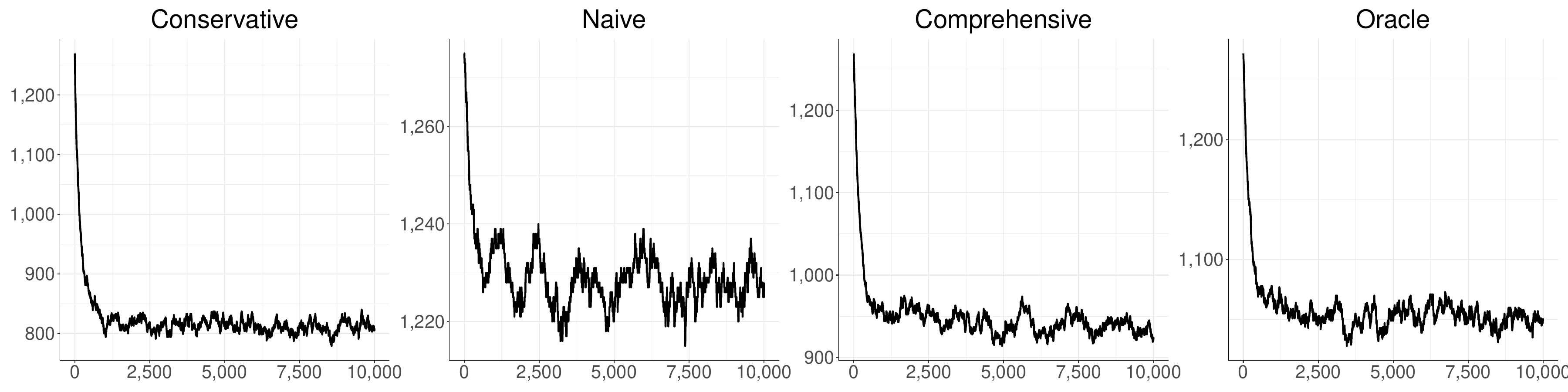} }
    \subfloat[Scenario 2: 3 multinomial fields with a single truth and 2 Gaussian fields with multiple truths.\label{fig:simulation-traceplot-scenario-2-high}]{\includegraphics[width=0.9\linewidth]{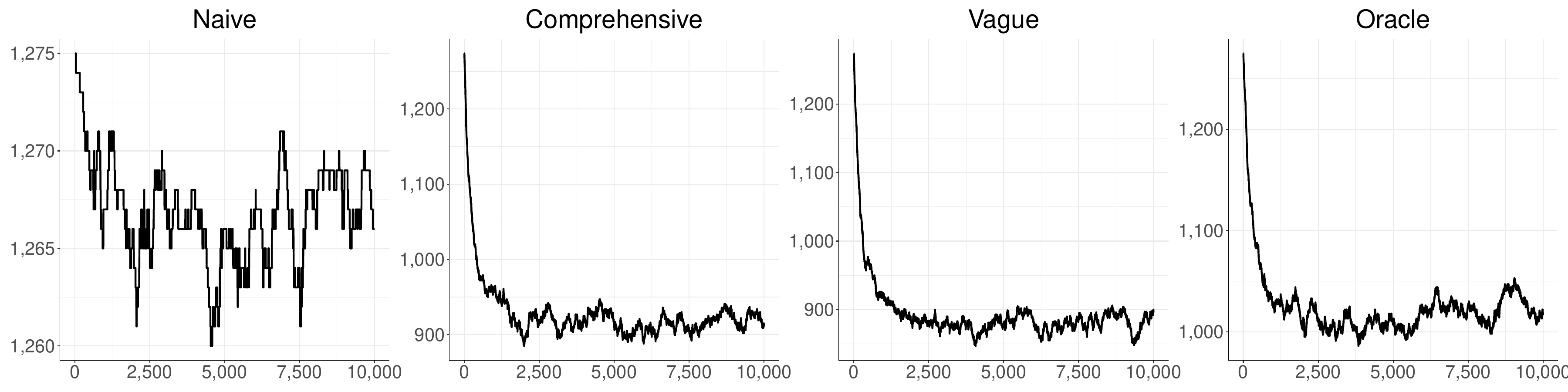} }
  \end{center}
  \caption{Trace plots of the number of unique $\lambda_{ij}$ according to the approaches for Scenario 1 and 2 with high overlap ratio of 70\%.}
  \label{fig:simulation-traceplot-high-overlap}
\end{figure}

\begin{figure}[H]
  \begin{center}
    \subfloat[Scenario 1: 3 multinomial fields with a single truth and 2 multinomial fields with multiple truths.\label{fig:simulation-traceplot-scenario-1-medium}]{\includegraphics[width=0.9\linewidth]{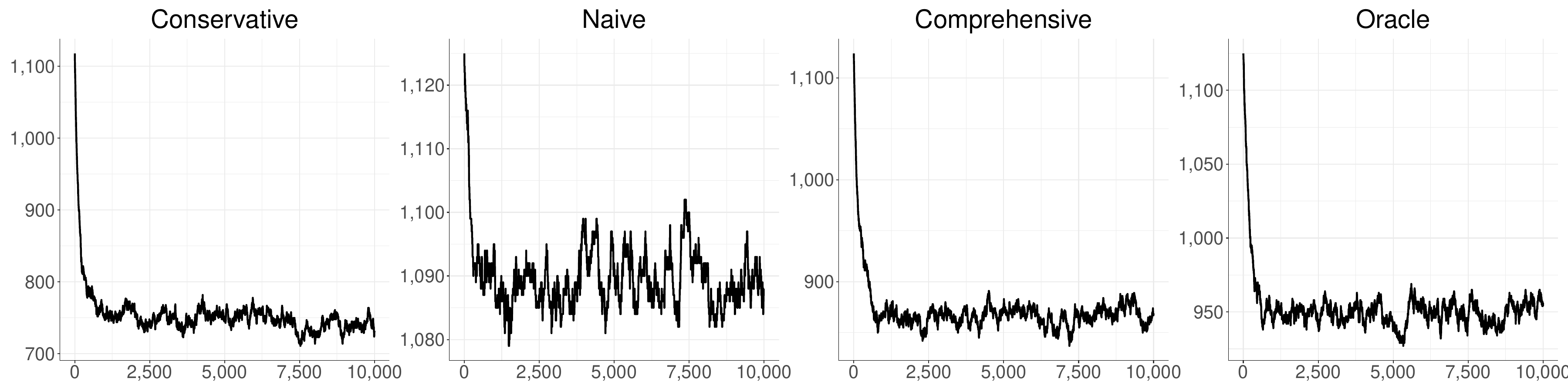} }
    \subfloat[Scenario 2: 3 multinomial fields with a single truth and 2 Gaussian fields with multiple truths.\label{fig:simulation-traceplot-scenario-2-medium}]{\includegraphics[width=0.9\linewidth]{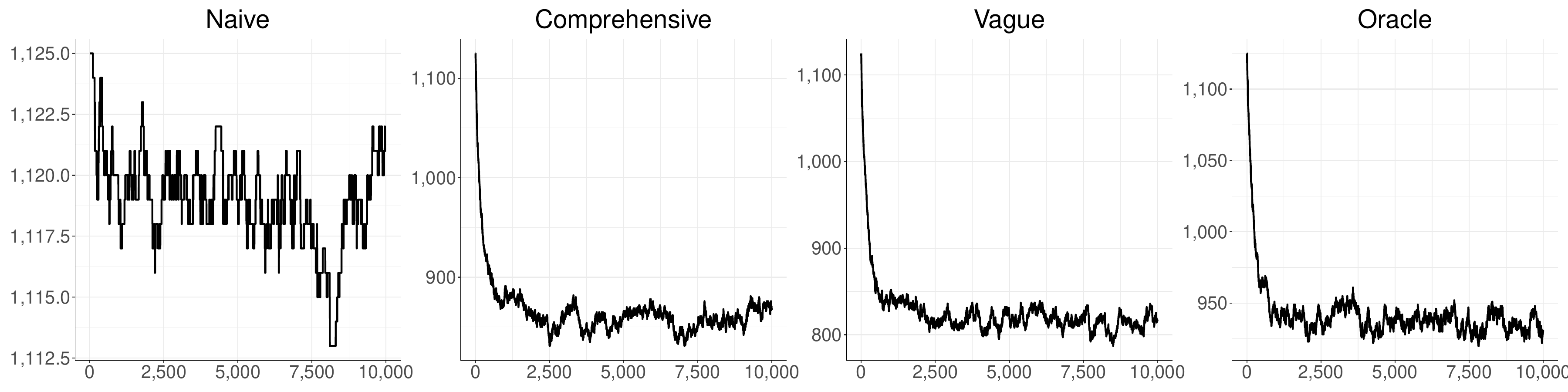} }
  \end{center}
  \caption{Trace plots of the number of unique $\lambda_{ij}$ according to the approaches for Scenario 1 and 2 with medium overlap ratio of 50\%.}
  \label{fig:simulation-traceplot-medium-overlap}
\end{figure}

\begin{figure}[H]
  \begin{center}
    \subfloat[Scenario 1: 3 multinomial fields with a single truth and 2 multinomial fields with multiple truths.\label{fig:simulation-traceplot-scenario-1-low}]{\includegraphics[width=0.9\linewidth]{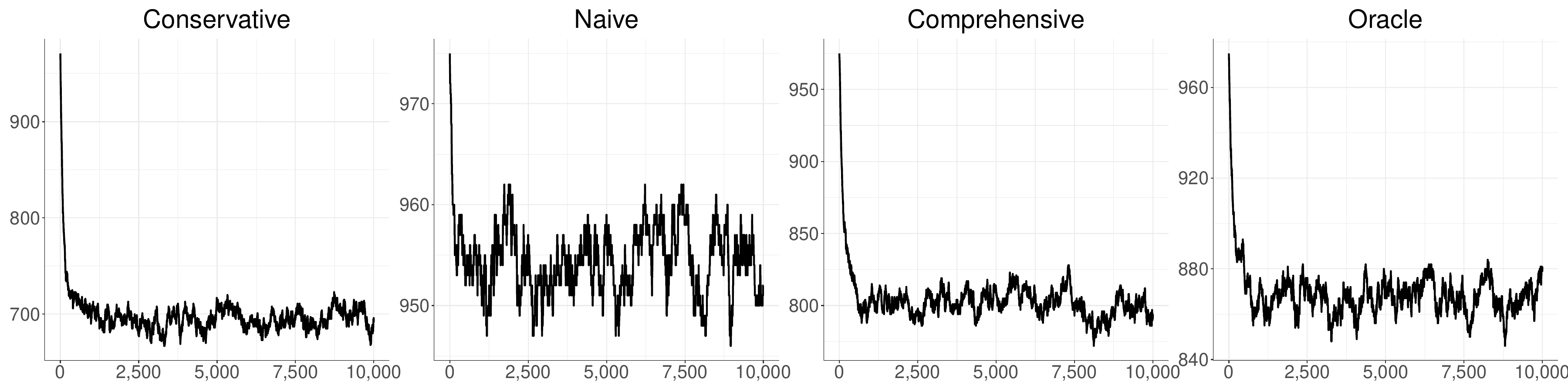} }
    \subfloat[Scenario 2: 3 multinomial fields with a single truth and 2 Gaussian fields with multiple truths.\label{fig:simulation-traceplot-scenario-2-low}]{\includegraphics[width=0.9\linewidth]{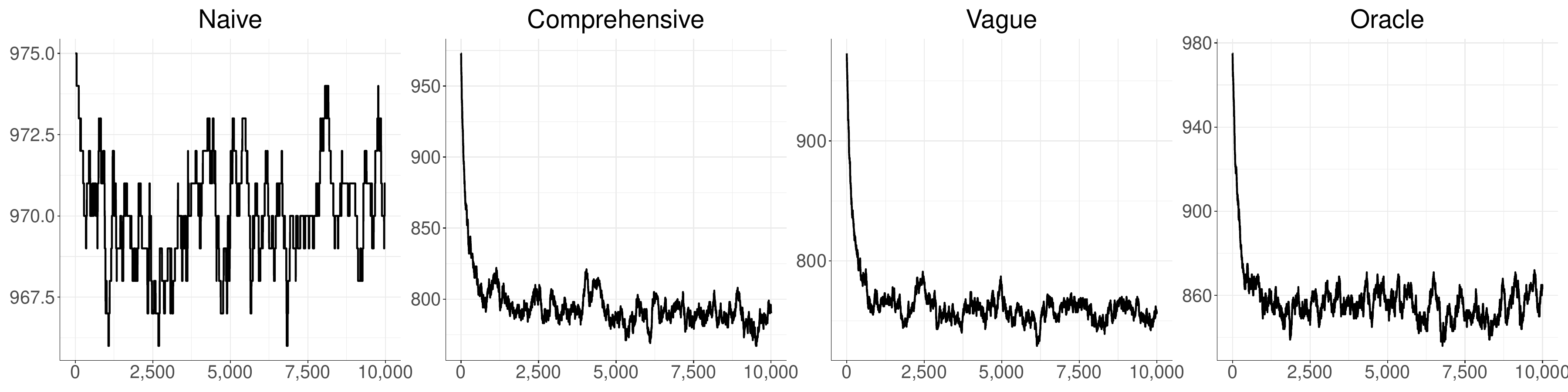} }
  \end{center}
  \caption{Trace plots of the number of unique $\lambda_{ij}$ according to the approaches for Scenario 1 and 2 with low overlap ratio of 30\%.}
  \label{fig:simulation-traceplot-low-overlap}
\end{figure}

\subsection{Detailed Results and Supplementary Experiments}
\label{app:simulation-results}

This section provides the results for all scenarios and approaches defined in Appendix~\ref{app:simulation-prior}.
For the sensitivity analysis, we provide the results for different overlap ratios, including high as 70\%, medium as 50\%, and low as 30\%.
Moreover, we provide the results of fitting the SMERED model, which uses the traditional hit-or-miss framework for comparison.
Note that SMERED was used only for the Conservative and Naive approaches for Scenario 1, because it is not applicable for continuous fields.

%High overlap ratio
\begin{table}[H]
  \centering
  \caption{The average (standard deviation) of the linkage structure estimation performances of 30 simulation datasets with high overlap ratio (70\%) with the CHOMPER model using either MCMC or EVIL. The performances are measured with F\textsubscript{1}-score (F\textsubscript{1}), false negative rate (FNR), and false discovery rate (FDR).}
  \label{tab:simulation-results-high-overlap}
  
  \begin{subtable}{\textwidth}
    \centering
    % \captionsetup{font=footnotesize}
    \caption{Scenario 1 - CHOMPER: Estimation performance when only categorical fields are used.}
    \label{tab:simulation-scenario-1-high-overlap}
    \scalebox{0.85}{
      \begin{tabular}{ccccccc}
        \hline
        \multirow{2}{*}{Approach} & \multicolumn{3}{c}{MCMC} & \multicolumn{3}{c}{EVIL}                                                                \\ \cline{2-7}
                                  & F\textsubscript{1}       & FNR                      & FDR         & F\textsubscript{1} & FNR         & FDR         \\ \hline
        Conservative              & 0.30 (0.02)              & 0.72 (0.02)              & 0.69 (0.03) & 0.40 (0.02)        & 0.51 (0.04) & 0.66 (0.04) \\
        Naive                     & 0.14 (0.02)              & 0.93 (0.01)              & 0.16 (0.06) & 0.22 (0.02)        & 0.87 (0.01) & 0.26 (0.05) \\
        Comprehensive             & 0.52 (0.02)              & 0.57 (0.02)              & 0.36 (0.03) & 0.63 (0.03)        & 0.40 (0.04) & 0.33 (0.09) \\
        Oracle                    & 0.47 (0.02)              & 0.66 (0.02)              & 0.25 (0.03) & 0.56 (0.05)        & 0.57 (0.05) & 0.19 (0.04) \\ \hline
      \end{tabular}
    }
  \end{subtable}
  
  \vspace{\baselineskip}
  
  \begin{subtable}{\textwidth}
    \centering
    % \captionsetup{font=footnotesize}
    \caption{Scenario 1 - SMERED: Estimation performance when only categorical fields are used.}
    \label{tab:simulation-scenario-smered-high-overlap}
    \scalebox{0.85}{
      \begin{tabular}{cccc}
        \hline
        Approach     & F\textsubscript{1} & FNR         & FDR         \\ \hline
        Conservative & 0.30 (0.02)        & 0.72 (0.02) & 0.69 (0.02) \\
        Naive        & 0.14 (0.02)        & 0.93 (0.01) & 0.17 (0.06) \\ \hline
      \end{tabular}
    }
  \end{subtable}
  
  \vspace{\baselineskip}
  
  \begin{subtable}{\textwidth}
    \centering
    % \captionsetup{font=footnotesize}
    \caption{Scenario 2: Estimation performance when both categorical and continuous fields are used.}
    \label{tab:simulation-scenario-2-high-overlap}
    \scalebox{0.85}{
      \begin{tabular}{ccccccc}
        \hline
        \multirow{2}{*}{Approach} & \multicolumn{3}{c}{MCMC} & \multicolumn{3}{c}{EVIL}                                                                \\ \cline{2-7}
                                  & F\textsubscript{1}       & FNR                      & FDR         & F\textsubscript{1} & FNR         & FDR         \\ \hline
        Conservative              & 0.30 (0.02)              & 0.72 (0.02)              & 0.69 (0.03) & 0.40 (0.02)        & 0.51 (0.04) & 0.66 (0.04) \\
        Naive                     & 0.03 (0.01)              & 0.99 (0.00)              & 0.02 (0.04) & 0.08 (0.03)        & 0.96 (0.01) & 0.10 (0.07) \\
        Comprehensive             & 0.59 (0.02)              & 0.50 (0.02)              & 0.29 (0.02) & 0.48 (0.02)        & 0.12 (0.03) & 0.67 (0.02) \\
        Vague                     & 0.53 (0.02)              & 0.52 (0.02)              & 0.40 (0.03) & 0.44 (0.02)        & 0.13 (0.02) & 0.70 (0.01) \\
        Oracle                    & 0.57 (0.02)              & 0.57 (0.02)              & 0.12 (0.02) & 0.74 (0.02)        & 0.16 (0.04) & 0.33 (0.03) \\ \hline
      \end{tabular}
    }
  \end{subtable}
\end{table}

%Medium overlap ratio
\begin{table}[H]
  \centering
  \caption{The average (standard deviation) of the linkage structure estimation performances of 30 simulation datasets with medium overlap ratio (50\%) with the CHOMPER model using either MCMC or EVIL. The performances are measured with F\textsubscript{1}-score (F\textsubscript{1}), false negative rate (FNR), and false discovery rate (FDR).}
  \label{tab:simulation-results-medium-overlap}
  
  \begin{subtable}{\textwidth}
    \centering
    % \captionsetup{font=footnotesize}
    \caption{Scenario 1 - CHOMPER: Estimation performance when only categorical fields are used.}
    \label{tab:simulation-scenario-1-medium-overlap}
    \scalebox{0.85}{
      \begin{tabular}{ccccccc}
        \hline
        \multirow{2}{*}{Approach} & \multicolumn{3}{c}{MCMC} & \multicolumn{3}{c}{EVIL}                                                                \\ \cline{2-7}
                                  & F\textsubscript{1}       & FNR                      & FDR         & F\textsubscript{1} & FNR         & FDR         \\ \hline
        Conservative              & 0.29 (0.02)              & 0.72 (0.02)              & 0.71 (0.02) & 0.36 (0.01)        & 0.47 (0.04) & 0.72 (0.02) \\
        Naive                     & 0.13 (0.02)              & 0.93 (0.01)              & 0.19 (0.06) & 0.22 (0.02)        & 0.87 (0.02) & 0.28 (0.06) \\
        Comprehensive             & 0.51 (0.02)              & 0.56 (0.02)              & 0.38 (0.03) & 0.63 (0.03)        & 0.41 (0.04) & 0.31 (0.09) \\
        Oracle                    & 0.47 (0.03)              & 0.66 (0.03)              & 0.26 (0.04) & 0.52 (0.03)        & 0.61 (0.03) & 0.23 (0.03) \\ \hline
      \end{tabular}
    }
  \end{subtable}
  
  \vspace{\baselineskip}
  
  \begin{subtable}{\textwidth}
    \centering
    % \captionsetup{font=footnotesize}
    \caption{Scenario 1 - SMERED: Estimation performance when only categorical fields are used.}
    \label{tab:simulation-scenario-smered-medium-overlap}
    \scalebox{0.85}{
      \begin{tabular}{cccc}
        \hline
        Approach     & F\textsubscript{1} & FNR         & FDR         \\ \hline
        Conservative & 0.29 (0.03)        & 0.72 (0.03) & 0.71 (0.03) \\
        Naive        & 0.13 (0.02)        & 0.93 (0.01) & 0.17 (0.06) \\ \hline
      \end{tabular}
    }
  \end{subtable}
  
  \vspace{\baselineskip}
  
  \begin{subtable}{\textwidth}
    \centering
    % \captionsetup{font=footnotesize}
    \caption{Scenario 2: Estimation performance when both categorical and continuous fields are used.}
    \label{tab:simulation-scenario-2-medium-overlap}
    \scalebox{0.85}{
      \begin{tabular}{ccccccc}
        \hline
        \multirow{2}{*}{Approach} & \multicolumn{3}{c}{MCMC} & \multicolumn{3}{c}{EVIL}                                                                \\ \cline{2-7}
                                  & F\textsubscript{1}       & FNR                      & FDR         & F\textsubscript{1} & FNR         & FDR         \\ \hline
        Conservative              & 0.29 (0.02)              & 0.72 (0.02)              & 0.71 (0.02) & 0.36 (0.01)        & 0.47 (0.04) & 0.72 (0.02) \\
        Naive                     & 0.03 (0.01)              & 0.99 (0.01)              & 0.06 (0.12) & 0.06 (0.02)        & 0.97 (0.01) & 0.06 (0.07) \\
        Comprehensive             & 0.59 (0.02)              & 0.49 (0.02)              & 0.30 (0.03) & 0.45 (0.02)        & 0.11 (0.03) & 0.70 (0.01) \\
        Vague                     & 0.53 (0.03)              & 0.52 (0.03)              & 0.41 (0.03) & 0.42 (0.01)        & 0.12 (0.03) & 0.72 (0.01) \\
        Oracle                    & 0.57 (0.03)              & 0.57 (0.03)              & 0.12 (0.03) & 0.74 (0.02)        & 0.14 (0.03) & 0.36 (0.02) \\ \hline
      \end{tabular}
    }
  \end{subtable}
\end{table}

%Low overlap ratio
\begin{table}[H]
  \centering
  \caption{The average (standard deviation) of the linkage structure estimation performances of 30 simulation datasets with low overlap ratio (30\%) with the CHOMPER model using either MCMC or EVIL. The performances are measured with F\textsubscript{1}-score (F\textsubscript{1}), false negative rate (FNR), and false discovery rate (FDR).}
  \label{tab:simulation-results-low-overlap}
  
  \begin{subtable}{\textwidth}
    \centering
    % \captionsetup{font=footnotesize}
    \caption{Scenario 1 - CHOMPER: Estimation performance when only categorical fields are used.}
    \label{tab:simulation-scenario-1-low-overlap}
    \scalebox{0.85}{
      \begin{tabular}{ccccccc}
        \hline
        \multirow{2}{*}{Approach} & \multicolumn{3}{c}{MCMC} & \multicolumn{3}{c}{EVIL}                                                                \\ \cline{2-7}
                                  & F\textsubscript{1}       & FNR                      & FDR         & F\textsubscript{1} & FNR         & FDR         \\ \hline
        Conservative              & 0.27 (0.02)              & 0.71 (0.03)              & 0.75 (0.02) & 0.31 (0.01)        & 0.46 (0.04) & 0.78 (0.01) \\
        Naive                     & 0.15 (0.03)              & 0.92 (0.02)              & 0.21 (0.08) & 0.23 (0.02)        & 0.86 (0.02) & 0.32 (0.07) \\
        Comprehensive             & 0.50 (0.03)              & 0.55 (0.04)              & 0.43 (0.04) & 0.57 (0.03)        & 0.43 (0.04) & 0.41 (0.03) \\
        Oracle                    & 0.46 (0.03)              & 0.66 (0.03)              & 0.31 (0.04) & 0.50 (0.03)        & 0.60 (0.03) & 0.33 (0.03) \\ \hline
      \end{tabular}
    }
  \end{subtable}
  
  \vspace{\baselineskip}
  
  \begin{subtable}{\textwidth}
    \centering
    % \captionsetup{font=footnotesize}
    \caption{Scenario 1 - SMERED: Estimation performance when only categorical fields are used.}
    \label{tab:simulation-scenario-smered-low-overlap}
    \scalebox{0.85}{
      \begin{tabular}{cccc}
        \hline
        Approach     & F\textsubscript{1} & FNR         & FDR         \\ \hline
        Conservative & 0.26 (0.02)        & 0.72 (0.03) & 0.75 (0.02) \\
        Naive        & 0.14 (0.02)        & 0.92 (0.01) & 0.21 (0.06) \\ \hline
      \end{tabular}
    }
  \end{subtable}
  
  \vspace{\baselineskip}
  
  \begin{subtable}{\textwidth}
    \centering
    % \captionsetup{font=footnotesize}
    \caption{Scenario 2: Estimation performance when both categorical and continuous fields are used.}
    \label{tab:simulation-scenario-2-low-overlap}
    \scalebox{0.85}{
      \begin{tabular}{ccccccc}
        \hline
        \multirow{2}{*}{Approach} & \multicolumn{3}{c}{MCMC} & \multicolumn{3}{c}{EVIL}                                                                \\ \cline{2-7}
                                  & F\textsubscript{1}       & FNR                      & FDR         & F\textsubscript{1} & FNR         & FDR         \\ \hline
        Conservative              & 0.27 (0.02)              & 0.71 (0.03)              & 0.75 (0.02) & 0.31 (0.01)        & 0.46 (0.04) & 0.78 (0.01) \\
        Naive                     & 0.02 (0.01)              & 0.99 (0.01)              & 0.04 (0.10) & 0.07 (0.02)        & 0.97 (0.01) & 0.14 (0.10) \\
        Comprehensive             & 0.58 (0.03)              & 0.48 (0.03)              & 0.34 (0.03) & 0.39 (0.01)        & 0.10 (0.02) & 0.75 (0.01) \\
        Vague                     & 0.52 (0.03)              & 0.50 (0.04)              & 0.45 (0.03) & 0.36 (0.01)        & 0.10 (0.02) & 0.78 (0.01) \\
        Oracle                    & 0.57 (0.04)              & 0.57 (0.04)              & 0.14 (0.03) & 0.69 (0.02)        & 0.09 (0.04) & 0.45 (0.03) \\ \hline
      \end{tabular}
    }
  \end{subtable}
\end{table}

As shown in Tables~\ref{tab:simulation-results-high-overlap},~\ref{tab:simulation-results-medium-overlap}, and~\ref{tab:simulation-results-low-overlap}, the performance of CHOMPER-MCMC and SMERED was similar regardless of the overlap ratio.
This empirically demonstrates that the CHOMPER model encompasses probabilistic entity resolution models that use that traditional hit-or-miss framework.

The simulation results show that attribute type and hitting range misspecification affects estimation performance.
For categorical attributes with multiple truths, high robustness to hitting range was observed in both CHOMPER-MCMC and CHOMPER-EVIL, regardless of the overlap ratio.
In contrast, when using continuous attributes, a clear difference in sensitivity between the inferential methods was observed.
CHOMPER-MCMC showed similar performance even with moderate misspecification, whereas CHOMPER-EVIL's performance decreased sharply.
This trend remained consistent across all overlap ratios.
Nevertheless, CHOMPER-EVIL may be advantageous when using continuous fields with high prior knowledge about the hitting range, as CHOMPER-EVIL achieved the highest average F\textsubscript{1}-score in Scenario 2 with the Oracle approach, irrespective of the overlap ratio.

\end{document}